\documentclass[a4paper,11pt]{article}
\usepackage{a4wide}
\usepackage[UKenglish]{babel}
\usepackage[noadjust]{cite}
\usepackage{amsmath}
\usepackage{amssymb}
\usepackage{amsthm} 
\usepackage{bm}
\usepackage{bbm}
\usepackage{mathtools}
\usepackage{mathrsfs}
\usepackage{csquotes}
\usepackage{todonotes}
\usepackage{tikz}
\usepackage{tikz-3dplot}
\usetikzlibrary{perspective}
\usepackage{graphicx,wrapfig}
\usepackage{bbm}
\usepackage{tcolorbox}

\DeclareSymbolFont{extraup}{U}{zavm}{m}{n}
\DeclareMathSymbol{\varheart}{\mathalpha}{extraup}{86}
\DeclareMathSymbol{\vardiamond}{\mathalpha}{extraup}{87}

\usepackage{hyperref} 
\hypersetup{colorlinks=true,citecolor=blue,linkcolor=blue,urlcolor=blue}

\newcommand{\YES}{\textsc{Yes}}
\newcommand{\NO}{\textsc{No}}
\newcommand{\Test}[2]{
\def\temp{#2}\ifx\temp\empty
  \operatorname{Test}_{#1}
\else
  \operatorname{Test}_{#1}^{#2}
\fi
}

\newcommand{\A}{\mathbf{A}}
\newcommand{\B}{\mathbf{B}}

\newcommand{\K}{\mathbf{K}}
\newcommand{\T}{\mathbf{T}}
\newcommand{\X}{\mathbf{X}}

\newcommand{\N}{\mathbb{N}}
\newcommand{\R}{\mathbb{R}}
\newcommand{\Q}{\mathbb{Q}}
\newcommand{\Z}{\mathbb{Z}}

\newcommand{\cT}{\mathcal{T}}
\newcommand{\bF}{\mathbb{F}}
\newcommand{\sM}{\mathscr{M}}
\newcommand{\freeM}{\mathbb{F}_{\Mminion}}
\newcommand{\freeN}{\mathbb{F}_{\Nminion}}
\newcommand{\freeQ}{\mathbb{F}_{\Qconv}}
\newcommand{\freeZ}{\mathbb{F}_{\Zaff}}
\newcommand{\freeS}{\mathbb{F}_{\Sminion}}
\newcommand{\Xk}{\X^\tensor{k}}
\newcommand{\Ak}{\A^\tensor{k}}
\newcommand{\Bk}{\B^\tensor{k}}
\newcommand{\inn}[2]{\mathbf{#1}\textbf{-in-}\mathbf{#2}}
\newcommand{\NAE}{\mathbf{NAE}}
\renewcommand{\vec}[1]{\mathbf{#1}}
\newcommand{\ba}{\vec{a}}
\newcommand{\bb}{\vec{b}}
\newcommand{\bc}{\vec{c}}
\newcommand{\bd}{\vec{d}}
\newcommand{\bi}{\vec{i}}
\newcommand{\bj}{\vec{j}}
\newcommand{\bn}{\vec{n}}

\newcommand{\bs}{\vec{s}}
\newcommand{\bu}{\vec{u}}
\newcommand{\bx}{\vec{x}}
\newcommand{\bv}{\vec{v}}
\newcommand{\by}{\vec{y}}
\newcommand{\bw}{\vec{w}}
\newcommand{\be}{\vec{e}}
\newcommand{\bell}{{\ensuremath{\boldsymbol\ell}}}

\newcommand{\bp}{\vec{p}}
\newcommand{\bq}{\vec{q}}
\newcommand{\bz}{\vec{z}}

\DeclareMathOperator{\Bmat}{B}
\DeclareMathOperator{\Amat}{A}

\newcommand{\blambda}{{\bm{\lambda}}}
\newcommand{\bmu}{{\bm{\mu}}}
\newcommand{\ale}{{\aleph_0}}

\newcommand{\pair}[2]{[\begin{array}{cc}#1 & #2\end{array}]}

\DeclareMathOperator{\tr}{tr}
\DeclareMathOperator{\Span}{span}

\DeclareMathOperator{\csupp}{csupp}
\DeclareMathOperator{\ArcC}{AC}
\DeclareMathOperator{\BLP}{BLP}
\DeclareMathOperator{\AIP}{AIP}
\DeclareMathOperator{\BA}{BA}
\DeclareMathOperator{\SDP}{SDP}
\DeclareMathOperator{\Pol}{Pol}
\DeclareMathOperator{\PCSP}{PCSP}
\DeclareMathOperator{\CSP}{CSP}
\DeclareMathOperator{\parPCSP}{(P)CSP}
\DeclareMathOperator{\PCSPs}{PCSPs}

\DeclareMathOperator{\parPCSPs}{(P)CSPs}
\DeclareMathOperator{\CLAP}{CLAP}
\DeclareMathOperator{\id}{id}
\DeclareMathOperator{\Hom}{\operatorname{Hom}}
\DeclareMathOperator{\supp}{supp}
\DeclareMathOperator{\SA}{SA}
\DeclareMathOperator{\SoS}{SoS}
\DeclareMathOperator{\BW}{BW}
\DeclareMathOperator{\dom}{dom}
\DeclareMathOperator{\ar}{ar}
\DeclareMathOperator{\armax}{armax}

\newcommand{\freeBA}{\mathbb{F}_{\BAminion}}
\newcommand{\freeMN}{\mathbb{F}_{\MNminion}}

\newcommand{\Mminion}{\ensuremath{{\mathscr{M}}}}
\newcommand{\Nminion}{\ensuremath{{\mathscr{N}}}}
\newcommand{\MNminion}{\Mminion\ltimes\Nminion}
\newcommand{\BAminion}{\ensuremath{{\Qconv\ltimes\Zaff}}}

\newcommand{\Qconv}{\ensuremath{{\mathscr{Q}_{\operatorname{conv}}}}}
\newcommand{\Zaff}{\ensuremath{{\mathscr{Z}_{\operatorname{aff}}}}}

\newcommand{\Hminion}{\ensuremath{\mathscr{H}}}
\newcommand{\Sminion}{\ensuremath{{\mathscr{S}}}}
\newcommand{\Mblpaip}{\ensuremath{\mathscr{M}_{\operatorname{BA}}}}
\newcommand{\bone}{\mathbf{1}}  
\newcommand{\bzero}{\mathbf{0}} 
\newcommand{\tensor}[1]{\textsuperscript{\raisebox{-.5pt}{\normalfont\textcircled{\raisebox{-.1pt}{\tiny #1}}}}}
\newcommand\cont[1]{\overset{\tiny#1}{\ast}}

\theoremstyle{plain}
\newtheorem{thm}{Theorem}
\newtheorem*{thm*}{Theorem}
\newtheorem{lem}[thm]{Lemma}
\newtheorem*{lem*}{Lemma}
\newtheorem{prop}[thm]{Proposition}
\newtheorem*{prop*}{Proposition}

\theoremstyle{definition}
\newtheorem{defn}[thm]{Definition}
\newtheorem{rem}[thm]{Remark}
\newtheorem{example}[thm]{Example}

\pgfmathsetmacro{\opacityDefault}{.05}
\pgfmathsetmacro{\bigShift}{4.5}
\newcommand\cube[5]{
\draw[fill=blue,opacity=#2,shift={(-#3,#5,-#4)},scale=.7](0,0,0)--++(0,-1,0)--++(0,0,-1)--++(0,1,0)--cycle;
\draw[fill=blue,opacity=#2,shift={(-#3,#5,-#4)},scale=.7](0,0,0)--++(-1,0,0)--++(0,0,-1)--++(1,0,0)--cycle;
\draw[fill=blue,opacity=#2,shift={(-#3,#5,-#4)},scale=.7](0,0,0)--++(-1,0,0)--++(0,-1,0)--++(1,0,0)--cycle;
}
\definecolor{lightgray}{RGB}{230,230,230}	
\newenvironment{graytbox}{	
\par\addvspace{0.1cm}	
\begin{tcolorbox}[width=\textwidth,	
                  boxsep=5pt,	
                  left=1pt,	
                  right=1pt,	
                  top=2pt,	
                  bottom=2pt,	
                  boxrule=0pt,	
                  arc=0pt,	
                  colback=lightgray,	
                  colframe=black
                  ]
}{	
\end{tcolorbox}	
}

\begin{document}

\title{	
Hierarchies of Minion Tests for PCSPs through Tensors\thanks{An extended
abstract of this work  appeared in the Proceedings of the 2023 ACM-SIAM Symposium on Discrete Algorithms (SODA'23)~\cite{cz23soda:minions}. 
The research leading to these results has received funding from the European Research Council (ERC) under the European Union's Horizon 2020 research and innovation programme (grant agreement No 714532). The paper reflects only the authors' views and not the views of the ERC or the European Commission. The European Union is not liable for any use that may be made of the information contained therein.
This work was also supported by UKRI EP/X024431/1. The work was done while the first author was at the University of Oxford. For the purpose of Open Access, the authors have applied a CC BY public copyright licence to any Author Accepted Manuscript version arising from this submission. All data is provided in full in the results section of this paper.}}	

\author{Lorenzo Ciardo\\
TU Graz\\
\texttt{lorenzo.ciardo@tugraz.at} 
\and 
Stanislav {\v{Z}}ivn{\'y}\\
University of Oxford\\
\texttt{standa.zivny@cs.ox.ac.uk}
}

\date{}
\maketitle

\begin{abstract}
\noindent
We provide a unified framework to study hierarchies of relaxations for Constraint Satisfaction Problems and their Promise variant. The idea is to split the description of a hierarchy into an algebraic part, depending on a \emph{minion} capturing the ``base level'', and a geometric part -- which we call \emph{tensorisation} -- inspired by multilinear algebra. 
We exploit the structure of the tensor spaces arising from our construction to prove general properties of hierarchies. We identify certain classes of minions, which we call \emph{linear} and \emph{conic}, whose corresponding hierarchies have particularly fine features.
We establish that the (combinatorial) bounded-width, Sherali--Adams LP, affine IP, Sum-of-Squares SDP, and combined ``LP + affine IP'' hierarchies are all captured by this framework.
In particular, in order to analyse the Sum-of-Squares SDP hierarchy, we also characterise the solvability of the standard SDP relaxation through a new minion. 
\end{abstract}

\section{Introduction}
\label{sec_intro}

What are the limits of efficient algorithms and where is the precise borderline
of tractability? The \emph{constraint satisfaction problem} ($\CSP$) offers a
general framework for studying such fundamental questions for a large class of
computational
problems~\cite{Creignouetal:siam01,Creignou08:complexity,kz17:dagstuhl} but yet
for a class that is amenable to identifying the mathematical structure governing
tractability. Canonical examples of $\CSP$s are satisfiability of 3-CNF formulas
(\textsc{3-Sat}), 
``not-all-equal'' satisfiability of 3-CNF formulas (\textsc{3-Nae-Sat}),
linear equations, several variants of (hyper)graph colourings, and the graph clique problem. All $\CSP$s can be seen as homomorphism problems between relational structures~\cite{Feder98:monotone}: Given two relational structures $\X$ and $\A$, is there a homomorphism from $\X$ to $\A$?
Intuitively, the  structure  $\X$ represents the variables of the $\CSP$ instance and their interactions, whereas the  structure $\A$  represents the constraint language; i.e., the alphabet and the allowed constraint relations.

The most studied types of $\CSP$s are so-called \emph{non-uniform}
$\CSP$s~\cite{Jeavons97:closure,Feder98:monotone,Kolaitis00:jcss,BKW17}, in
which the target
structure $\A$ is fixed whereas the source structure $\X$ is given on input;
this computational problem is denoted by $\CSP(\A)$. From the examples above,
\textsc{3-Sat}, \textsc{3-Nae-Sat}, (hyper)graph colourings with constantly many
colours, linear equations of bounded arity over finite fields, and linear
equations of bounded arity over the rationals are all examples of non-uniform
$\CSP$s, all on finite domains except the last one.
For instance, in the graph $c$-colouring problem the target structure $\A$ is a $c$-clique and the structure $\X$ is the input graph. The existence of a homomorphism from a graph to a $c$-clique is equivalent to the existence of a colouring of the graph with $c$ colours.
The graph clique problem is an example of a $\CSP$ with 
a fixed class of source structures~\cite{Grohe07:jacm,Marx13:jacm} but an arbitrary target structure and, thus, it is not a non-uniform $\CSP$.

We will be concerned with polynomial-time tractability of $\CSP$s. Studied research directions include investigating questions such as: Is there a
solution~\cite{Bulatov17:focs,Zhuk20:jacm}? How many solutions are there,
exactly~\cite{Creignou96:ic,Bulatov13:jacm-dichotomy,Dyer13:sicomp} or approximately~\cite{Bulatov13:jacm,Chen15:jcss}?
What is the maximum number of simultaneously satisfied constraints,
exactly~\cite{Creignou95:jcss,Huber14:sicomp,tz16:jacm-complexity} or approximately~\cite{Deineko08:jacm,Austrin10:sicomp,Raghavendra08:everycsp}? What is
the minimum number of simultaneously unsatisfied
constraints~\cite{Khanna00:approximability,Dalmau18:jcss}? Given an almost
satisfiable instance, can one find a somewhat satisfying
solution~\cite{Dalmau13:toct,Barto16:sicomp,Dalmau19:sicomp}? 
In this paper, we will focus on the following question:

\medskip

\noindent \emph{Given a satisfiable instance, can one find a solution that is satisfying in a weaker sense}~\cite{AGH17,BBKO21,BG21:sicomp}? 

\medskip

\noindent This was formalised as \emph{promise constraint satisfaction problems}
($\PCSP$s) by Austrin, Guruswami and H{\aa}stad~\cite{AGH17} and Brakensiek and
Guruswami~\cite{BG21:sicomp}. Let $\A$ and $\B$ be two fixed relational
structures\footnote{Unless otherwise stated, we shall use the word ``structure'' to mean finite-domain structures; if the domain is allowed to be infinite, we shall say it explicitly.} such that there is a homomorphism from $\A$ to $\B$, indicated by
$\A\to\B$.
Intuitively, the structure $\A$ represents the ``strict'' constraints and the structure $\B$ represents the corresponding ``weak'' constraints.
An instance of the $\PCSP$ over the template $(\A,\B)$, denoted by
$\PCSP(\A,\B)$, is a relational structure $\X$ such that there is a homomorphism
from $\X$ to $\A$. The task is to find a homomorphism from $\X$ to $\B$, which
exists since homomorphisms compose. 
What we described above is the \emph{search} variant of the PCSP. In the
\emph{decision} variant, one is given a relational structure $\X$ and the task is
to decide whether there is a homomorphism from $\X$ to $\A$ or whether there is
not a homomorphism from $\X$ to $\B$. As $\X\to\A$ implies $\X\to\B$, the two cases cannot happen simultaneously.
Clearly, the decision variant of the PCSP reduces to the search
variant (see the discussion in~\cite{BBKO21} after Definition~2.6), but it is not known whether there is a reduction in the other
direction for all PCSPs. 
In this paper, we shall focus on the decision variant.

$\PCSP$s are a vast generalisation of $\CSP$s including problems that cannot be
expressed as $\CSP$s.
The work of Barto, Bul\'in, Krokhin, and Opr\v{s}al~\cite{BBKO21} lifted
and greatly extended the algebraic framework developed for
$\CSP$s~\cite{Jeavons97:closure,Bulatov05:classifying,BOP18} to the realm of $\PCSP$s.
Subsequently, there has been a series of recent works on the computational
complexity of
$\PCSP$s building on~\cite{BBKO21}, including applicability of local consistency
and convex
relaxations~\cite{BG19,bgwz20,Butti21:mfcs,cz23sicomp:clap,Atserias22:soda} and
complexity of fragments of
$\PCSP$s~\cite{GS20:icalp,KOWZ22,AB21,Barto21:stacs,BWZ21,BG21:talg,Barto22:soda,NZ22}.
Strong results on $\PCSP$s have also been established via
other techniques than those in~\cite{BBKO21}, mostly analytical methods,
e.g., hardness of various (hyper)graph colourings~\cite{Khot01,DRS05,Huang13,ABP20} and  other $\PCSP$s~\cite{Bhangale21:stoc,Braverman21:itcs,BGS23,Bhangale22:stoc}.

An example of a $\PCSP$, identified in~\cite{AGH17}, is (in the search variant)
finding a satisfying assignment to a $k$-CNF formula given that a $g$-satisfying
assignment exists; i.e., an assignment that satisfies at least $g$ literals in
each clause. Austrin et al.~\cite{AGH17} established that this problem is NP-hard if $g/k<1/2$ and solvable via a constant level of the Sherali--Adams linear programming relaxation otherwise. This classification was later extended to problems over arbitrary finite domains by Brandts et al.~\cite{BWZ21}.

A second example of a $\PCSP$, identified in~\cite{BG21:sicomp}, is (in the
search variant) finding a ``not-all-equal'' assignment to a monotone $3$-CNF formula given that a ``$1$-in-$3$'' assignment is promised to exist; i.e., given a $3$-CNF formula with positive literals only and the promise that an assignment exists that satisfies exactly one literal in each clause, the task is to find an assignment that satisfies one or two literals in each clause.
This problem is solvable in polynomial time via a constant level of the Sherali--Adams linear programming relaxation~\cite{BG21:sicomp} but not via a reduction to finite-domain $\CSP$s~\cite{BBKO21}.

A third example of a $\PCSP$ is the well-known \emph{approximate graph
colouring} problem: Given a $c$-colourable graph, find a $d$-colouring of it,
for constants $c$ and $d$ with $c\leq d$. 
This corresponds to $\PCSP(\K_c,\K_d)$, where $\K_p$ is the
clique on $p$ vertices. Despite a long
history dating back to 1976~\cite{GJ76}, the complexity of this problem is
only understood under stronger
assumptions~\cite{Dinur09:sicomp,GS20:icalp,Braverman21:focs}
and for special cases~\cite{KhannaLS00,Khot01,GK04,Huang13,BrakensiekG16,BBKO21,KOWZ22}.
It is believed that the problem is NP-hard already in the decision
variant~\cite{GJ76}, i.e., deciding whether a graph is $c$-colourable or not
even $d$-colourable, for each $3\leq c\leq d$. By using the framework developed
in the current work, non-solvability of approximate graph colouring through
standard algorithmic techniques was established in the follow-up
works~\cite{CZ25sicomp:approximate}.

\medskip
$\PCSP$s can be solved by designing \emph{tests}. If
a test, applied to a given instance of the problem, is positive then the answer
is $\YES$; if it is negative then the answer is $\NO$. 
The challenge is then to find tests that are able to guarantee a low number -- ideally, zero -- of false positives and false negatives.
Clearly, a test is itself a decision problem. However, its nature may be substantially different, and less complicated, than the nature of the original problem.

Given a $\PCSP$ template $(\A,\B)$, we may use any (potentially infinite)
structure $\T$ to make a test for $\PCSP(\A,\B)$: We simply let the outcome of
the test on an instance structure $\X$ be $\YES$ if $\X\to\T$, and $\NO$ if
$\X\not\to\T$. In other words, $\CSP(\T)$ is a test for $\PCSP(\A,\B)$. Let $\X$
be an instance of $\PCSP(\A,\B)$. If $\X\to\T$ whenever $\X\to\A$, the test is
guaranteed not to generate false negatives, and we call it \emph{complete}.
If $\X\to\B$ whenever $\X\to\T$, the test is guaranteed not to generate false
positives, and we call it \emph{sound}.
Since homomorphisms compose, the test is complete if and only if $\A\to\T$, and it is sound if and only if $\T\to\B$.\footnote{The second ``only if'' implication follows from a standard compactness argument, see Section~\ref{sec_minion_tests}.}
When both of these conditions hold,
we say that the test \emph{solves} $\PCSP(\A,\B)$. In this case, one obtains a so-called ``sandwich reduction'' from $\PCSP(\A,\B)$ to $\CSP(\T)$ (see~\cite[Section~3.1]{BG19}).

To make a test $\T$ useful as a polynomial-time algorithm to solve a $\PCSP$, one requires that 
$\CSP(\T)$ should be tractable. It was conjectured in~\cite[Section~3.1]{BG19} that every tractable
(finite-domain) $\PCSP$ is solved by a tractable sandwich. In other words, if the
conjecture is true, \emph{sandwich reductions are the sole source of tractability for PCSPs}.
For the conjecture to be true, one needs to admit structures $\T$ having infinite domains. For example, this is the case for the ``1-in-3 vs. not-all-equal'' problem, whose template we denote by $(\inn{1}{3},\NAE)$:
As
shown in~\cite[Section~8]{BBKO21}, there is
no (finite) structure $\T$ such that $\inn{1}{3}\to\T\to\NAE$ and $\CSP(\T)$ is tractable. 

The complexity of both $\CSP$s and $\PCSP$s was shown to be determined by
higher-order symmetries of the solution sets of the problems, known as
\emph{polymorphisms}, denoted by $\Pol(\A)$ for
$\CSP(\A)$~\cite{Bulatov05:classifying} and by $\Pol(\A,\B)$ for
$\PCSP(\A,\B)$~\cite{BG21:sicomp,BBKO21}. For $\CSP$s, polymorphisms form \emph{clones}; in
particular, they are closed under composition. This means that some symmetries
may be obtainable through compositions of other symmetries, so that one can hope
to capture properties of entire families of $\CSP$s (e.g., bounded width,
tractability, etc.) through the presence of a certain polymorphism and, more
generally, to describe their complexity through universal-algebraic tools. A
chief example of this approach is the positive resolution of the
dichotomy conjecture for $\CSP$s by Bulatov~\cite{Bulatov17:focs} and
Zhuk~\cite{Zhuk20:jacm}, establishing that finite-domain non-uniform $\CSP$s are either in P or are NP-complete.
For $\PCSP$s, however, polymorphisms are not closed under composition, and the
algebraic structure they are endowed with -- known as a \emph{minion} -- is much
less rich than clones. In particular, the structure theory of finite algebras, a central topic
of universal algebra, is not applicable in this setting.

For a PCSP template $(\A,\B)$, one would ideally aim to build tests for
$\PCSP(\A,\B)$ in a systematic way. One method to do so is by considering tests
associated with minions and, in particular, their \emph{free structures}. The
free structure $\freeM(\A)$ of a minion $\Mminion$ generated by a structure
$\A$~\cite[Definition~4.1]{BBKO21} is a 
(potentially infinite) 
structure obtained, essentially,
by simulating the relations in $\A$ on a domain consisting of elements of
$\Mminion$. Then, we define $\Test{\Mminion}{}(\X,\A)=\YES$ if
$\X\to\freeM(\A)$, and $\NO$ otherwise. (Note that $\X$ is the input to the
problem; the minion $\Mminion$ and the relational structure $\A$, coming from a
$\PCSP$ template, are (fixed) parameters of the test.)

For certain choices of $\Mminion$, $\Test{\Mminion}{}$ is \textit{always} a
tractable test; i.e., $\CSP(\freeM(\A))$ is tractable for any $\A$. 
This is the case for the minions $\mathscr{H}=\Pol(\textsc{Horn-3-Sat})$ (whose
elements are nonempty subsets of a given set), $\Qconv$ (whose elements are
stochastic vectors), and $\Zaff$ (whose elements are affine integers vectors).
As it was shown in~\cite{BBKO21}, these three minions correspond to three
well-studied algorithmic relaxations: $\Test{\Hminion}{}$ is Arc Consistency
($\ArcC$)~\cite{Mackowrth77:aij}, $\Test{\Qconv}{}$ is the Basic Linear
Programming relaxation ($\BLP$)~\cite{Kun12:itcs}, and $\Test{\Zaff}{}$ is the
Affine Integer Programming relaxation ($\Zaff$)~\cite{BG19}. In~\cite{bgwz20},
the algorithm $\BLP+\AIP$ (which we shall call $\BA$ in this work) corresponding to a combination of linear and integer
programming was shown to be captured by a certain minion $\Mblpaip$. 
In summary, several widely used algorithms for $\parPCSPs$ are
minion tests; in particular, Arc Consistency, which is the simplest example of
consistency algorithms, and standard algorithms based on relaxations.

Convex relaxations have been instrumental in the understanding of the
complexity of many variants of $\CSP$s, including constant
approximability of Min-$\CSP$s~\cite{Ene13:soda,Dalmau18:jcss} and
Max-$\CSP$s~\cite{Khot07:sicomp,Raghavendra08:everycsp}, robust satisfiability
of $\CSP$s~\cite{Zwick98:stoc,Kun12:itcs,Barto16:sicomp}, and exact solvability
of optimisation $\CSP$s~\cite{ktz15:sicomp,tz17:sicomp}. 
An important line of work focused on making convex relaxations stronger and
stronger via the so-called ``lift-and-project'' method,
which includes the Sherali--Adams LP hierarchy~\cite{Sherali1990}, the SDP hierarchy of
Lov\'asz and Schrijver~\cite{Lovasz91:jo}, and the (stronger) SDP hierarchy of
Lasserre~\cite{Lasserre02}, also known as the Sum-of-Squares hierarchy
(see~\cite{Laurent03} for a comparison of these hierarchies).
The study of the power of various convex hierarchies has led to several breakthroughs,
e.g.,~\cite{Arora09:jacm,Chan16:jacm-lp,Kothari22:sicomp,Ghosh18:toc,Tulsiani09:stoc,Lee15:stoc}.

In the same spirit as lift-and-project hierarchies of convex relaxations, the (combinatorial)
$k$-consistency algorithm (also known as the $k$-bounded-width algorithm) has a central role in the study of tractability for constraint satisfaction problems~\cite{Feder98:monotone,Atserias07:icalp}. Here $k$
is an integer bounding the number of variables considered in reasoning about
partial solutions; the case $k=1$ corresponds to Arc Consistency mentioned
above. The notion of local consistency, in addition to being one of the key
concepts in constraint satisfaction, has also emerged independently in finite
model theory~\cite{Kolaitis95:jcss}, graph theory~\cite{Hell1996duality}, and
proof complexity~\cite{Atserias08:jcss}. The power of local consistency for (non-promise) CSPs
is now fully understood~\cite{Bulatov09:width,Barto14:jacm,Barto14:jloc}. 
Recent works identified necessary conditions on local consistency to solve PCSPs~\cite{Atserias22:soda,ciardo2024periodic,chan2025random}.
Very recent work established the importance of infinite-domain sandwiches for promise CSPs~\cite{Mottet25:lics,Pinsker25:mfcs}.

\paragraph{Contributions}
The main contribution of this work is the introduction of \emph{a general
framework for refining algorithmic relaxations of $\parPCSPs$}. Given a minion
$\Mminion$, we present a technique to systematically turn $\Test{\Mminion}{}$
into the corresponding \emph{hierarchy of minion tests}: a sequence of
increasingly tight relaxations $\Test{\Mminion}{k}$ for $k\in\N$.

The technique we adopt to build hierarchies of minion tests is inspired by multilinear algebra. We describe a  \emph{tensorisation} construction that turns a given structure $\X$ into a structure $\Xk$ on a different signature, where both the domain and the relations are multidimensional objects living in tensor spaces.
Essentially, $\Test{\Mminion}{k}$ works by applying $\Test{\Mminion}{}$ to \emph{tensorised} versions of the structures $\X$ and $\A$ rather than to $\X$ and $\A$ themselves.  
This allows us to study the functioning of the algorithms in the hierarchy by describing the structure of a space of tensors -- which can be accomplished by using multilinear algebra.
This approach has not appeared in the literature on
Sherali--Adams, bounded width, Sum-of-Squares, hierarchies of integer
programming, and related algorithmic techniques such as the high-dimensional
Weisfeiler--Leman algorithm~\cite{AtseriasM13,Butti21:mfcs}. Butti and
Dalmau~\cite{Butti21:mfcs} recently characterised for $\CSP$s when the $k$-th
level of the Sherali--Adams LP programming hierarchy accepts in terms of a
construction different from the one introduced in this work. Unlike the
tensorisation, the construction considered in~\cite{Butti21:mfcs} yields a
relational structure whose domain includes the set of constraints of the
original structure.

One key feature of our framework is that it is \emph{modular}, in that
it allows splitting the description of a hierarchy of minion tests into an
\emph{algebraic} part, corresponding to the minion $\Mminion$, and a
\emph{geometric} part, entirely dependent on the tensorisation construction and
hence common to all hierarchies. By considering certain well-behaved families of
minions, which we call \emph{linear} and \emph{conic}, we can then deduce
general properties of the corresponding hierarchies by only describing the structure of tensor spaces.

Letting
the minion $\Mminion$ be $\Hminion$ (resp., $\Qconv$, $\Zaff$), we shall retrieve in this
way the bounded-width
hierarchy (resp., the Sherali--Adams LP hierarchy, the affine integer
programming hierarchy). Additionally, we describe 
a new minion $\Sminion$
capturing the power of the basic semidefinite programming relaxation
($\SDP$),\footnote{We point out that Brakensiek, Guruswami,
and Sandeep~\cite{bgs_robust23stoc} independently provided a minion characterisation for the power of $\SDP$
that is essentially identical to the one we obtain in the current work.
We also note that the complexity of testing SDP feasibility is an open
problem, cf.~\cite{Gartner2012approximation}.} and we prove that $\Test{\Sminion}{k}$ coincides with the
Sum-of-Squares hierarchy. As a consequence, our framework is
able to provide a unified description of all these four well-known
hierarchies of algorithmic relaxations.
In addition to
casting \emph{known} hierarchies of relaxations as hierarchies of minion tests,
this approach can be used to design \emph{new} hierarchies. In particular, we describe an operation that we call \emph{semi-direct product} of minions, which consists in combining multiple minions to form a new minion associated with a stronger relaxation. In practice, this method can be used to design an algorithm that combines the features of different known algorithmic techniques. We show that the minion $\Mblpaip$  associated with the $\BA$ relaxation from~\cite{bgwz20} is the semi-direct product of $\Qconv$ and $\Zaff$, and we formally introduce the $\BA^k$ hierarchy as the hierarchy $\Test{\Mblpaip}{k}$.

The scope of this framework is potentially not limited to
constraint satisfaction: The multilinear pattern that we found at the core of different algorithmic hierarchies
appears to be transversal to the constraint satisfaction setting and, instead, inherently
connected to the algorithmic techniques themselves, which can be applied to
classes of computational problems living beyond the realms of
($\operatorname{P}$)$\CSP$s.

\paragraph{Subsequent work}
The tensorisation methodology introduced in this paper has later been used by the authors in
follow-up work on the applicability of relaxation hierarchies to specific problems. In
particular, it has been used to prove that the approximate graph colouring problem is not
solved by the hierarchy for the combined basic linear programming and affine
integer programming relaxation~\cite{CZ25sicomp:approximate}.
Recently, Dalmau and Opr\v{s}al~\cite{Dalmau24:lics} showed that various concepts introduced in this work -- in particular, the concept of conic minions -- are naturally connected to the description of a certain type of reductions between (P)CSPs, see~\cite[Section~4.2]{Dalmau24:lics} as well as Remark~\ref{rem:do23} in the current paper.
\paragraph{Organisation}	
The remaining part of the paper is organised as follows. 
Section~\ref{sec_prelims} contains relevant terminology for tensors and gives a formal description of promise CSPs as well as the relaxations and hierarchies used throughout the paper (see also Appendix~\ref{appendix_notes_relaxations_hierarchies}). Section~\ref{sec_minion_tests} contains some initial, basic results on minion tests. Section~\ref{sec_minion_for_SDP} introduces a minion $\Sminion$ capturing the power of $\SDP$. Sections~\ref{sec_hierarchies_of_minion_tests},~\ref{sec_hierarchies_of_linear_minion_tests}, and~\ref{sec_hierarchies_of_conic_minion_tests} are the technical core of the paper; they provide a description of hierarchies of tests built on arbitrary minions, linear minions, and conic minions, respectively. Section~\ref{sec_semi_direct_product} describes the semi-direct product of minions, needed to capture the $\BA^k$ hierarchy. The machinery assembled in the previous sections is finally used in Section~\ref{sec_proof_of_main_thm} to prove that five well-known algorithmic hierarchies for (P)CSPs are captured by our framework.

\section{Preliminaries}
\label{sec_prelims}

\paragraph{Notation}
We denote by $\N$ the set of positive integers. For $k\in\N$, we denote by $[k]$ the set $\{1,\ldots,k\}$. We indicate by
$\be_i$ the $i$-th standard unit vector of the appropriate size (which will
be clear from the context); i.e., the $i$-th entry of $\be_i$ is $1$, and all other entries are $0$. $\bzero_p$ and $\bone_p$ denote the all-zero and all-one vector of size $p$, respectively, while $I_p$ and $O_{p,q}$ denote the $p\times p$ identity matrix and the $p\times q$ all-zero matrix, respectively.
Given a matrix $M$, we let $\tr(M)$ and $\csupp(M)$ be the trace and the set of indices of nonzero columns of $M$, respectively. The symbol $\ale$ denotes the cardinality of $\N$. Given a set $S$, we denote  the identity map on $S$ by $\operatorname{id}_S$ (or simply by $\operatorname{id}$ when $S$ is clear from the context).

\subsection{Terminology for tensors}
\label{sec_tensors_tuples}

\paragraph{Tuples}
Given a set $S$, two integers $k,\ell\in\N$, a tuple $\bs=(s_1,\dots,s_k)\in S^k$, and a tuple $\bi=(i_1,\dots,i_\ell)\in [k]^\ell$, $\bs_\bi$ shall denote the \emph{projection}
of $\bs$ onto $\bi$, i.e., the tuple in $S^\ell$ defined by $\bs_\bi= (s_{i_1},\dots,s_{i_\ell})$. Given two tuples $\bs=(s_1,\dots,s_k)\in S^k$ and $\tilde{\bs}=(\tilde{s}_1,\dots,\tilde{s}_\ell)\in S^\ell$, their \emph{concatenation} is the tuple $(\bs,\tilde{\bs})= (s_1,\dots,s_k,\tilde{s}_1,\dots,\tilde{s}_\ell)\in S^{k+\ell}$. We also define $\{\bs\}= \{s_1,\dots,s_k\}$.
Given two sets $S,\tilde S$ and two tuples $\bs=(s_1,\dots,s_k)\in S^k$, $\tilde{\bs}=(\tilde{s}_1,\dots,\tilde{s}_k)\in \tilde{S}^k$, we write $\bs\prec\tilde{\bs}$ if, for any $\alpha,\beta\in [k]$, $s_\alpha=s_\beta$ implies $\tilde{s}_\alpha=\tilde{s}_\beta$. The expression $\bs\not\prec\tilde{\bs}$ shall mean the negation of $\bs\prec\tilde{\bs}$. Notice that the relation ``$\prec$'' is preserved under projections: If $\bs\prec\tilde{\bs}$ and $\bi\in [k]^\ell$, then $\bs_\bi\prec\tilde{\bs}_\bi$.

\paragraph{Semirings}
A \emph{semiring} $\mathcal{S}$ consists of a set $S$ equipped with two binary operations ``$+$'' and ``$\cdot$'' such that 
\begin{itemize}
\item
$(S,+)$ is a commutative monoid with an identity element ``$0_{\mathcal{S}}$'' (i.e., $(r+s)+t=r+(s+t)$, $0_{\mathcal{S}}+r=r+0_{\mathcal{S}}=r$, and $r+s=s+r$);
\item
$(S,\cdot)$ is a monoid with an identity element ``$1_{\mathcal{S}}$'' (i.e., $(r\cdot s)\cdot t=r\cdot (s\cdot t)$ and $1_{\mathcal{S}}\cdot r=r\cdot 1_{\mathcal{S}}=r$);
\item
``$\cdot$'' distributes over ``$+$'' (i.e., $r\cdot (s+t)=(r\cdot s)+(r\cdot t)$ and $(r+s)\cdot t=(r\cdot t)+(s\cdot t)$);
\item
``$0_{\mathcal{S}}$'' is a multiplicative absorbing element (i.e., $0_{\mathcal{S}}\cdot r=r\cdot 0_{\mathcal{S}}=0_{\mathcal{S}}$).
\end{itemize}
Examples of semirings are $\Z$, $\Q$, and $\R$ with the usual addition and multiplication operations, or the Boolean semiring $(\{0,1\},\vee,\wedge)$. 

Let $V$ be a finite set, and choose an element $s_v\in\mathcal{S}$ for each $v\in V$. We let the formal expression $\sum_{v\in V}s_v$ equal $0_{\mathcal{S}}$ if $V=\emptyset$.
To increase the readability, we shall usually write $0$ and $1$ for $0_{\mathcal{S}}$ and $1_{\mathcal{S}}$; the relevant semiring $\mathcal{S}$ will always be clear from the context.

\paragraph{Tensors}

Given a set $S$, an integer $k\in\N$, and a tuple $\bn=(n_1,\dots,n_k)\in\N^k$, by $\cT^{\bn}(S)$ we denote the set of functions from $[n_1]\times\dots\times [n_k]$ to $S$, which we visualise as hypermatrices or tensors having $k$ \emph{modes}, where the $i$-th mode has size $n_i$ for $i\in [k]$. If $\bn=n\cdot\bone_k=(n,\dots,n)$ is a constant tuple, $\cT^{\bn}(S)$ is a set of \emph{cubical} tensors, each of whose modes has the same length $n$. For example, if $\bn=n\cdot\bone_2=(n,n)$, $\cT^\bn(S)$ is the set of $n\times n$ matrices having entries in $S$.
We sometimes denote an element of $\cT^{\bn}(S)$ by $T=(t_\bi)$, where $\bi\in [n_1]\times\dots\times [n_k]$ and $t_\bi$ is the image of $\bi$ under $T$. Moreover, given two tuples $\bn\in\N^k$ and $\tilde{\bn}\in\N^{\ell}$, we sometimes write $\cT^{\bn,\tilde{\bn}}(S)$ for $\cT^{(\bn,\tilde{\bn})}(S)$, where $(\bn,\tilde{\bn})$ is the concatenation of $\bn$ and $\tilde{\bn}$. Whenever $k\geq 2$ and $n_i=1$ for some $i\in [k]$, we can (and will) identify $\cT^{\bn}(S)$ with $\cT^{\hat\bn}(S)$, where $\hat{\bn}\in \N^{k-1}$ is obtained from $\bn$ by deleting the $i$-th entry. Note that the definition of tensors straightforwardly extends to the case that $n_i=\ale$ for some $i\in [k]$.

\paragraph{Contraction}
Take a semiring $\mathcal{S}$. For $k,\ell,m\in\N$, take $\bn\in\N^k$, $\bp\in\N^\ell$, and $\bq\in\N^m$. 
The \emph{contraction} of two tensors $T=(t_\bi)\in\cT^{\bn,\bp}(\mathcal{S})$ and $\tilde T=(\tilde t_\bi)\in\cT^{\bp,\bq}(\mathcal{S})$, denoted by  
$T\overset{\mathrm{\ell}}{\ast}\tilde{T}$,
is the tensor in $\cT^{\bn,\bq}(\mathcal{S})$ such that, for $\bi\in [n_1]\times\dots\times [n_k]$ and $\bj\in[q_1]\times\dots\times [q_m]$, the $(\bi,\bj)$-th entry of $T\overset{\mathrm{\ell}}{\ast}\tilde{T}$ is given by 
\begin{align}
\label{eqn_418_86}
\sum_{\bz\in [p_1]\times\dots\times [p_\ell]}t_{(\bi,\bz)}\tilde{t}_{(\bz,\bj)}
\end{align}
(where the addition and multiplication are meant in the semiring $\mathcal{S}$). This notation straightforwardly extends to the cases when $k=0$ or $m=0$, i.e., when we are contracting over all modes of $T$ or $\tilde{T}$. In such cases, we write $T\ast\tilde{T}$ for $T\overset{\mathrm{\ell}}{\ast}\tilde{T}$. 
The contraction operation is not associative in general. For instance, for $T$ and $\tilde T$ as above and $\hat T\in\cT^{\bn,\bq}(\mathcal{S})$, the expression $(T\cont{\ell}\tilde T)\cont{k+m}\hat{T}=(T\cont{\ell}\tilde T)\ast\hat{T}$ is well defined, while changing the order of the contractions results in an expression that is not well defined in general.
On the other hand, it is not hard to check that the order in which the contractions are performed is irrelevant if the contractions are taken over disjoint sets of modes. 
When the order does matter, we include brackets. 

\begin{example}
Given two vectors $\bu,\bv\in\cT^{p}(\R)$ and two matrices $M\in \cT^{n, p}(\R)$, $N\in\cT^{p, q}(\R)$, we have that
$\bu\cont{1}\bv=\bu\ast \bv=\bu^T\bv$ (the dot product of $\bu$ and $\bv$),
$M\cont{1}\bu=M\ast\bu=M\bu$ (the matrix-vector product of $M$ and $\bu$), and
$M\cont{1}N=MN$ (the matrix 
product of $M$ and $N$), as can be easily checked by applying~\eqref{eqn_418_86}.
\end{example}
For $\bi\in [n_1]\times\dots\times [n_k]$, we denote by $E_\bi$ the \emph{$\bi$-th standard unit tensor}; i.e., the tensor in $\cT^{\bn}(\mathcal{S})$ all of whose entries are $0_{\mathcal{S}}$, except the $\bi$-th entry that is $1_{\mathcal{S}}$. Given $T\in \cT^{\bn}(\mathcal{S})$, notice that $E_\bi\cont{k} T=E_\bi\ast T$ is the $\bi$-th entry of $T$ (in fact, we shall always indicate the entries of a tensor in this way). The \emph{support} of $T$ is the set of indices of all nonzero entries of $T$; i.e., the set 
$\supp(T)=\{\bi\in [n_1]\times\dots\times [n_k]:E_\bi\ast T\neq 0_{\mathcal{S}}\}$.

\subsection{Promise CSPs}
\label{subsec_prelimns_pcsps}
\paragraph{Structures}
A \textit{signature} $\sigma$ is a finite set of relation symbols $R$, each having an \textit{arity} 
$\ar(R)\in\N$.
A \textit{$\sigma$-structure} $\A$ consists of a set $A$ (called the \textit{domain}) and, for each $R\in\sigma$, a relation $R^{\A}\subseteq A^{\ar(R)}$. A $\sigma$-structure $\A$ is finite if the size $|A|$ of its domain $A$
is finite. In this case, we often assume that the domain of $\A$ is $A=[n]$. (In general, we will reserve the letter $n$ to denote the domain size of $\A$.)

Let $\A$ and $\B$ be $\sigma$-structures. A \emph{homomorphism} from $\A$ to $\B$ is a map $h:A\to B$ such that, for each $R\in\sigma$ with $r=\ar(R)$ and for each $\ba=(a_1,\ldots,a_r)\in A^r$, if $\ba\in R^{\A}$ then $h(\ba)=(h(a_1),\ldots,h(a_r))\in R^{\B}$. We denote the existence of a homomorphism from $\A$ to $\B$ by $\A\to\B$.
Fix a pair $(\A,\B)$ of $\sigma$-structures such that $\A\to\B$. The
\emph{promise constraint satisfaction problem} parameterised by the \textit{template} $(\A,\B)$, denoted by
$\PCSP(\A,\B)$, is the following decision problem:\footnote{From now on, we
shall work only with the \emph{decision} version of the PCSP.}
The input is a $\sigma$-structure $\X$ and the goal is to answer $\YES$ if $\X\to\A$ and $\NO$ if $\X\not\to\B$. The promise is that it is not the case that $\X\not\to\A$ and $\X\to\B$. 
We write $\CSP(\A)$ for $\PCSP(\A,\A)$, the classic (non-promise) constraint satisfaction problem.

\paragraph{Polymorphisms and minions}
The algebraic theory of $\PCSP$s developed in~\cite{BG21:sicomp,BBKO21} relies on the notions of polymorphism and minion. 
Let $\A$ be a $\sigma$-structure. For $L\in\N$, the \emph{$L$-th power} of $\A$ is the $\sigma$-structure $\A^L$ with domain $A^L$ whose relations are defined as follows: Given $R\in\sigma$ and an $L\times \ar(R)$ matrix $M$ such that all rows of $M$ are tuples in $R^\A$, the columns of $M$ form a tuple in $R^{\A^L}$. 

An $L$-ary \emph{polymorphism} of a $\PCSP$ template $(\A,\B)$ is a homomorphism from $\A^L$ to $\B$.
Minions were defined in~\cite{BBKO21} as sets of functions with certain properties. We shall use here the abstract definition of minions, following~\cite{bgwz20} and subsequent literature in (P)CSPs.
A \emph{minion} $\mathscr{M}$ consists in the disjoint union of nonempty sets $\mathscr{M}^{(L)}$ for $L\in \N$ equipped with (so-called \emph{minor}) operations $(\cdot)_{/\pi}:\mathscr{M}^{(L)}\rightarrow\mathscr{M}^{(L')}$ for all functions $\pi:[L]\rightarrow [L']$, which satisfy
$M_{/\operatorname{id}}=M$ and, for
 $\pi:[L]\rightarrow [L']$ and $\tilde{\pi}:[L']\rightarrow [L'']$,
 $(M_{/\pi})_{/\tilde{\pi}}=M_{/\tilde{\pi}\circ \pi}$
for all $M\in\mathscr{M}^{(L)}$. 

\begin{example}
\label{example_minion_pol_A_B}
The set $\Pol(\A,\B)$ of all polymorphisms of a $\PCSP$ template $(\A,\B)$ is a minion with the minor operations defined by 
$f_{/\pi}(a_1,\dots,a_{L'})= f(a_{\pi(1)},\dots,a_{\pi(L)})$
for $f:\A^L\to\B$ and $\pi:[L]\to [L']$.
In this minion, the minor operations correspond to identifying coordinates, permuting coordinates, and introducing dummy coordinates (of polymorphisms).
\end{example}
\begin{example}
\label{examples_famous_minions}
Other examples of minions that shall appear frequently in this work are $\Qconv$, $\Zaff$, and $\Hminion$, capturing the power of the algorithms $\BLP$, $\AIP$, and $\ArcC$, respectively. 
The $L$-ary elements of $\Qconv$ are rational vectors of size $L$ that are \textit{stochastic} (i.e., whose entries are nonnegative and sum up to $1$), with the minor
operations defined as follows: For $\bq\in\Qconv^{(L)}$
and $\pi:[L]\to[L']$, $\bq_{/\pi}= P\bq$, where $P$ is the $L'\times
L$ matrix whose $(i,j)$-th entry is $1$ if $\pi(j)=i$, and $0$ otherwise.
$\Zaff$ is defined similarly to $\Qconv$, the only difference being that its $L$-ary elements are affine integer vectors (i.e., their entries are integer -- possibly negative -- numbers and sum up to $1$).
$\Hminion$ is the minion of polymorphisms of the $\CSP$ template $\textsc{Horn-3-Sat}$, i.e., the Boolean structure whose three relations are ``$x\wedge y\Rightarrow z$'', $\{0\}$, and $\{1\}$.\footnote{$\textsc{Horn-3-SAT}$ is often defined in the literature with one more relation ``$x\wedge y\Rightarrow \neg z$'' but this is redundant.} Equivalently (cf.~\cite{BBKO21}), $\Hminion$ can be described as follows: For any $L\in\N$, the $L$-ary elements of $\Hminion$ are Boolean functions of the form $f_Z(x_1,\dots,x_L)=\bigwedge_{z\in Z}x_z$ for any $Z\subseteq [L]$, $Z\neq\emptyset$; the minor operations are defined as in Example~\ref{example_minion_pol_A_B}.
We shall also mention the minion $\Mblpaip$ capturing the algorithm $\BA$ described in the Introduction. Its $L$-ary elements are $L\times 2$ matrices whose first column $\bu$ belongs to $\Qconv^{(L)}$ and whose second column $\bv$ belongs to $\Zaff^{(L)}$, and such that if the $i$-th entry of $\bu$ is zero then the $i$-th entry of $\bv$ is also zero, for each $i\in [L]$. The minor operation is defined on each column individually; i.e., $\pair{\bu}{\bv}_{/\pi}=\pair{\bu_{/\pi}}{\bv_{/\pi}}$.
\end{example}

For two minions $\mathscr{M}$ and $\mathscr{N}$, a \emph{minion homomorphism} $\xi:\mathscr{M}\rightarrow\mathscr{N}$ is a map that preserves arities and minors: Given $M\in\mathscr{M}^{(L)}$ and $\pi:[L]\rightarrow[L']$, $\xi(M)\in \mathscr{N}^{(L)}$ and $\xi(M_{/\pi})=\xi(M)_{/\pi}$. We denote the existence of a minion homomorphism from $\mathscr{M}$ to $\mathscr{N}$ by $\mathscr{M}\to\mathscr{N}$. 
If a minion homomorphism is invertible as a function -- in which case, its inverse must also be a minion homomorphism -- we say that it is a \textit{minion isomorphism}. 

We will also need the concept of free structure from~\cite{BBKO21}. 
Let $\mathscr{M}$ be a minion and let $\A$ be a (finite) $\sigma$-structure. The
\emph{free structure} of $\sM$ generated by $\A$ is a $\sigma$-structure
$\bF_{\sM}(\A)$ with domain $\sM^{(|A|)}$ (potentially infinite). Given a
relation symbol $R\in\sigma$ of arity $r$, a tuple $(M_1,\dots,M_r)$ of elements
of $\sM^{(|A|)}$ belongs to $R^{\bF_{\sM}(\A)}$ if and only if there is some
$Q\in \sM^{(|R^\A|)}$ such that $M_i=Q_{/\pi_i}$ for each $i\in[r]$, where
$\pi_i:R^\A\to A$ maps $\ba\in R^\A$ to its $i$-th coordinate $a_i$. (Here and henceforth, we are implicitly identifying $R^\A$ with $|R^\A|$ and $A$ with $|A|$ for readability.) The
definition of free structure may at this point strike the reader as rather
technical. We shall see that, if we consider certain quite general classes of
minions, this object unveils an interesting geometric description of linear and multilinear nature.

\subsection{Relaxations and hierarchies}
\label{sec_relaxations_hierarchies}
The following relaxations of $\parPCSPs$ shall appear in this paper: \textit{Arc
consistency} ($\ArcC$) is a propagation algorithm that checks for the existence
of assignments satisfying the local constraints of the given $\parPCSP$
instance~\cite{Mackowrth77:aij}; the \textit{basic linear programming} ($\BLP$)
relaxation looks for compatible probability distributions on
assignments~\cite{Kun12:itcs}; the \textit{affine integer programming} ($\AIP$)
relaxation turns the constraints into linear equations, that can be solved over
the integers using (a variant of) Gaussian elimination~\cite{BG19}; the \textit{basic
semidefinite programming} ($\SDP$) relaxation is essentially a strengthening of
$\BLP$, where probabilities are replaced by vectors satisfying orthogonality
requirements~\cite{Raghavendra08:everycsp}; the \textit{combined basic linear programming and affine integer programming} ($\BA$) relaxation is a hybrid algorithm blending $\BLP$ and $\AIP$~\cite{bgwz20}. We refer the reader to~\cite{BBKO21} for a formal description of $\ArcC$, $\BLP$, and $\AIP$, and to~\cite{bgwz20} for $\BA$, while $\SDP$ shall be formally defined later in this section.

In this work, we shall mainly focus on algorithmic \emph{hierarchies}.
The bounded-width ($\BW^k$) hierarchy (also known as the \textit{local consistency checking algorithm}) refines $\ArcC$ by propagating local solutions over bigger and bigger portions of the instance, while the Sherali--Adams LP ($\SA^k$), affine integer programming ($\AIP^k$), Sum-of-Squares ($\SoS^k$), and combined basic linear programming and affine integer programming ($\BA^k$) hierarchies strengthen the $\BLP$, $\AIP$, $\SDP$, and $\BA$ relaxations, respectively, by looking for compatible distributions of assignments over sub-instances of some fixed size.
Below, we give a formal description of 
the five hierarchies mentioned above, as well as the $\SDP$ relaxation (see also the discussion in Appendix~\ref{appendix_notes_relaxations_hierarchies} for a comparison with different formulations of these
algorithms appearing in the literature on $\parPCSPs$). 

\paragraph{$\mbox{BW}^{\mbox{k}}$}
Given two $\sigma$-structures $\X$ and $\A$ and a subset $S\subseteq X$, a \emph{partial homomorphism} from $\X$ to $\A$ with domain $S$ is a homomorphism from $\X[S]$ to $\A$, where $\X[S]$ is the substructure of $\X$ \emph{induced} by $S$ -- i.e., it is the $\sigma$-structure whose domain is $S$ and, for any $R\in\sigma$, $R^{\X[S]}=R^\X\cap S^{\ar(R)}$.
We say that the $k$-th level of the \emph{bounded-width algorithm} accepts when applied to $\X$ and $\A$, and we write $\BW^k(\X,\A)=\YES$, if there exists a nonempty collection $\mathcal{F}$ of partial homomorphisms from $\X$ to $\A$ with at most $k$-element domains such that $(i)$ $\mathcal{F}$ is closed under restrictions, i.e., for every $f\in \mathcal{F}$ and every $V\subseteq\dom(f)$, $f|_{V}\in \mathcal{F}$, and $(ii)$ $\mathcal{F}$ has the extension property up to $k$, i.e., for every $f\in \mathcal{F}$ and every $V\subseteq X$ with $|V|\leq k$ and $\dom(f)\subseteq V$, there exists $g\in \mathcal{F}$ such that $g$ extends $f$ and $\dom(g)=V$.

We say that $\BW^k$ \emph{solves} a $\PCSP$ template $(\A,\B)$ if $\X\to\B$
whenever $\BW^k(\X,\A)=\YES$. Note that the algorithm is always complete: If
$\X\to\A$ then $\BW^k(\X,\A)=\YES$.

\paragraph*{$\mbox{SA}^{\mbox{k}}$, $\mbox{AIP}^{\mbox{k}}$, and $\mbox{BA}^{\mbox{k}}$}
For $k\in\N$, we say that a $\sigma$-structure $\A$ is \emph{$k$-enhanced} if the signature $\sigma$ contains a $k$-ary symbol $R_k$ such that $R_k^\A=A^k$.
Given two $k$-enhanced $\sigma$-structures $\X,\A$, we introduce a variable $\lambda_{R,\bx,\ba}$ for every $R\in\sigma$, $\bx\in R^\X$, $\ba\in R^\A$. Consider the following system of equations:
\begin{align}
\label{eqn_SA_and_AIPk}
\left.
\tag{$\clubsuit$}
\begin{array}{llllllll}
\mbox{($\clubsuit 1$)} & \displaystyle\sum_{\ba\in R^\A}\lambda_{R,\bx,\ba}&=&1 & R\in\sigma,\bx\in R^\X\\
\mbox{($\clubsuit 2$)} & \displaystyle \sum_{\ba\in R^\A,\,\ba_\bi=\bb}\lambda_{R,\bx,\ba}&=&\lambda_{R_k,\bx_\bi,\bb} & 
R\in\sigma,\bx\in R^\X,\bi\in [\ar(R)]^k,\bb\in A^k\\
\mbox{($\clubsuit 3$)} & \displaystyle
  \lambda_{R,\bx,\ba}&=&0&R\in\sigma,\bx\in R^\X,\ba\in R^\A,
  \bx\not\prec\ba.\footnotemark
\end{array}
\right\}
\end{align}
\footnotetext{\label{foot:rep}The condition $\bx\not\prec\ba$ formalises the requirement that
the same variables (elements of $\bx$) should not be assigned different values
(elements of $\ba$). Some papers avoid this requirement by imposing that
$\parPCSP$ instances should have no repetition of variables in the constraint scopes; i.e., elements of
$\bx$ should all be distinct.}

We say that the $k$-th level of the \emph{Sherali--Adams linear programming hierarchy} accepts when applied to $\X$ and $\A$, and we write $\SA^k(\X,\A)=\YES$, if the system~\eqref{eqn_SA_and_AIPk} admits a solution such that all variables take rational nonnegative values.\footnote{See Appendix~\ref{subsec_appendix_SA_AIPk} (in particular, Lemma~\ref{lem_our_SA_equals_BD_SA}) for an equivalent description of $\SA^k$ that does not involve the enhancement operation.} 
Similarly, we say that the $k$-th level of the \emph{affine integer programming hierarchy} accepts when applied to $\X$ and $\A$, and we write $\AIP^k(\X,\A)=\YES$, if the system above admits a solution such that all variables take integer values. Moreover, we say that the $k$-th level of the \emph{combined basic linear programming and affine
integer programming hierarchy} accepts when applied to $\X$ and $\A$, and we write $\BA^k(\X,\A)=\YES$, if the system above admits both a solution such that all variables take rational nonnegative values and a solution such that all variables take integer values, and the following \emph{refinement condition} holds: Denoting the rational nonnegative and the integer solutions by the superscripts $(\Bmat)$ and $(\Amat)$, respectively, we require that $\lambda^{(\Amat)}_{R,\bx,\ba}=0$ whenever $\lambda^{(\Bmat)}_{R,\bx,\ba}=0$, for each $R\in\sigma,\bx\in R^\X,\ba\in R^\A$.

We say that $\SA^k$ \emph{solves} a $\PCSP$ template $(\A,\B)$ if $\X\to\B$
whenever $\SA^k(\X,\A)=\YES$. The definition for $\AIP^k$ and $\BA^k$ is analogous. Note
that the three algorithms are always complete: If $\X\to\A$ then
$\SA^k(\X,\A)=\AIP^k(\X,\A)=\BA^k(\X,\A)=\YES$. 

\paragraph{SDP}

Given two $\sigma$-structures $\X,\A$, let $\gamma=
|X|\cdot|A|+\sum_{R\in\sigma}|R^\X|\cdot|R^\A|$. 
We introduce a variable $\blambda_{x,a}$ taking values in $\R^\gamma$ for every $x\in X$, $a\in A$, and a variable $\blambda_{R,\bx,\ba}$ taking values in $\R^\gamma$ for every $R\in\sigma,\bx\in R^\X,\ba\in R^\A$.\footnote{We point out that requiring the SDP vectors to be elements of $\R^\gamma$ is equivalent to letting them belong to an arbitrary finite-dimensional real vector space. This is because a real square matrix has a Cholesky decomposition if and only if it has a square Cholesky decomposition (if and only if it is positive semidefinite). The same holds for the $\SoS^k$ system~\eqref{eqn_Lasserre_def}.}
Consider the following system of equations:
\begin{align}
\left.
\label{eqn_SDP_def}
\tag{$\vardiamond$}
\begin{array}{llllll}
\mbox{($\vardiamond1$)}&\displaystyle\sum_{a\in A}\|\blambda_{x,a}\|^2&=&1& x\in X\\
\mbox{($\vardiamond2$)}&\blambda_{x,a}\cdot\blambda_{x,a'}&=&0&x\in X,a\neq a'\in A\\
\mbox{($\vardiamond3$)}&\blambda_{R,\bx,\ba}\cdot\blambda_{R,\bx,\ba'}&=&0&R\in\sigma,\bx\in R^\X,\ba\neq\ba'\in R^\A\\
\mbox{($\vardiamond4$)}&\displaystyle\sum_{\substack{\ba\in R^\A,\,a_i=a}}\blambda_{R,\bx,\ba}&=&\blambda_{x_i,a}&R\in\sigma,\bx\in R^\X,a\in A, i\in [\ar(R)].
\end{array}
\right\}
\end{align}
We say that the \emph{standard semidefinite programming relaxation} accepts when applied to $\X$ and $\A$, and we write $\SDP(\X,\A)=\YES$, if the system~\eqref{eqn_SDP_def} admits a solution. We say that $\SDP$ \emph{solves} a $\PCSP$ template $(\A,\B)$ if $\X\to\B$ whenever $\SDP(\X,\A)=\YES$. Note that the algorithm is always complete: If $\X\to\A$ then $\SDP(\X,\A)=\YES$.

\paragraph*{$\mbox{SoS}^{\mbox{k}}$}
Given two $k$-enhanced $\sigma$-structures $\X,\A$,
let $\gamma=\sum_{R\in\sigma}|R^\X|\cdot|R^\A|$. We introduce a variable $\blambda_{R,\bx,\ba}$ taking values in $\R^\gamma$ for every $R\in\sigma,\bx\in R^\X,\ba\in R^\A$. Consider the following system of equations:
\begin{align}
\left.
\label{eqn_Lasserre_def}
\tag{$\spadesuit$}
\begin{array}{lllllll}
\mbox{($\spadesuit 1$)} & \displaystyle\sum_{\ba\in R^\A}\|\blambda_{R,\bx,\ba}\|^2&=&1 & R\in\sigma,\bx\in R^\X\\
\mbox{($\spadesuit 2$)}&\blambda_{R,\bx,\ba}\cdot\blambda_{R,\bx,\ba'}&=&0&R\in\sigma,\bx\in R^\X,\ba\neq \ba'\in R^\A\\
\mbox{($\spadesuit 3$)} & \displaystyle \sum_{\ba\in R^\A,\,\ba_\bi=\bb}\blambda_{R,\bx,\ba}&=&\blambda_{R_k,\bx_\bi,\bb} & 
R\in\sigma,\bx\in R^\X,\bi\in [\ar(R)]^k,\bb\in A^k\\
\mbox{($\spadesuit 4$)} & \displaystyle
\|\blambda_{R,\bx,\ba}\|^2&=&0&R\in\sigma,\bx\in R^\X,\ba\in R^\A,
  \bx\not\prec\ba.
\end{array}
\right\}
\end{align}

We say that the $k$-th level of the \emph{Sum-of-Squares semidefinite programming hierarchy} accepts when applied to $\X$ and $\A$, and we write $\SoS^k(\X,\A)=\YES$, if the system~\eqref{eqn_Lasserre_def} admits a solution.
We say that $\SoS^k$ \emph{solves} a $\PCSP$ template $(\A,\B)$ if $\X\to\B$ whenever $\SoS^k(\X,\A)=\YES$. Note that the algorithm is always complete: If $\X\to\A$ then $\SoS^k(\X,\A)=\YES$.

\begin{rem}
\label{rem_enhancement}
Note that the definitions of the $\SA^k$, $\AIP^k$, $\BA^k$, and $\SoS^k$ hierarchies given above require that the structures $\X$ and $\A$ to which the hierarchies are applied should be $k$-enhanced. This allows expressing marginalisation requirements between partial assignments of possibly different sizes via the same condition expressing marginalisation between constraints and variables appearing in the constraint (namely, condition $\clubsuit 2$ in~\eqref{eqn_SA_and_AIPk} and condition $\spadesuit 3$ in~\eqref{eqn_Lasserre_def}), applied to the extra complete relation $R_k$. See, for example, Lemma~\ref{lem_our_SA_equals_BD_SA} in Appendix~\ref{subsec_appendix_SA_AIPk} in the case of the Sherali--Adams hiererchy. In contrast, the definition of $\BW^k$ does not explicitly require $k$-enhancement. In fact, it is clear from that description that, for any two structures $\X$ and $\A$, it holds that $\BW^k(\X,\A)=\BW^k(\tilde\X,\tilde\A)$, where $\tilde\X$ (resp. $\tilde\A$) is the $k$-enhanced version of $\X$ (resp. $\A$). Note that this is merely a consequence of the specific formulation we use to define the hierarchies. The reason why we adopt the current formulations is that they are formally closer to the tensorisation construction, as we shall see later in the paper.
\end{rem}

\section{Minion tests}
\label{sec_minion_tests}

Let $(\A,\B)$ be a $\PCSP$ template. As discussed in Section~\ref{sec_intro},
any (potentially infinite) structure $\T$ on the same signature as $\A$ and $\B$
can be viewed as a test for the decision problem $\PCSP(\A,\B)$: Given an
instance $\X$, the test returns $\YES$ if $\X\to\T$, and $\NO$ otherwise. 
As the next definition illustrates, minions provide a systematic method to build tests for $\PCSPs$.

\begin{defn}
\label{defn_minion_test}
Let $\Mminion$ be a minion. The \emph{minion test} $\Test{\Mminion}{}$ is the
  decision problem defined as follows: Given two $\sigma$-structures $\X$ and $\A$, return $\YES$ if $\X\to\freeM(\A)$, and $\NO$ otherwise. 
\end{defn}
If $\X$ is an instance of $\PCSP(\A,\B)$ for some template $(\A,\B)$, we write $\Test{\Mminion}{}(\X,\A)=\YES$ if $\Test{\Mminion}{}$ applied to $\X$ and $\A$ returns $\YES$ (i.e., if $\X\to\freeM(\A)$), and  we write $\Test{\Mminion}{}(\X,\A)=\NO$ otherwise. 
Note that, in the expression ``$\Test{\Mminion}{}(\X,\A)$'', $\X$ is the input structure of the $\PCSP$, while $\A$ is the fixed structure from the $\PCSP$ template.

Leaving $\SDP$ aside for the moment, it turns out that the algebraic structure lying at the core of all relaxations mentioned in Section~\ref{sec_relaxations_hierarchies}, of seemingly different nature, is the same, as all of them are minion tests for specific minions.
\begin{thm}[\protect{\cite[Theorems~7.4, 7.9, and 7.19]{BBKO21},\cite[Lemma~5.4]{bgwz20}}]
\label{thm_cited_minion_tests_cases}
$
\ArcC=\Test{\Hminion}{}$,
$
\BLP=\Test{\Qconv}{}$,
$
\AIP=\Test{\Zaff}{}$,
$
\BA=\Test{\Mblpaip}{}$.
\end{thm}

One reason why minion tests are an interesting type of tests is that they are always complete.
\begin{prop}
\label{prop_minion_tests_are_complete}
$\Test{\Mminion}{}$ is complete for any minion $\Mminion$; i.e., for any $\X$ and $\A$ with $\X\to\A$, we have $\X\to\freeM(\A)$.
\end{prop}
This immediately follows from the next lemma, implicitly proved in~\cite{BBKO21} for the case of function minions (see Remark~\ref{rem_linear_to_function_minion}). For completeness, we include below the simple proof (which closely follows the proof in~\cite{BBKO21}, see the comment after Definition 4.1 therein).
\begin{lem}[\cite{BBKO21}]
\label{lem_A_maps_free_structure_of_A}
Let $\mathscr{M}$ be a minion and let $\A$ be a $\sigma$-structure. Then, $\A\to\mathbb{F}_{\mathscr{M}}(\A)$.
\end{lem}
\begin{proof}
Take a unary element $M\in\mathscr{M}^{(1)}$, and consider the map 
\begin{align*}
f:A&\to\mathscr{M}^{(n)}\\
a&\mapsto M_{/\rho_a}
\end{align*}
where $\rho_a:[1]\to [n]=A$ is defined by $\rho_a(1)=a$. Take $R\in\sigma$ of arity $r$, and consider a tuple $\ba=(a_1,\dots,a_r)\in R^\A$. Let $m=|R^\A|$, and consider the function $\pi:[1]\to [m]$ defined by $\pi(1)=\ba$. Let ${Q}=M_{/\pi}\in\mathscr{M}^{(m)}$. For each $i\in [r]$, recall the function $\pi_i:[m]\to [n]$ defined by $\pi_i(\bb)=b_i$, where $\bb=(b_1,\dots,b_r)\in R^\A$. Observe that $\rho_{a_i}=\pi_i\circ\pi$ for each $i\in[r]$. We obtain
\begin{align*}
f(\ba)&=(f(a_1),\dots,f(a_r))
=
(M_{/\rho_{a_1}},\dots,M_{/\rho_{a_r}})
=
(M_{/\pi_1\circ\pi},\dots,M_{/\pi_r\circ\pi})\\
&=  
  ((M_{/\pi})_{/\pi_1},\dots,(M_{/\pi})_{/\pi_r})
=
({Q}_{/\pi_1},\dots,{Q}_{/\pi_r})
\in 
R^{\mathbb{F}_{\mathscr{M}}(\A)},
\end{align*}
thus showing that $f$ is a homomorphism from $\A$ to $\mathbb{F}_{\mathscr{M}}(\A)$.
\end{proof}

A second feature of minion tests is that their soundness can be characterised algebraically, as stated in the next proposition.
\begin{prop}
\label{prop_solvability_minion_tests}
Let $\Mminion$ be a minion and let $(\A,\B)$ be a $\PCSP$ template. Then, $\Test{\Mminion}{}$ solves $\PCSP(\A,\B)$ if and only if $\Mminion\to\Pol(\A,\B)$.
\end{prop}

Our proof of Proposition~\ref{prop_solvability_minion_tests} uses a standard compactness argument from~\cite{malcev1937untersuchungen} (see also~\cite{Hodges1993model}), and it follows the lines of~\cite[Theorem~7.9 and Remark~7.13]{BBKO21}, where the same result is derived from K\H{o}nig's Lemma in the restricted case of \emph{locally countable minions} -- i.e., minions $\Mminion$ having the property that the set $\Mminion^{(L)}$ is countable for any $L$. Since the minion $\Sminion$ described in Definition~\ref{defn_minion_SDP} is not locally countable, we shall need this stronger version of the result when proving that $\Sminion$ provides an algebraic characterisation of the power of $\SDP$ (cf.~Theorem~\ref{thm_characterisation_power_SDP}).
  
We say that a potentially infinite $\sigma$-structure $\B$ is \emph{compact} if, for any potentially infinite $\sigma$-structure $\A$, $\A\to\B$ if and only if $\A'\to\B$ for every finite substructure $\A'$ of $\A$.
The following result is a direct consequence of the uncountable version of the \emph{compactness theorem of logic} (\cite[Theorem~6.1.1]{Hodges1993model}).
\begin{thm}[\cite{malcev1937untersuchungen,Hodges1993model}] 
\label{thm_finite_structures_are_compact}
Every finite  
$\sigma$-structure is compact.
\end{thm}

\begin{proof}[Proof of Proposition~\ref{prop_solvability_minion_tests}]
We first observe that the condition $\Mminion\to\Pol(\A,\B)$ is equivalent to the condition $\freeM(\A)\to\B$ by~\cite[Lemma~4.4]{BBKO21} 
(see also~\cite{cz23sicomp:clap} for the proof for  abstract minions).

Suppose that $\freeM(\A)\to \B$. Given an instance $\X$, if $\Test{\Mminion}{}(\X,\A)=\YES$ then $\X\to\freeM(\A)$, and composing the two homomorphisms yields $\X\to\B$. Hence, $\Test{\Mminion}{}$ is sound on the template $(\A,\B)$. Since, as noted above, $\Test{\Mminion}{}$ is always complete, we deduce that $\Test{\Mminion}{}$ solves $\PCSP(\A,\B)$.

Conversely, suppose that $\Test{\Mminion}{}$ solves $\PCSP(\A,\B)$. Let $\textbf{F}$ be a finite substructure of $\freeM(\A)$, and notice that the inclusion map yields a homomorphism from $\textbf{F}$ to $\freeM(\A)$. Hence, $\Test{\Mminion}{}(\textbf{F},\A)=\YES$, so $\textbf{F}\to\B$. Since $\B$ is compact by Theorem~\ref{thm_finite_structures_are_compact}, we deduce that $\freeM(\A)\to\B$, as required.
\end{proof}

\section{A minion for SDP}
\label{sec_minion_for_SDP}

The goal of this section is to design a minion $\Sminion$ capturing the power of $\SDP$, thus showing that, similarly to $\ArcC$, $\BLP$, $\AIP$, and $\BLP+\AIP$, also $\SDP$ is a minion test. We remark that an equivalent characterisation for the power of $\SDP$ was also obtained independently by Brakensiek, Guruswami, and Sandeep in~\cite{bgs_robust23stoc}.

\begin{defn}
\label{defn_minion_SDP}
For $L\in\mathbb{N}$, let $\Sminion^{(L)}$ be the set of real $L\times\ale$ matrices ${M}$ such that
\begin{align}
\mbox{(C1)}
\;
\csupp(M) \mbox{ is finite}&&
\mbox{(C2)}
\;
{M}{M}^T \mbox{ is a diagonal matrix}&&
\mbox{(C3)}
\;
\tr({M}{M}^T)=1.
\end{align}
Given a function $\pi:[L]\rightarrow [L']$ and a matrix ${M}\in \Sminion^{(L)}$,
  we let ${M}_{/\pi}= P {M}$, where $P$ is the $L'\times L$ matrix whose
  $(i,j)$-th entry is $1$ if $\pi(j)=i$, and $0$ otherwise. We set 
$\Sminion=\bigsqcup_{L\in\mathbb{N}}\Sminion^{(L)}$.
\end{defn}
First of all, we show that $\Sminion$ is closed with respect to the minor maps described above and, thus, it is indeed a minion. 

\begin{prop}
\label{prop_Sminion_is_linear_minion}
$\Sminion$ is a 
minion.
\end{prop}
\begin{proof}
For $\pi:[L]\rightarrow [L']$ and ${M}\in \Sminion^{(L)}$, we have that $M_{/\pi}=PM\in\cT^{L',\ale}(\R)$ and $\csupp(PM)$ is finite (where $P$ is the $L'\times
L$ matrix associated with $\pi$, as per Definition~\ref{defn_minion_SDP}). One easily checks that $(i)$
$P^T\bone_{L'}=\bone_{L}$, and
$(ii)$
$P P^T$ is a diagonal matrix. Using that both $MM^T$ and $PP^T$ are diagonal, we find that $M_{/\pi}(M_{/\pi})^T=PMM^TP^T$ is diagonal, too. Moreover, since the trace of a diagonal matrix equals the sum of its entries, we obtain
\begin{align*}
\tr({M}_{/\pi}({M}_{/\pi})^T)=\bone_{L'}^T{M}_{/\pi}({M}_{/\pi})^T\bone_{L'}=
\bone_{L'}^TP {M}{M}^TP^T\bone_{L'}
=\bone_L^T{M}{M}^T\bone_L=\tr({M}{M}^T)=1.
\end{align*}
It follows that ${M}_{/\pi}\in \Sminion^{(L')}$. Furthermore, one easily checks that ${M}_{/\operatorname{id}}={M}$ and, given $\tilde\pi:[L']\rightarrow[L'']$, 
$
({M}_{/\pi})_{/\tilde\pi}=
{M}_{/\tilde\pi\circ\pi}
$, which concludes the proof that $\Sminion$ is a minion. 
\end{proof}

In the remaining part of this section, we prove that $\Sminion$ captures the power of the $\SDP$ relaxation, as stated below.

\begin{prop}
\label{prop_acceptance_SDP}
$\SDP=\Test{\Sminion}{}$. In other words, given two $\sigma$-structures $\X$ and $\A$, $\SDP(\X,\A)=\YES$ if and only if $\X\to\freeS(\A)$.
\end{prop} 
Combining Propositions~\ref{prop_acceptance_SDP} and~\ref{prop_solvability_minion_tests}, we immediately obtain the following algebraic characterisation of the power of $\SDP$.
\begin{graytbox}	
\begin{thm}	
\label{thm_characterisation_power_SDP}	
Let $(\A,\B)$ be a $\PCSP$ template. Then, $\SDP$ solves $\PCSP(\A,\B)$ if and only if $\Sminion\to\Pol(\A,\B)$.	
\end{thm}	
\end{graytbox}

For a $\sigma$-structure $\A$, a symbol $R\in\sigma$ of arity $r$, and a number $i\in [r]$, we consider the $|A|\times |R^\A|$ matrix $P_i$  whose $(a,\ba)$-th entry is $1$ if $a_i=a$, and $0$ otherwise. We shall use the following simple description of the entries of $P_i$.
\begin{lem}[\protect{\cite[Section~4.1]{cz23sicomp:clap}}]
\label{lem_basic_P_i}
Let $\A$ be a $\sigma$-structure, let $R\in\sigma$ of arity $r$, and let $i\in [r],a\in A$. Then,
\begin{align*}
\be_a^T P_i=\sum_{\substack{\ba\in R^\A\\a_i=a}}\be_\ba^T.
\end{align*}
\end{lem}

We shall also need the following observation, whose simple proof is deferred to Appendix~\ref{subsec_appendix_SDP}.

\begin{prop}
\label{prop_trivial_stuff_SDP}
Let $\X,\A$ be two $\sigma$-structures.
The system~\eqref{eqn_SDP_def} implies the following facts:
\begin{align*}
\begin{array}{lllllll}
(i)&\displaystyle\|\sum_{a\in A}\blambda_{x,a}\|^2=1 & x\in X;\\
(ii)&\displaystyle\sum_{\ba\in R^\A}\|\blambda_{R,\bx,\ba}\|^2=\|\sum_{\ba\in R^\A}\blambda_{R,\bx,\ba}\|^2=1 & R\in\sigma,\bx\in R^\X;\\
(iii)&\displaystyle\sum_{\substack{\ba\in R^\A\\a_i=a,\;a_j=a'}}\|\blambda_{R,\bx,\ba}\|^2=\blambda_{x_i,a}\cdot\blambda_{x_j,a'} & R\in\sigma,\bx\in R^\X,a,a'\in A,i,j\in[\ar(R)].
\intertext{If, in addition, $\X$ and $\A$ are $2$-enhanced,}
(iv)&\displaystyle\sum_{a\in A}\blambda_{x,a}=\sum_{a\in A}\blambda_{x',a} & x,x'\in X.
\end{array}
\end{align*}
\end{prop}

In order to prove that $\SDP=\Test{\Sminion}{}$, we essentially need to encode the vectors $\blambda$ witnessing an $\SDP$ solution as rows of matrices belonging to $\Sminion$. Note that the vectors $\blambda$ live in a vector space having a finite dimension -- namely, the number $\gamma=|X|\cdot|A|+\sum_{R\in\sigma}|R^\X|\cdot|R^\A|$ (cf. the description of $\SDP$ in Section~\ref{sec_relaxations_hierarchies}). On the other hand, the matrices in $\Sminion$ have rows of infinite size, living in $\R^\ale$. This issue is easily solved by considering a $\gamma$-dimensional subspace of $\R^\ale$ and working in an orthonormal basis of such subspace, through a standard orthonormalisation argument. 

\begin{proof}[Proof of Proposition~\ref{prop_acceptance_SDP}]
Suppose that $\SDP(\X,\A)=\YES$, and let the family of vectors $\blambda_{x,a},\blambda_{R,\bx,\ba}\in\R^\gamma$ witness it, for $x\in X,a\in A,R\in\sigma,\bx\in R^\X,\ba\in R^\A$, where $\gamma=
|X|\cdot|A|+\sum_{R\in\sigma}|R^\X|\cdot|R^\A|$. Consider the map $\xi:X\to\Sminion^{(|A|)}$ assigning to each $x\in X$ the $|A|\times\aleph_0$ matrix whose rows are the vectors $\blambda_{x,a}$ for each $a\in A$, padded with infinitely many zeroes; i.e.,
\begin{align*}
\be_a^T\xi(x)\be_j=
\left\{
\begin{array}{lllll}
\be_j^T\blambda_{x,a}&\mbox{ if }j\leq \gamma\\
0&\mbox{ otherwise}
\end{array}
\right.
\hspace{1cm}
x\in X,a\in A,j\in\N.
\end{align*}
We claim that $\xi$ is well defined. First,
  $\xi(x)\be_j=\bzero$ for each $j>\gamma$, so condition (C1) from Definition~\ref{defn_minion_SDP} is satisfied. Given $a,a'\in A$, it holds that $\be_a^T\xi(x)\xi(x)^T\be_{a'}
=
\blambda_{x,a}\cdot\blambda_{x,a'}$.
  If $a\neq a'$, this quantity is zero by $\vardiamond2$, so (C2) is satisfied. Finally,
\begin{align*}
\tr(\xi(x)\xi(x)^T)
&=
\sum_{a\in A}\be_a^T\xi(x)\xi(x)^T\be_a
=
\sum_{a\in A}\|\blambda_{x,a}\|^2=1
\end{align*}
by $\vardiamond1$,
  so (C3) is also satisfied and the claim is true. We now show that $\xi$ yields a homomorphism from $\X$ to $\freeS(\A)$. Take $R\in\sigma$ of arity $r$ and $\bx=(x_1,\dots,x_r)\in R^\X$. Consider the $|R^\A|\times\ale$ matrix $Q$ whose rows are the vectors 
  $\blambda_{R,\bx,\ba}$ for each $\ba\in R^\A$ padded with infinitely many zeroes.
  Using the same arguments as above, we have that $Q$ satisfies (C1) and
  $\be_\ba^TQQ^T\be_{\ba'}=\blambda_{R,\bx\,\ba}\cdot\blambda_{R,\bx,\ba'}$, so
  (C2) follows from $\vardiamond 3$ and (C3) from point $(ii)$ of Proposition~\ref{prop_trivial_stuff_SDP}. Therefore, $Q\in\Sminion^{(|R^\A|)}$. We now claim that $\xi(x_i)=Q_{/\pi_i}$ for each $i\in [r]$. Indeed, for $a\in A$ and $j\in[\gamma]$, we have 
\begin{align*}
\be_a^T\xi(x_i)\be_j
&=
\be_j^T\blambda_{x_i,a}
=
\sum_{\substack{\ba\in R^\A\\ a_i=a}}\be_j^T\blambda_{R,\bx,\ba}
=
\sum_{\substack{\ba\in R^\A\\ a_i=a}}\be_\ba^T Q\be_j
=
\be_a^T P_iQ\be_j
=
\be_a^TQ_{/\pi_i}\be_j,
\end{align*}
where the second and fourth equalities follow from $\vardiamond 4$ and Lemma~\ref{lem_basic_P_i}, respectively. Also, clearly,
$\be_a^T\xi(x_i)\be_j=\be_a^TQ_{/\pi_i}\be_j=0$
if $j\in\N\setminus [\gamma]$. As a consequence, the claim holds. It follows that $\xi(\bx)\in R^{\freeS(\A)}$, so that $\xi$ is a homomorphism.

Conversely, let $\xi:\X\to\freeS(\A)$ be a homomorphism.
For $R\in\sigma$ of arity $r$ and $\bx=(x_1,\dots,x_r)\in R^\X$, we can fix a
matrix $Q_{R,\bx}\in\Sminion^{(|R^\A|)}$ satisfying
$\xi(x_i)={Q_{R,\bx}}_{/\pi_i}$ for each $i\in [r]$. Consider the sets
$S_1=\{\xi(x)^T\be_a:x\in X,a\in A\}$ and
$S_2=\{Q_{R,\bx}^T\be_\ba:R\in\sigma,\bx\in R^\X,\ba\in R^\A\}$, and the vector
space $\mathcal{U}=\Span(S_1\cup S_2)\subseteq \R^\ale$. Observe that $\dim
\mathcal{U}\leq |S_1\cup S_2|\leq |S_1|+|S_2|\leq |X|\cdot|A|+\sum_{R\in\sigma}
|R^\X|\cdot|R^\A|=\gamma$. Consider a vector space $\mathcal{V}$ of dimension
$\gamma$ such that $\mathcal{U}\subseteq\mathcal{V}\subseteq \R^\ale$. Using the
Gram--Schmidt process,\footnote{We note that the Gram--Schmidt process also applies to vector spaces
of countably infinite dimension.}  we find a projection matrix ${Z}\in\cT^{\ale,\gamma}(\R)$ such that ${Z}^T{Z}=I_\gamma$ and ${Z}{Z}^T\bv=\bv$ for any $\bv\in\mathcal{V}$. Consider the family of vectors
\begin{align}
\label{eqn_1618_1104}
\begin{array}{lllll}
\blambda_{x,a}={Z}^T\xi(x)^T\be_a &x\in X,a\in A,\\
\blambda_{R,\bx,\ba}={Z}^TQ_{R,\bx}^T\be_\ba & R\in\sigma,\bx\in R^\X,\ba\in R^\A.
\end{array}
\end{align}
We claim that~\eqref{eqn_1618_1104} witnesses that $\SDP(\X,\A)=\YES$. To check $\vardiamond 1$, observe that
\begin{align*}
\sum_{a\in A}\|\blambda_{x,a}\|^2
&=
\sum_{a\in A}\be_a^T\xi(x){Z}{Z}^T\xi(x)^T\be_a
=
\sum_{a\in A}\be_a^T\xi(x)\xi(x)^T\be_a
=
\tr(\xi(x)\xi(x)^T)
=
1,
\end{align*}
where the second equality follows from the fact that $\xi(x)^T\be_{a}\in
S_1\subseteq \mathcal{U}\subseteq\mathcal{V}$ and the fourth from (C3). 
In a similar way, using (C2), we obtain
\begin{align*}
\blambda_{x,a}\cdot\blambda_{x,a'}
&=
\be_a^T\xi(x){Z}{Z}^T\xi(x)^T\be_{a'}
=
\be_a^T\xi(x)\xi(x)^T\be_{a'}
=
0,\\
\blambda_{R,\bx,\ba}\cdot\blambda_{R,\bx,\ba'}
&=
\be_\ba^TQ_{R,\bx}{Z}{Z}^TQ_{R,\bx}^T\be_{\ba'}
=
\be_\ba^TQ_{R,\bx}Q_{R,\bx}^T\be_{\ba'}
=
0
\end{align*}
if $a\neq a'\in A$ and $\ba\neq\ba'\in R^\A$. This shows that $\vardiamond 2$ and $\vardiamond 3$ hold. Finally, to prove $\vardiamond 4$, we observe that 
\begin{align*}
\sum_{\substack{\ba\in R^\A\\a_i=a}}\blambda_{R,\bx,\ba}
&=
\sum_{\substack{\ba\in R^\A\\a_i=a}}{Z}^TQ_{R,\bx}^T\be_\ba
=
\big(\sum_{\substack{\ba\in R^\A\\a_i=a}}\be_\ba^TQ_{R,\bx}{Z}\big)^T
=
\big(\be_a^TP_iQ_{R,\bx}{Z}\big)^T
=
\big(\be_a^T\xi(x_i){Z}\big)^T\\
&=
{Z}^T\xi(x_i)^T\be_a
=
\blambda_{x_i,a},
\end{align*}
where the third equality follows from Lemma~\ref{lem_basic_P_i}. Therefore, the claim is true and the proof is complete.
\end{proof}

\section{Hierarchies of minion tests}
\label{sec_hierarchies_of_minion_tests}

As we have seen in Section~\ref{sec_minion_tests}, minions give a systematic method for designing tests for $\parPCSPs$. We now describe a construction, which we call \emph{tensorisation}, that provides a technique to systematically refine minion tests and turning them into algorithmic \textit{hierarchies}.

Let $S$ be a set and let $k\in\N$. Recall that, for a tuple $\bn=(n_1,\dots,n_k)\in\N^k$, $\cT^{\bn}(S)$ denotes the set of all functions from $[n_1]\times\dots\times [n_k]$ to $S$, which we visualise as hypermatrices or tensors. 
Furthermore, given a signature $\sigma$, we denote by $\sigma^{\tensor{k}}$ the signature consisting of the same symbols as $\sigma$ such that each symbol $R$ of arity $r$ in $\sigma$ has arity $r^k$ in $\sigma^{\tensor{k}}$. 

\begin{graytbox}
\begin{defn}
\label{defn_tensorisation}
The \emph{$k$-th tensor power} of a $\sigma$-structure $\A$ is the
  $\sigma^{\tensor{k}}$-structure $\A^{\tensor{k}}$ having domain $A^k$ and
  relations defined as follows: For each symbol $R\in\sigma$ of arity $r$ in $\sigma$, we set $R^{\A^\tensor{k}}=\left\{\ba^{\tensor{k}}:\ba\in R^\A\right\}$,
where, for $\ba\in R^\A$, $\ba^\tensor{k}$ is the tensor in $\cT^{r\cdot\bone_k}(A^k)$ defined as follows: For any $\bi\in [r]^k$, the $\bi$-th  element of $\ba^\tensor{k}$ is $\ba_\bi$.
\end{defn}
\end{graytbox}
In other words, it holds that $E_\bi\ast\ba^\tensor{k}=\ba_\bi$ for any $\bi\in [r]^k$. Note that $\ba^\tensor{k}$ can be visualised as the formal Segre outer product of $k$ copies of $\ba$ (cf.~\cite{lim2013tensors}).

It is easy to check that $\A^\tensor{1}=\A$. Moreover, the function $R^\A\to R^{\A^\tensor{k}}$ given by $\ba\mapsto\ba^\tensor{k}$ is a bijection, so the cardinality of $R^{\A^\tensor{k}}$ equals the cardinality of $R^\A$.
\begin{example}
\label{example_tensorisation_1406}
Let us describe the third tensor power of the $3$-clique -- i.e., the structure $\K_3^\tensor{3}$. The domain of $\K_3^\tensor{3}$ is $[3]^3$, i.e., the set of tuples of elements in $[3]$ having length $3$. Let $R$ be the symbol corresponding to the binary edge relation in $\K_3$, so that $R^{\K_3}=\{(1,2),(2,1),(2,3),(3,2),(3,1),(1,3)\}$. Then, $R^{\K_3^\tensor{3}}$ has arity $2^3=8$ and it is a subset of $\mathcal{T}^{(2,2,2)}([3]^3)$. 
Specifically, $R^{\K_3^\tensor{3}}=\{(1,2)^\tensor{3},(2,1)^\tensor{3},(2,3)^\tensor{3},(3,2)^\tensor{3},(3,1)^\tensor{3},(1,3)^\tensor{3}\}$ where, e.g., 
\[
(2,3)^\tensor{3}=\left[\begin{array}{@{}cc|cc@{}}
(2,2,2)&(2,2,3)&(3,2,2)&(3,2,3)\\
(2,3,2)&(2,3,3)&(3,3,2)&(3,3,3)
\end{array}\right].
\]
Here, the vertical line separates the two $2\times 2$ layers of the $2\times 2\times 2$ tensor: The left block contains the layer whose entries have first coordinate $1$, while the right block contains the layer whose entries have first coordinate $2$.
\end{example}

Recall that a $\sigma$-structure $\A$ is \emph{$k$-enhanced} if $\sigma$ contains a $k$-ary symbol $R_k$ such that $R_k^\A=A^k$. Observe that any two $\sigma$-structures $\A$ and $\B$ are homomorphic if and only if the structures $\tilde{\A}$ and $\tilde{\B}$ obtained by adding $R_k$ to their signatures are homomorphic (thus, $\PCSP(\A,\B)$ is equivalent to $\PCSP(\tilde{\A},\tilde{\B})$).
We now give the main definition of this work.

\begin{graytbox}	
\begin{defn}	
\label{defn_hierarchy_minion_test}	
For a minion $\Mminion$ and an integer $k\in\N$, the \emph{$k$-th level of the	
  minion test} $\Test{\Mminion}{}$, denoted by $\Test{\Mminion}{k}$, is the	
  decision problem defined as follows: Given two $k$-enhanced	
  $\sigma$-structures $\X$ and $\A$, return $\YES$ if $\Xk\to\freeM(\Ak)$, and $\NO$ otherwise.	
\end{defn}	
\end{graytbox}	
We point out that, in the definition above, it is crucial to require that the structures $\X$ and $\A$ be $k$-enhanced in order for the test to capture the hierarchies of (P)CSP relaxations studied in the literature, as we shall see later. Such hierarchies involve marginalisation constraints over sets of variables of possibly different sizes (see, for example, the condition ($\varheart 2$) in the equivalent description~\eqref{eqn_defn_Sherali_Adams_BD} of the Sherali--Adams hierarchy in Appendix~\ref{subsec_appendix_SA_AIPk}). 
Our goal is then to simulate these constraints in the homomorphism $\Xk\to\freeM(\Ak)$ of $\sigma^{\tensor{k}}$-structures witnessing the truth of the test. To that end, we ask that the homomorphism should preserve an extra relation containing \textit{all} possible tuples of variables of length $k$---which is precisely the relation $R_k$.

Comparing Definition~\ref{defn_hierarchy_minion_test} with Definition~\ref{defn_minion_test}, we see that $\Test{\Mminion}{k}(\X,\A)=\Test{\Mminion}{}(\Xk,\Ak)$. In other words, the $k$-th level of a minion test is just the minion test applied to the tensor power of the ($k$-enhanced) structures. 
We have seen (cf.~Proposition~\ref{prop_minion_tests_are_complete}) that a minion test is always complete. It turns out that this property keeps holding for any level of a minion test. 
\begin{prop}
\label{prop_hierarchies_are_complete}
$\Test{\Mminion}{k}$ is complete for any minion $\Mminion$ and any integer $k\in\N$.
\end{prop}
The proof of Proposition~\ref{prop_hierarchies_are_complete} relies on the fact that homomorphisms between structures are in some sense invariant under the tensorisation construction, as formally stated in the next proposition.
We let $\Hom(\A,\B)$ denote the set of homomorphisms from $\A$ to $\B$.

\begin{prop}
\label{prop_correspondence_morphisms_tensor_structures}
Let $k\in \N$ and let $\A,\B$ be two $\sigma$-structures. Then
\begin{enumerate}
\item[$(i)$]
 $\A\to\B$ if and only if $\A^\tensor{k}\to\B^\tensor{k}$;
\item[$(ii)$]
if $\A$ is $k$-enhanced, there is a bijection $\rho:\Hom(\A,\B)\to\Hom(\A^\tensor{k},\B^\tensor{k})$.
\end{enumerate}
\end{prop}

\begin{proof}
Let ${f}:\A\to\B$ be a homomorphism, and consider the function $f^\ast:A^k\to B^k$ defined by $f^\ast((a_1,\dots,a_k))= ({f}(a_1),\dots,{f}(a_k))$.\footnote{Note that, throughout the paper, we use the same symbol $f$ to denote both a function $a\mapsto f(a)$ for $a\in A$ and its component-wise application $\ba\mapsto f(\ba)=(f(a_1),\dots,f(a_p))$ for $\ba\in A^p$. Only in this proof, however, it is convenient to introduce the new symbol $f^*$ to specifically denote the component-wise application of $f$ to tuples of length $k$.} Take $R\in\sigma$ of arity $r$, and consider $\ba^\tensor{k}\in R^{\A^\tensor{k}}$, where $\ba\in R^\A$. Let $\bb=f(\ba)$.
Since ${f}$ is a homomorphism, $\bb\in R^\B$, so $\bb^\tensor{k}\in R^{\B^\tensor{k}}$. For any $\bi\in [r]^k$, we have 
\begin{align*}
E_\bi\ast f^\ast\left(\ba^\tensor{k}\right)&=f^\ast\left(E_\bi\ast\ba^\tensor{k}\right)=f^\ast(\ba_\bi)=\bb_\bi=E_\bi\ast\bb^\tensor{k},
\end{align*}
which yields $f^\ast(\ba^\tensor{k})=\bb^\tensor{k}\in R^{\B^\tensor{k}}$. Hence, $f^\ast:\A^\tensor{k}\to\B^\tensor{k}$ is a homomorphism.

Conversely, let $g:\A^\tensor{k}\to\B^\tensor{k}$ be a homomorphism. We define the function $g_\ast:A\to B$ by setting $g_\ast(a)=\be_1^Tg((a,\dots,a))$ for each $a\in A$. Take $R\in\sigma$ of arity $r$, and consider a tuple $\ba=(a_1,\dots,a_r)\in R^\A$. Since $\ba^\tensor{k}\in R^{\A^\tensor{k}}$ and $g$ is a homomorphism, we have that $g(\ba^\tensor{k})\in R^{\B^\tensor{k}}$. Therefore, $g(\ba^\tensor{k})=\bb^\tensor{k}$ for some $\bb=(b_1,\dots,b_r)\in R^\B$. For each $j\in [r]$, consider the tuple $\bi=(j,\dots,j)\in [r]^k$ and observe that
\begin{align*}
g((a_j,\dots,a_j))=g(\ba_\bi)=g\left(E_\bi\ast\ba^\tensor{k}\right)=E_\bi\ast g(\ba^\tensor{k})
=
E_\bi\ast\bb^\tensor{k}
=\bb_\bi
=
(b_j,\dots,b_j).
\end{align*}
Hence, we find
\begin{align*}
g_\ast(\ba)=\left(\be_1^Tg((a_1,\dots,a_1)),\dots,\be_1^Tg((a_r,\dots,a_r))\right)
=
\left(b_1,\dots,b_r\right)=\bb\in R^\B.
\end{align*}
Therefore, $g_\ast:\A\to\B$ is a homomorphism. This concludes the proof of $(i)$.

To prove $(ii)$, observe first that, if $\A\not\to\B$, then $\Hom(\A,\B)=\Hom(\A^\tensor{k},\B^\tensor{k})=\emptyset$, so there is a trivial bijection in this case. If $\A\to\B$, consider the map $\rho:\Hom(\A,\B)\to\Hom(\A^\tensor{k},\B^\tensor{k})$ defined by $f\mapsto f^\ast$ and the map $\rho':\Hom(\A^\tensor{k},\B^\tensor{k})\to\Hom(\A,\B)$ defined by $g\mapsto g_\ast$. For $f:\A\to\B$ and $a\in A$, we have 
\begin{align*}
(f^\ast)_\ast(a)
=
\be_1^Tf^\ast((a,\dots,a))
=
\be_1^T(f(a),\dots,f(a))=f(a)
\end{align*}
so that $\rho'\circ\rho=\id_{\Hom(\A,\B)}$. Consider now $g:\A^\tensor{k}\to\B^\tensor{k}$, and take $\ba=(a_1,\dots,a_k)\in A^k$. Using the assumption that $\A$ is $k$-enhanced, we have $\ba\in R_k^\A$, which implies $\ba^\tensor{k}\in R_k^{\A^\tensor{k}}$. Hence, $g(\ba^\tensor{k})\in R_k^{\B^\tensor{k}}$, so $g(\ba^\tensor{k})=\bb^\tensor{k}$ for some $\bb=(b_1,\dots,b_k)\in R_k^\B\subseteq B^k$. For $j\in [k]$ and $\bi=(j,\dots,j)\in [k]^k$, we have
\begin{align*}
g((a_j,\dots,a_j))
=
g(\ba_\bi)
=
g\left(E_\bi\ast\ba^\tensor{k}\right)
=
E_\bi\ast g\left(\ba^\tensor{k}\right)
=
E_\bi\ast\bb^\tensor{k}
=
\bb_\bi
=
(b_j,\dots,b_j).
\end{align*}
Letting $\bi'=(1,\dots,k)\in [k]^k$, we obtain
\begin{align*}
(g_\ast)^\ast(\ba)
&=
(g_\ast(a_1),\dots,g_\ast(a_k))
=
\left(\be_1^Tg((a_1,\dots,a_1)),\dots,\be_1^Tg((a_k,\dots,a_k))\right)
=
(b_1,\dots,b_k)\\
&=
\bb
=
\bb_{\bi'}
=
E_{\bi'}\ast\bb^\tensor{k}
=
E_{\bi'}\ast g\left(\ba^\tensor{k}\right)
=
g\left(E_{\bi'}\ast\ba^\tensor{k}\right)
=
g(\ba_{\bi'})
=
g(\ba),
\end{align*}
so that $\rho\circ\rho'=\id_{\Hom(\A^\tensor{k},\B^\tensor{k})}$, which concludes the proof of $(ii)$. 
\end{proof}

\begin{rem}
Part $(ii)$ of Proposition~\ref{prop_correspondence_morphisms_tensor_structures} does not hold if we relax the requirement that $\A$ be $k$-enhanced. More precisely, in this case, the function $\rho:\Hom(\A,\B)\to\Hom(\A^\tensor{k},\B^\tensor{k})$ defined in the proof of Proposition~\ref{prop_correspondence_morphisms_tensor_structures} still needs to be injective, but may not be surjective. Therefore, we have $\left |\Hom(\A,\B)\right |\leq\left |\Hom(\A^\tensor{k},\B^\tensor{k})\right |$, and the inequality may be strict.

For example, consider the Boolean structure $\A$ having a single unary relation $R_1^\A=A=\{0,1\}$. So, $\A$ is $1$-enhanced but not $2$-enhanced. Observe that $|\Hom(\A,\A)|=4$. The tensorised structure $\A^\tensor{2}$ has domain $\{0,1\}^2=\{(0,0),(0,1),\allowbreak(1,0),(1,1)\}$, and its (unary) relation is $R_1^{\A^\tensor{2}}=\{0^\tensor{2},1^\tensor{2}\}=\{(0,0),(1,1)\}$. Therefore, each map $f:\{0,1\}^2\to\{0,1\}^2$ such that $f((0,0))\in\{(0,0),(1,1)\}$ and $f((1,1))\in\{(0,0),\allowbreak(1,1)\}$ yields a proper homomorphism $\A^\tensor{2}\to\A^\tensor{2}$. It follows that $\left|\Hom(\A^\tensor{2},\A^\tensor{2})\right|=64$, so $\Hom(\A,\A)$ and $\Hom(\A^\tensor{2},\A^\tensor{2})$ are not in bijection.
\end{rem}

It readily follows from Proposition~\ref{prop_correspondence_morphisms_tensor_structures} that hierarchies of minion tests are always complete.

\begin{proof}[Proof of Proposition~\ref{prop_hierarchies_are_complete}]
Let $\X$ and $\A$ be two $k$-enhanced $\sigma$-structures and suppose that $\X\to\A$. Proposition~\ref{prop_correspondence_morphisms_tensor_structures} yields $\Xk\to\Ak$, while Lemma~\ref{lem_A_maps_free_structure_of_A} yields $\Ak\to\freeM(\Ak)$. The composition of the two homomorphisms witnesses that $\Test{\Mminion}{k}(\X,\A)=\YES$, as required.
\end{proof}

Before continuing, we now discuss some of the basic algebraic properties of the tensorisation construction of Definition~\ref{defn_tensorisation}.
\begin{rem}
We now recall the notion of \textit{pp-power} of relational structures, following the presentation in~\cite{BBKO21}.
A \textit{pp-formula} over a signature $\sigma$ is a formal expression $\psi$ consisting of an existentially quantified conjunction of predicates of the form $(i)$ ``$x=y$'', or $(ii)$ ``$(x_{i_1},\dots,x_{i_r})\in R$'' for some $R\in\sigma$ of arity $r$, where $x,y,x_{i_1},\dots,x_{i_r}$ are variables. 
Let $f$ be the number of free (i.e., unquantified) variables in $\psi$.
Given a $\sigma$-structure $\A$, the interpretation of $\psi$ in $\A$ is the set $\psi(\A)\subseteq A^f$ containing all tuples $(a_1,\dots,a_f)\in A^f$ that  satisfy $\psi$, where each symbol $R$ appearing in $\psi$ is interpreted in $\A$. Let now $\A'$ be a $\sigma'$-structure for some possibly different signature $\sigma'$, such that $A'=A$. We say that $\A'$ is \textit{pp-definable} from $\A$ if for each symbol $S\in\sigma'$ it holds that $S^{\A'}=\psi_S(\A)$ for some pp-formula $\psi_S$ over $\sigma$. Suppose now that $A'=A^m$ for some $m\in\N$, and
let $\operatorname{vec}_m(\A')$ be the structure with domain $A$ and relations defined as follows: For each $r$-ary symbol $S\in\sigma'$, $\operatorname{vec}_m(\A')$ has an $rm$-ary relation containing the tuple $(b^{(1)}_1,\dots,b^{(1)}_m,\dots,b^{(r)}_1,\dots,b^{(r)}_m)$ for each tuple $(\bb^{(1)},\dots,\bb^{(r)})\in S^{\A'}$.
We say that $\A'$ is an $m$-fold \textit{pp-power} of $\A$ if $\operatorname{vec}_m(\A')$ is pp-definable from $\A$.

It is not hard to verify that the $k$-th tensor power $\Ak$ of a $\sigma$-structure $\A$ is a specific ($k$-fold) pp-power of $\A$. Indeed, consider the structure $\operatorname{vec}_k(\Ak)$, and notice that the relation corresponding to a symbol $R\in\sigma$ (having arity $r$ in $\sigma$ and arity $r^k$ in $\sigma^\tensor{k}$) has arity $k\cdot r^k$ in $\operatorname{vec}_k(\Ak)$. Such relation is easily seen to be pp-definable from $\A$. As an example, suppose for concreteness that $k=2$ and $\sigma$ contains a single ternary symbol $R$. Then, we have
\begin{align*}
    R^{\A^\tensor{2}}&=\left\{\left[\begin{array}{ccc}
         (a_1,a_1)&(a_1,a_2)&(a_1,a_3)  \\
         (a_2,a_1)&(a_2,a_2)&(a_2,a_3)  \\
         (a_3,a_1)&(a_3,a_2)&(a_3,a_3)  \\
    \end{array}\right]:(a_1,a_2,a_3)\in R^\A\right\}\quad\mbox{and}\\
    R^{\operatorname{vec}_2(\A^\tensor{2})}&=\{(a_1,a_1,a_1,a_2,a_1,a_3,a_2,a_1,a_2,a_2,a_2,a_3,a_3,a_1,a_3,a_2,a_3,a_3):(a_1,a_2,a_3)\in R^\A\}.
\end{align*}
Clearly, there exists a pp-formula $\psi$ over $\sigma$ such that the $18$-ary relation of $\operatorname{vec}_2(\A^\tensor{2})$ (where $18=2\cdot 3^2$) satisfies $R^{\operatorname{vec}_2(\A^\tensor{2})}=\psi(\A)$. This easily generalises to arbitrary structures and arbitrary powers $k$.
In the general case, the pp-formula $\psi_R$ satisfying $R^{\operatorname{vec}_k(\Ak)}=\psi_R(\A)$ (for some symbol $R$ of arity $r$ in $\sigma$) is
\begin{align*}
    \psi_R(x_1,\dots,x_{k\cdot r^k})=\exists_{z_1,\dots,z_r}\;\;(z_1,\dots,z_r)\in R\;\;\wedge\;\;\bigwedge_{\bi\in [r]^k}\bigwedge_{j\in [k]}x_{k\cdot (p(\bi)-1)+j}=z_{i_j},
\end{align*}
where $p$ is a fixed bijection from $[r]^k$ to $[r^k]$.

It also follows that, given a PCSP template $(\A,\B)$, the PCSP template $(\Ak,\Bk)$ is a pp-power of  $(\A,\B)$ in the sense of~\cite[Definition~4.7]{BBKO21}---i.e., $\operatorname{vec}_k(\Ak)$ and $\operatorname{vec}_k(\Bk)$ are definable from $\A$ and $\B$, respectively, using the \textit{same} pp-formulas. (Note that $(\Ak,\Bk)$ is a proper PCSP template by virtue of part $(i)$ of  Proposition~\ref{prop_correspondence_morphisms_tensor_structures}.) It then immediately follows from~\cite[Corollary~4.10]{BBKO21} that there exists a minion homomorphism $\Pol(\A,\B)\to\Pol(\Ak,\Bk)$. In fact, it was proved in~\cite[Theorem~37]{cz22arxiv:tensors} that $\Pol(\A,\B)$ and $\Pol(\Ak,\Bk)$ are actually homomorphically equivalent (and even isomorphic when $\A$ is $k$-enhanced).
\end{rem}

Recall now the description of the free structure of a minion $\Mminion$ given in Section~\ref{subsec_prelimns_pcsps}. In the case that the free structure is applied to the $k$-th tensor power $\Ak$ of a $\sigma$-structure $\A$, it follows from Definition~\ref{defn_tensorisation} that the domain of $\freeM(\Ak)$ is the set $\Mminion^{(n^k)}$ (where, as usual, we are denoting $n=|A|$). Furthermore, given a
relation symbol $R\in\sigma$ whose arity in $\sigma$ is $r$, a tuple $(M_\bi)_{\bi\in [r]^k}$ of elements
of $\Mminion^{(n^k)}$ belongs to $R^{\bF_{\sM}(\Ak)}$ if and only if there is some
$Q\in \sM^{(|R^\A|)}$ such that $M_\bi=Q_{/\pi_\bi}$ for each $\bi\in[r]^k$, where
$\pi_\bi:R^\A\to A^k$ maps $\ba\in R^\A$ to its $\bi$-th projection $\ba_\bi$.

It is well known that each of the hierarchies of relaxations described in Section~\ref{sec_relaxations_hierarchies} has the property that higher levels are at least as powerful as lower levels. As the next result shows, this is in fact a property of all hierarchies of minion tests.
\begin{prop}
\label{prop_hierarchies_minion_tests_get_stronger}
Let $\Mminion$ be a minion, let $k,p\in\N$ be such that $k>p$, and let $\X,\A$ be two $k$- and $p$-enhanced $\sigma$-structures. If $\Test{\Mminion}{k}(\X,\A)=\YES$ then $\Test{\Mminion}{p}(\X,\A)=\YES$.
\end{prop} 
\begin{proof}
Let $\xi:\Xk\to\freeM(\Ak)$ be a homomorphism witnessing that
  $\Test{\Mminion}{k}(\X,\A)=\YES$. 
Since $k>p$, we can choose two tuples
  $\bv\in [k]^p$ and $\bw\in [p]^k$ such that $\bw_\bv=(1,\dots,p)$. (For instance, we may
  take $\bv=(1,\ldots,p)$ and
  $\bw=(1,\ldots,p,\ldots,p)$.) Consider the map $\tau:A^k\to A^p$ defined by $\ba\mapsto \ba_\bv$. We claim that the map $\vartheta:X^p\to\Mminion^{(n^p)}$ defined by $\bx\mapsto \xi(\bx_\bw)_{/\tau}$ yields a homomorphism from $\X^\tensor{p}$ to $\freeM(\A^\tensor{p})$, thus witnessing that $\Test{\Mminion}{p}(\X,\A)=\YES$. To that end, for $R\in\sigma$, take $\bx\in R^\X$ and observe that, since $\xi$ is a homomorphism and $\bx^\tensor{k}\in R^{\Xk}$, $\xi(\bx^\tensor{k})\in R^{\freeM(\Ak)}$. Therefore, there exists $Q\in\Mminion^{(|R^\A|)}$ satisfying $\xi(\bx_\bi)=Q_{/\pi_\bi}$ for each $\bi\in [r]^k$. If we manage to show that 
$\vartheta(\bx_\bj)=Q_{/\pi_\bj}$
for each $\bj\in [r]^p$,
we would deduce that $\vartheta
(\bx^\tensor{p})\in R^{\freeM(\A^\tensor{p})}$, thus proving the claim. Observe that 
\begin{align*}
\vartheta(\bx_\bj)
&=
\xi(\bx_{\bj_{\bw}})_{/\tau}
=
(Q_{/\pi_{\bj_\bw}})_{/\tau}
=
Q_{/\tau\,\circ\,\pi_{\bj_\bw}},
\end{align*}
so we are left to show that $\tau\circ\pi_{\bj_\bw}=\pi_\bj$. Indeed, given any $\ba\in R^\A$,
\begin{align*}
(\tau\circ\pi_{\bj_\bw})(\ba)
&=
\tau(\pi_{\bj_\bw}(\ba))
=
\tau(\ba_{\bj_\bw})
=
\ba_{\bj_{\bw_\bv}}
=
\ba_\bj
=
\pi_\bj(\ba),
\end{align*}
as required.
\end{proof}
\noindent It follows from Proposition~\ref{prop_hierarchies_minion_tests_get_stronger} that, if some level of a minion test is sound for a template $(\A,\B)$ (equivalently, if it solves $\PCSP(\A,\B)$), then any higher level is sound for $(\A,\B)$ (equivalently, it solves $\PCSP(\A,\B)$).

The next theorem shows that the framework defined above is general enough to capture each of the five hierarchies for $\parPCSPs$ described in Section~\ref{sec_relaxations_hierarchies}.
\begin{graytbox}	
\begin{thm}[informal]
\label{thm_main_multilinear_tests}
If $k\in\N$ is at least the maximum arity of the template,
\begin{align*}
    \begin{array}{lll}
         \bullet\;\BW^k=\Test{\Hminion}{k}\\[4pt]
         \bullet\;\SA^k=\Test{\Qconv}{k}\\[4pt]
         \bullet\;\AIP^k=\Test{\Zaff}{k}\\[4pt]
         \bullet\;\SoS^k=\Test{\Sminion}{k}\\[4pt]
         \bullet\;\BA^k=\Test{\Mblpaip}{k}.
    \end{array}
\end{align*}
\end{thm}
\end{graytbox}

We will prove Theorem~\ref{thm_main_multilinear_tests} in Section~\ref{sec_proof_of_main_thm}. First, it will be convenient to focus on hierarchies of tests corresponding to minions having specific characteristics---which we call linear and conic.

\section{Linear minions}
\label{sec_hierarchies_of_linear_minion_tests}
By the Definition~\ref{defn_hierarchy_minion_test} of a hierarchy of minion tests, $\Test{\Mminion}{k}$ applied to an instance $\X$ of $\PCSP(\A,\B)$ checks for the existence of a homomorphism from $\Xk$ to $\freeM(\Ak)$. Therefore, to describe the hierarchy and get knowledge on its functioning it is necessary to study the structure $\freeM(\Ak)$.
Certain features of the hierarchies of minion tests -- in particular, that they
are complete (Proposition~\ref{prop_hierarchies_are_complete}) and their power does not decrease while the level increases (Proposition~\ref{prop_hierarchies_minion_tests_get_stronger}) -- hold
true for any minion, as they only depend on basic properties of the
tensorisation construction. In order to prove
Theorem~\ref{thm_main_multilinear_tests}, however, it is necessary to dig
deeper by investigating how the tensorisation construction interacts with the free structure. 
To that end, we isolate a property shared by all minions mentioned in this work: Their objects can be interpreted as matrices, and their minor operations can be expressed as matrix multiplications. We call such minions 
\emph{linear}. 

\begin{defn}
\label{defn_linear_minion}
A minion $\Mminion$ is \emph{linear} if there exists a semiring
  $\mathcal{S}$
 with additive identity $0_{\mathcal{S}}$ and multiplicative identity $1_{\mathcal{S}}$ and a number
${d}\in\N\cup\{\ale\}$ (called \emph{depth}) such that
\begin{enumerate}
\item
the elements of $\Mminion^{(L)}$ are $L\times d$ matrices whose entries belong to $\mathcal{S}$, for each $L\in\N$;
\item
given $L,L'\in\N$, $\pi:[L]\to [L']$, and $M\in\Mminion^{(L)}$, $M_{/\pi}=PM$,
    where $P$ is the $L'\times L$ matrix such that, for $i\in [L']$ and $j\in
    [L]$, the $(i,j)$-th entry of $P$ is $1_{\mathcal{S}}$ if $\pi(j)=i$, and $0_{\mathcal{S}}$ otherwise.
\end{enumerate}
\end{defn}
Observe that pre-multiplying a matrix $M$ by $P$ amounts to performing a combination of the following three elementary operations to the rows of $M$: swapping two rows, replacing two rows with their sum, and inserting a zero row. Hence, we may equivalently define a linear minion as a collection of matrices over $\mathcal{S}$ that is closed under such elementary operations. 

\begin{rem}
\label{rem_linear_to_function_minion}
We now show that any linear minion can be naturally interpreted as a minion of functions. Given two (potentially infinite) sets $A,B$ and an integer $L\in\N$, let $\mathscr{F}_{A,B}^{(L)}$ be the set of all functions $f:A^L\to B$.  For $\pi:[L]\to[L']$, we define $f_{/\pi}\in\mathscr{F}_{A,B}^{(L')}$ to be the function given by 
\[
(a_1,\dots,a_{L'})\mapsto f(a_{\pi(1)},\dots,a_{\pi(L)}).
\]
It is easy to verify that the disjoint union $\mathscr{F}_{A,B}=\bigsqcup_{L\in\N}\mathscr{F}_{A,B}^{(L)}$ equipped with such minor operations is a minion. We let a \textit{function minion} over $A,B$ be any subminion of $\mathscr{F}_{A,B}$ (i.e., a non-empty subset of $\mathscr{F}_{A,B}$ that is closed under the minor operations)~\cite{BBKO21}. For example, the polymorphism minion $\Pol(\A,\B)$ of a PCSP template $(\A,\B)$ described in Example~\ref{example_minion_pol_A_B} is a function minion over the sets $A,B$. 
Now, given a linear minion $\Mminion$ of depth $d$ over a semiring $\mathcal{S}$, we can naturally see $\Mminion$ as a function minion over the sets $\mathcal{S},\mathcal{S}^d$ as follows. Consider the map $\xi:\Mminion\to\mathscr{F}_{\mathcal{S},\mathcal{S}^d}$ that associates with a matrix $M\in\Mminion^{(L)}$ the linear operator $\xi(M)$ from $\mathcal{S}^L$ to $\mathcal{S}^d$ corresponding to the matrix $M^T$.
It is not hard to verify that $\xi$ preserves arities and minors, and it is thus a minion homomorphism.
Indeed, given a map $\pi:[L]\to[L']$ and a tuple $\bs\in\mathcal S^{L'}$, and letting $P$ be the $L'\times L$ Boolean matrix associated with $\pi$ as per Definition~\ref{defn_linear_minion}, we have
\begin{align*}
    \xi(M_{/\pi})(\bs)
    =
    (M_{/\pi})^T\bs=
    M^TP^T\bs
    =
    \xi(M)(P^T\bs)
    =
    \xi(M)_{/\pi}(\bs).
\end{align*}
Furthermore, $\xi$ is injective,
so it induces a minion isomorphism from $\Mminion$ to a subminion of $\mathscr{F}_{\mathcal{S},\mathcal{S}^d}$. As a consequence, $\xi$ witnesses that $\Mminion$ is isomorphic to
a function minion.
\end{rem}

As illustrated by the next proposition, the family of linear minions is rich enough to include the minions associated with all minion tests studied in the literature on $\PCSP$s, including $\SDP$.

\begin{prop}
\label{prop_examples_linear_minions}
The following minions are (isomorphic to) linear minions:
\begin{align*}
&\bullet\;\;\Hminion, \mbox{ with } \mathcal{S}=(\{0,1\},\vee,\wedge) \mbox{ and } {d}=1
&&
\bullet\;\;\Qconv, \mbox{ with } \mathcal{S}=\Q \mbox{ and } {d}=1
\\
&
\bullet\;\;\Zaff, \mbox{ with } \mathcal{S}=\Z \mbox{ and } {d}=1
&&
\bullet\;\;\Sminion, \mbox{ with } \mathcal{S}=\R \mbox{ and } {d}=\ale
\\
&\bullet\;\;\Mblpaip, \mbox{ with } \mathcal{S}=\Q \mbox{ and } {d}=2.
\end{align*}
\end{prop}

\begin{proof}
The result for $\Qconv$, $\Zaff$, and $\Mblpaip$ directly follows from their definitions in Example~\ref{examples_famous_minions}, while the result for $\Sminion$ is clear from Definition~\ref{defn_minion_SDP}.

We now turn to $\Hminion$. Recall that, in Example~\ref{examples_famous_minions}, we described $\Hminion$ as a set of functions rather than a set of matrices. We now prove that $\Hminion$ is isomorphic to a linear minion. Given $L\in\N$ and $\emptyset\neq
  {Z}\subseteq [L]$, we identify the Boolean function $f_{Z}=\bigwedge_{{z}\in
  {Z}}x_{z}\in\Hminion^{(L)}$ with the indicator vector $\bv_{Z}\in \{0,1\}^L$ whose
  $i$-th entry, for $i\in [L]$, is $1$ if $i\in {Z}$, and $0$ otherwise. To conclude, we need to show that, under this identification, the minor operations of $\Hminion$ correspond to the minor operations given in Definition~\ref{defn_linear_minion}. In other words, we claim that the function ${f_{Z}}_{/\pi}$ corresponds to the vector $P\bv_{Z}$ for any $L'\in\N$ and any $\pi:[L]\to [L']$, where $P$ is the $L'\times L$ matrix described in Definition~\ref{defn_linear_minion}. First, observe that
\begin{align*}
{f_{Z}}_{/\pi}(x_1,\dots,x_{L'})
&=
f_{Z}(x_{\pi(1)},\dots,x_{\pi(L)})
=
\bigwedge_{{z}\in {Z}}x_{\pi({z})}
=
\bigwedge_{t\in \pi({Z})}x_{t}
=
f_{\pi({Z})}(x_1,\dots,x_{L'}),
\end{align*}
so ${f_{Z}}_{/\pi}=f_{\pi({Z})}$. To conclude, we need to show that $P\bv_{Z}=\bv_{\pi({Z})}$, where $\bv_{\pi({Z})}$ is the indicator vector of the nonempty set $\pi({Z})\subseteq [L']$. Notice that the matrix multiplication is performed in the semiring $\mathcal{S}=(\{0,1\},\vee,\wedge)$.
For any $i\in [L']$, we have
\begin{align*}
\be_i^TP\bv_{Z}
&=
\bigvee_{j\in [L]}\left((\be_i^TP\be_j)\wedge(\be_j^T\bv_{Z})\right)
=
\bigvee_{\substack{j\in {Z}\\\pi(j)=i}}1
=
\bigvee_{i\in \pi({Z})}1
=
\be_i^T\bv_{\pi({Z})},
\end{align*}
as required.
\end{proof}

As a consequence, the machinery we build in this section (and in Section~\ref{sec_hierarchies_of_conic_minion_tests}, where we consider an even more specialised minion class) shall be crucial to show that the framework of hierarchies of minion tests captures all hierarchies of relaxations in Theorem~\ref{thm_main_multilinear_tests}.

Recall that, as per Definition~\ref{defn_minion_test}, the minion test associated with a minion $\Mminion$ works by checking whether a given instance is homomorphic to the free structure of $\Mminion$; in other words, $\Test{\Mminion}{}$ for a template $(\A,\B)$ is $\CSP(\freeM(\A))$. It is then worth checking what the latter object looks like in the case that $\Mminion$ is linear. The next remark shows that, in this case, $\freeM(\A)$ has a simple matrix-theoretic description.

\begin{rem}
\label{rem_free_structure_linear_minion}
Given a linear minion $\Mminion$ with semiring $\mathcal{S}$ and depth $d$, and a $\sigma$-structure $\A$, the free structure $\mathbb{F}_{\Mminion}(\A)$ of $\Mminion$ generated by $\A$ has the following description:
\begin{itemize}
\item
The elements of its domain $\Mminion^{(|A|)}$ are $|A|\times d$ matrices having entries in $\mathcal{S}$.
\item
For $R\in\sigma$ of arity $r$, the elements of $R^{\mathbb{F}_{\Mminion}(\A)}$ are tuples of the form $(P_1{Q},\dots,P_r{Q})$, where ${Q}\in\Mminion^{(|R^\A|)}$ is a $|R^\A|\times d$ matrix having entries in $\mathcal{S}$ and, for $i\in [r]$, $P_i$ is the $|A|\times |R^\A|$ matrix whose $(a,\ba)$-th entry is $1_{\mathcal{S}}$ if $a_i=a$, and $0_{\mathcal{S}}$ otherwise.
\end{itemize}
\end{rem}

\subsection{Multilinear tests}

We say that a test is \emph{multilinear} if it can be expressed as $\Test{\Mminion}{k}$ for some linear minion $\Mminion$ and some integer $k$. In the same way as, for a template $(\A,\B)$, $\Test{\Mminion}{}$ is $\CSP(\freeM(\A))$, it follows from Definition~\ref{defn_hierarchy_minion_test} that $\Test{\Mminion}{k}$ corresponds to $\CSP(\freeM(\Ak))$, as it checks for the existence of a homomorphism between the tensor power of the instance and the free structure of $\Mminion$ generated by the tensor power of $\A$. (However, recall that $\Test{\Mminion}{k}$ requires that $\X$ and $\A$ be $k$-enhanced, as per Definition~\ref{defn_hierarchy_minion_test}.)
As we have seen in Remark~\ref{rem_free_structure_linear_minion}, when $\Mminion$ is linear, the structure $\freeM(\A)$ consists in a space of matrices with relations defined through specific matrix products. Similarly, we now show that 
$\freeM(\Ak)$ is a space of tensors, endowed with relations that can be described through the tensor contraction operation. 

Given a semiring $\mathcal{S}$, a symbol $R\in\sigma$ of arity $r$, and a tuple $\bi\in [r]^k$, consider the tensor $
P_\bi\in \cT^{n\cdot\bone_k,|R^\A|}(\mathcal{S})
$ 
defined by 
\begin{align}
\label{defn_tensor_Pi}
E_\ba\ast P_\bi\ast E_{\ba'}=
\left\{
\begin{array}{cl}
1 & \mbox{if }\ba'_\bi=\ba\\
0 & \mbox{otherwise}
\end{array}
\right.
 \hspace{1cm}
\forall \ba\in {A}^k, \ba'\in R^\A.
\end{align}
Observe that the tensor $P_\bi$ is the multilinear equivalent of the matrix $P_i$ from Remark~\ref{rem_free_structure_linear_minion}. We point out that, just like $P_i$, the tensor $P_\bi$ depends on the symbol $R$. We leave this dependence implicit to avoid introducing additional notation; the symbol $R$ shall always be clear from the context. We also observe that, in the expression~\eqref{defn_tensor_Pi}, the tuples $\ba$ and $\ba'$ have different roles as coordinates of $P_\bi$: The former is a tuple of $k$-many coordinates in $[n]=A$, while the latter is a single coordinate in $[|R^\A|]$.

Let $\Mminion$ be a linear minion with semiring $\mathcal{S}$ and depth $d$.
The domain of $\freeM(\Ak)$ is $\Mminion^{(n^k)}$, which we visualise as a subset of $\cT^{n\cdot\bone_k,d}(\mathcal{S})$. Given a symbol $R\in\sigma$ of arity $r$, consider a block tensor $M=(M_\bi)_{\bi\in [r]^k}\in \cT^{r\cdot\bone_k}(\cT^{n\cdot\bone_k,d}(\mathcal{S}))=\cT^{rn\cdot\bone_k,d}(\mathcal{S})$. From the definition of free structure, we have that $M\in R^{\freeM(\Ak)}$ if and only if there exists $Q\in\Mminion^{(|R^\A|)}$ such that $M_\bi=Q_{/\pi_\bi}=P_\bi\cont{1} Q$ for each $\bi\in [r]^k$. 

To give a first glance of this object, we illustrate below the structure of $\freeM(\Ak)$ in the case that $\Mminion=\Qconv$, $k=3$, and $\A=\K_3$. 

\begin{example}
\label{example_free_structure}
Let us denote $\mathbb{F}_{\Qconv}(\K_3^\tensor{3})$ by $\textbf{F}$.
The domain of $\textbf{F}$ is the set of nonnegative tensors in
  $\mathcal{T}^{3\cdot\bone_3}(\Q)$
  whose entries sum up to $1$. The relation $R^{\textbf{F}}$ is the set of those tensors $M\in\mathcal{T}^{2\cdot\bone_3}(\mathcal{T}^{3\cdot\bone_3}(\Q))=\mathcal{T}^{6\cdot\bone_3}(\Q)$ such that there exists a stochastic vector $\bq=(q_1,\dots,q_6)\in\Qconv^{(6)}$
(which should be interpreted as a probability distribution over the elements of $R^{\K_3}$, i.e., over the directed edges in $\K_3$) for which the $\bi$-th block $M_\bi$ of $M$ satisfies $M_\bi=\bq_{/\pi_\bi}$ for each $\bi\in [2]^3$. The $(1,1,1)$-th and the $(2,1,2)$-th entries of $M$ are given below:
\begin{small}
\begin{align*}
M_{(1,1,1)}&=
\left[\begin{array}{@{}ccc|ccc|ccc@{}}
q_1+q_6&0&0&0&0&0&0&0&0\\
0&0&0&0&q_2+q_3&0&0&0&0\\
0&0&0&0&0&0&0&0&q_4+q_5
\end{array}\right]\!\!,\\
M_{(2,1,2)}&=
\left[\begin{array}{@{}ccc|ccc|ccc@{}}
0&0&0       &0&q_1&0      &0&0&q_6\\
q_2&0&0     &0&0&0        &0&0&q_3\\
q_5&0&0     &0&q_4&0      &0&0&0
\end{array}\right]\!\!.
\end{align*}
\end{small}
Figure~\ref{fig_hypermatrix} illustrates the tensor $M\in R^{\textbf{F}}$ corresponding to the uniform distribution $\bq=\frac{1}{6}\cdot\bone_6$.
\end{example}

\begin{figure}
\centering
\begin{tikzpicture}
\begin{scope}[3d view={110}{15},local bounding box=C,scale=.650]
\begin{scope}[shift={(0,0,0)}]
  \foreach \x in {0,...,2}
    \foreach \y in {0,...,2} 
    	\foreach \z in {0,...,2}
    		\cube{0}{\opacityDefault}{\x}{\y}{\z};
\cube{0}{.9}{0}{0}{0};
\cube{0}{.9}{1}{1}{1};
\cube{0}{.9}{2}{2}{2};
\end{scope}

%
%
%
%
\begin{scope}[shift={(0,\bigShift,0)}]
  \foreach \x in {0,...,2}
    \foreach \y in {0,...,2} 
    	\foreach \z in {0,...,2}
    		\cube{0}{\opacityDefault}{\x}{\y}{\z};
\cube{0}{.4}{2}{2}{1};  
\cube{0}{.4}{2}{2}{0};  
\cube{0}{.4}{1}{1}{2};	
\cube{0}{.4}{1}{1}{0};
\cube{0}{.4}{0}{0}{2};	
\cube{0}{.4}{0}{0}{1};
\end{scope}
%
%
%
%
%
\begin{scope}[shift={(0,0,-\bigShift)}]
  \foreach \x in {0,...,2}
    \foreach \y in {0,...,2} 
    	\foreach \z in {0,...,2}
    		\cube{0}{\opacityDefault}{\x}{\y}{\z};
\cube{0}{.4}{2}{1}{2};  
\cube{0}{.4}{2}{0}{2};  
\cube{0}{.4}{1}{2}{1};	
\cube{0}{.4}{1}{0}{1};
\cube{0}{.4}{0}{2}{0};	
\cube{0}{.4}{0}{1}{0};
\end{scope}
%
%
%
%
%
%
\begin{scope}[shift={(0,\bigShift,-\bigShift)}]
  \foreach \x in {0,...,2}
    \foreach \y in {0,...,2} 
    	\foreach \z in {0,...,2}
    		\cube{0}{\opacityDefault}{\x}{\y}{\z};
\cube{0}{.4}{2}{1}{1};  
\cube{0}{.4}{2}{0}{0};  
\cube{0}{.4}{1}{2}{2};	
\cube{0}{.4}{1}{0}{0};
\cube{0}{.4}{0}{2}{2};	
\cube{0}{.4}{0}{1}{1};
\end{scope}
%
%
%
%
%
\begin{scope}[shift={(-\bigShift-2,0,0)}]
  \foreach \x in {0,...,2}
    \foreach \y in {0,...,2} 
    	\foreach \z in {0,...,2}
    		\cube{0}{\opacityDefault}{\x}{\y}{\z};
\cube{0}{.4}{2}{1}{1};  
\cube{0}{.4}{2}{0}{0};  
\cube{0}{.4}{1}{2}{2};	
\cube{0}{.4}{1}{0}{0};
\cube{0}{.4}{0}{2}{2};	
\cube{0}{.4}{0}{1}{1};
\end{scope}
%
%
%
%
%
\begin{scope}[shift={(-\bigShift-2,\bigShift,0)}]
  \foreach \x in {0,...,2}
    \foreach \y in {0,...,2} 
    	\foreach \z in {0,...,2}
    		\cube{0}{\opacityDefault}{\x}{\y}{\z};
\cube{0}{.4}{2}{1}{2};  
\cube{0}{.4}{2}{0}{2};  
\cube{0}{.4}{1}{2}{1};	
\cube{0}{.4}{1}{0}{1};
\cube{0}{.4}{0}{2}{0};	
\cube{0}{.4}{0}{1}{0};
\end{scope}
%
%
%
%
\begin{scope}[shift={(-\bigShift-2,0,-\bigShift)}]
  \foreach \x in {0,...,2}
    \foreach \y in {0,...,2} 
    	\foreach \z in {0,...,2}
    		\cube{0}{\opacityDefault}{\x}{\y}{\z};
\cube{0}{.4}{2}{2}{1};  
\cube{0}{.4}{2}{2}{0};  
\cube{0}{.4}{1}{1}{2};	
\cube{0}{.4}{1}{1}{0};
\cube{0}{.4}{0}{0}{2};	
\cube{0}{.4}{0}{0}{1};
\end{scope}
%
%
%
%
\begin{scope}[shift={(-\bigShift-2,\bigShift,-\bigShift)}]
  \foreach \x in {0,...,2}
    \foreach \y in {0,...,2} 
    	\foreach \z in {0,...,2}
    		\cube{0}{\opacityDefault}{\x}{\y}{\z};
\cube{0}{.9}{0}{0}{0};
\cube{0}{.9}{1}{1}{1};
\cube{0}{.9}{2}{2}{2};
\end{scope}

\end{scope}

\end{tikzpicture}
\caption{A tensor $M\in R^{\textbf{F}}$ from Example~\ref{example_free_structure}, corresponding to the uniform distribution on the set of edges of $\K_3$. The opacity of a cell is proportional to the value of the corresponding entry:
{
{\protect\begin{tikzpicture}[baseline=-1.2ex] \begin{scope}[3d view={110}{15},local bounding box=C,scale=.4]\protect\cube{0}{.9}{0}{0}{0};\end{scope}\end{tikzpicture}} = $\frac{1}{3}$,\;\;
{\protect\begin{tikzpicture}[baseline=-1.2ex] \begin{scope}[3d view={110}{15},local bounding box=C,scale=.4]\protect\cube{0}{.4}{0}{0}{0};\end{scope}\end{tikzpicture}} = $\frac{1}{6}$,\;\;
{\protect\begin{tikzpicture}[baseline=-1.2ex] \begin{scope}[3d view={110}{15},local bounding box=C,scale=.4]\protect\cube{0}{.05}{0}{0}{0};\end{scope}\end{tikzpicture}} = $0$.
}
}
\label{fig_hypermatrix}
\end{figure}

We now show that the entries of $P_\bi$ satisfy the following simple equality, analogous to Lemma~\ref{lem_basic_P_i}.

\begin{lem}
\label{lem_multiplication_rule_P}
Let $k\in\N$, let $\A$ be a $\sigma$-structure, let $R\in\sigma$ of arity $r$, and consider the tuples $\ba\in {A}^k$ and $\bi\in [r]^k$. Then
\begin{align*}
E_\ba\ast P_\bi=\sum_{\substack{\bb\in R^\A\\ \bb_\bi=\ba}}E_\bb.
\end{align*}
\end{lem}

\begin{proof}
For any $\ba'\in R^\A$, we have 
\begin{align*}
\left(\sum_{\substack{\bb\in R^\A\\ \bb_\bi=\ba}}E_\bb\right)\ast E_{\ba'}
=
\sum_{\substack{\bb\in R^\A\\ \bb_\bi=\ba}}(E_\bb\ast E_{\ba'})
=
\sum_{\substack{\bb\in R^\A\\ \bb_\bi=\ba\\ \bb=\ba'}}1
=
\left\{
\begin{array}{cl}
1 & \mbox{if }\ba'_\bi=\ba\\
0 & \mbox{otherwise}
\end{array}
\right.
=
(E_\ba\ast P_\bi)\ast E_{\ba'},
\end{align*}
from which the result follows.
\end{proof}

The next lemma shows that certain entries of a tensor in the relation $R^{\mathbb{F}_{\Mminion}(\A^{\tensor{k}})}$ (i.e., the interpretation of $R$ in the free structure of $\Mminion$ generated by $\A^\tensor{k}$) need to be zero.

\begin{lem}
\label{lem_vanishing_free_structure_0510}
Let $\Mminion$ be a linear minion of depth $d$, let $k\in \N$, let $\A$ be a $\sigma$-structure, let $R\in \sigma$ of arity $r$, and suppose $M=(M_\bi)_{\bi\in [r]^k}\in R^{\freeM(\Ak)}$. Then $E_\ba\ast M_\bi=\bzero_d$ for any $\bi\in [r]^k$, $\ba\in {A}^k$ such that $\bi\not\prec\ba$.\footnote{Recall that $\bi\prec\ba$ means that, for each $\alpha,\beta\in [k]$, $i_\alpha=i_\beta$ implies $a_\alpha=a_\beta$.}
\end{lem}

\begin{proof}
Observe that there exists $Q\in\Mminion^{(|R^\A|)}$ such that $M_\bi=Q_{/\pi_\bi}$ for each $\bi\in [r]^k$. Using Lemma~\ref{lem_multiplication_rule_P}, we obtain
\begin{align*}
E_\ba\ast M_\bi 
&=
E_\ba\ast Q_{/\pi_\bi}
=
E_\ba\ast (P_\bi\cont{1} Q)
=
(E_\ba\ast P_\bi)\ast Q
=
\sum_{\substack{\bb\in R^\A\\ \bb_\bi=\ba}}E_\bb \ast Q
=
\sum_{\substack{\bb\in\emptyset}}E_\bb \ast Q
=
\bzero_d,
\end{align*}
where the fifth equality follows from the fact that $\bb_\bi=\ba$ implies $\bi\prec\ba$;
indeed, in that case, $i_\alpha=i_\beta$ implies $a_\alpha=b_{i_\alpha}=b_{i_\beta}=a_\beta$.
\end{proof}

It follows from Lemma~\ref{lem_vanishing_free_structure_0510} and the previous discussion that, if $\Mminion$ is linear, $\freeM(\Ak)$ can be visualised as a space of sparse and highly symmetric tensors, whose nonzero entries form regular patterns. This feature becomes more
evident for higher values of the level $k$. In turn, the structure of
$\freeM(\Ak)$ is reflected in the properties of the homomorphisms $\xi$ from
$\Xk$ to it -- which, by virtue of Definition~\ref{defn_hierarchy_minion_test},
are precisely the solutions sought by $\Test{\Mminion}{k}$.  
Next, we highlight certain features of such homomorphisms that will be used to prove Theorem~\ref{thm_main_multilinear_tests} in Section~\ref{sec_proof_of_main_thm}.

\begin{lem}
\label{cor_vanishing_contractions}
Let $\Mminion$ be a linear minion, let $k\in\N$, let $\X,\A$ be two $k$-enhanced $\sigma$-structures, and let $\xi:\X^{\tensor{k}}\to\freeM(\Ak)$ be a homomorphism. Then $E_\ba\ast \xi(\bx)=\bzero_d$ for any $\bx\in X^k$, $\ba\in A^k$ such that $\bx\not\prec\ba$.  
\end{lem}

\begin{proof}
From $\bx\in X^k=R_k^\X$, we derive $\bx^\tensor{k}\in R_k^{\X^\tensor{k}}$; since $\xi$ is a homomorphism, this yields $\xi(\bx^\tensor{k})\in R_k^{\freeM(\Ak)}$. Writing $\xi(\bx^\tensor{k})$ in block form as $\xi(\bx^\tensor{k})=(\xi(\bx_\bi))_{\bi\in [k]^k}$ and applying Lemma~\ref{lem_vanishing_free_structure_0510}, we obtain $E_\ba\ast\xi(\bx_\bi)=\bzero_d$ for any $\bi\in [k]^k$ such that $\bi\not\prec\ba$. Write $\bx=(x_1,\dots,x_k)$ and $\ba=(a_1,\dots,a_k)$. Since $\bx\not\prec\ba$, there exist $\alpha,\beta\in [k]$ such that $x_\alpha=x_\beta$ and $a_\alpha\neq a_\beta$. Let $\bi'\in [k]^k$ be the tuple obtained from $(1,\dots,k)$ by replacing the $\beta$-th entry with $\alpha$. Observe that $\bx_{\bi'}=\bx$ and $\bi'\not\prec\ba$. Hence,
\begin{align*}
\bzero_d=E_\ba\ast\xi(\bx_{\bi'})=E_\ba\ast\xi(\bx),
\end{align*}
as required.
\end{proof}

Using Lemma~\ref{cor_vanishing_contractions}, we obtain some more information on the image of a homomorphism from $\X^{\tensor{k}}$ to $\freeM(\Ak)$. Given a signature $\sigma$, we let $\armax(\sigma)$ denote the maximum arity of a relation symbol in $\sigma$. 

\begin{lem}
\label{lem_ea_ast_Q}
Let $\Mminion$ be a linear minion, let $k\in\N$, let $\X,\A$ be two $k$-enhanced $\sigma$-structures such that $k\geq\armax(\sigma)$, and let $\xi:\X^{\tensor{k}}\to\freeM(\Ak)$ be a homomorphism. For $R\in\sigma$ of arity $r$, let $\bx\in R^\X$ and $\ba\in R^\A$ be such that $\bx\not\prec\ba$. Let $Q\in\Mminion^{(|R^\A|)}$ be such that $Q_{/\pi_\bi}=\xi(\bx_\bi)$ for each $\bi\in [r]^k$. Then $E_\ba\ast Q=\bzero_d$.
\end{lem}

\begin{proof}
Write $\bx=(x_1,\dots,x_r)$ and $\ba=(a_1,\dots,a_r)$. Since $\bx\not\prec\ba$, there exist $\alpha,\beta\in [r]$ such that $x_\alpha=x_\beta$ and $a_\alpha\neq a_\beta$. Using that $k\geq r$, we can take the tuple $\bi=(1,2,\dots,r,r,\dots,r)\in [r]^k$. Consider $\bx_\bi\in X^k$, $\ba_\bi\in A^k$. Notice that $x_{i_\alpha}=x_\alpha=x_\beta=x_{i_\beta}$ and $a_{i_\alpha}=a_\alpha\neq a_\beta=a_{i_\beta}$, so $\bx_\bi\not\prec\ba_\bi$. Applying Lemma~\ref{cor_vanishing_contractions}, we find 
\begin{align*}
\bzero_d=E_{\ba_\bi}\ast\xi(\bx_\bi)=E_{\ba_\bi}\ast Q_{/\pi_\bi}
=
E_{\ba_\bi}\ast P_\bi \ast Q.
\end{align*}
Using Lemma~\ref{lem_multiplication_rule_P}, we conclude that 
\begin{align*}
\bzero_d
&=
\sum_{\substack{\bb\in R^\A\\\bb_\bi=\ba_\bi}}E_\bb\ast Q
=
E_\ba\ast Q,
\end{align*}
where the last equality follows from the fact that $\bb_\bi=\ba_\bi$ if and only if $\bb=\ba$.
\end{proof}

If $\X$ and $\A$ are $k$-enhanced $\sigma$-structures, any homomorphism $\xi$ from $\Xk$ to $\freeM(\Ak)$ must satisfy certain symmetries
that ultimately depend on the fact that $\xi$ preserves $R_k$. As shown below in Proposition~\ref{lem_a_symmetry}, these symmetries can be concisely expressed through a list of tensor equations. Given a tuple $\bi\in [k]^k$, we let $\Pi_\bi\in \cT^{n\cdot\bone_{2k}}(\mathcal{S})$ be the tensor defined by 
\begin{align}
\label{defn_tensor_Pi_i}
E_\ba\ast \Pi_\bi\ast E_{\ba'}=
\left\{
\begin{array}{ll}
1 & \mbox{if }\ba'_\bi=\ba\\
0 & \mbox{otherwise}
\end{array}
\right.
\hspace{1cm}
\forall \ba,\ba'\in {A}^k.
\end{align}
Observe that, unlike for the tensor $P_\bi$ defined in~\eqref{defn_tensor_Pi}, the tuples $\ba$ and $\ba'$ have now the same role as coordinates of $\Pi_\bi$.
The tensor defined above satisfies the following simple identity, which should be compared to the one in Lemma~\ref{lem_multiplication_rule_P} concerning $P_\bi$. 
\begin{lem}
\label{lem_multiplication_rule}
For any $\ba\in {A}^k$ and $\bi\in [k]^k$, 
\begin{align*}
E_\ba\ast \Pi_\bi=\sum_{\substack{\bb\in {A}^k\\\bb_\bi=\ba}}E_\bb.
\end{align*}
\end{lem}
\begin{proof}
For any $\ba'\in {A}^k$, we have
\begin{align*}
\left(\sum_{\substack{\bb\in {A}^k\\\bb_\bi=\ba}}E_\bb\right)\ast E_{\ba'}=
\sum_{\substack{\bb\in {A}^k\\\bb_\bi=\ba}}(E_\bb\ast E_{\ba'})
=
\sum_{\substack{\bb\in {A}^k\\\bb_\bi=\ba\\ \bb=\ba'}}1=
\left\{
\begin{array}{ll}
1 & \mbox{if }\ba'_\bi=\ba\\
0 & \mbox{otherwise}
\end{array}
\right.
=
\left(E_\ba\ast \Pi_\bi\right)\ast E_{\ba'},
\end{align*}
from which the result follows.
\end{proof}

\begin{rem}
\label{rem_Pi_i_equals_P_i_sometimes}
It is clear from the expressions~\eqref{defn_tensor_Pi} and~\eqref{defn_tensor_Pi_i} that, for any $\bi\in [k]^k$, $\Pi_\bi$ coincides with the tensor $P_\bi$ associated with the relation symbol $R_k$. 
\end{rem}

\begin{prop}
\label{lem_a_symmetry}
Let $\Mminion$ be a linear minion, let $k\in\N$, let $\X,\A$ be two $k$-enhanced
  $\sigma$-structures, and let $\xi:X^k\to\Mminion^{(n^k)}$ be a map. Then, $\xi$ preserves $R_k$ (interpreted as a symbol in $\sigma^{\tensor{k}}$) if and only if 
\begin{align}
\label{eqn_consistency}
\xi(\bx_\bi)=\Pi_\bi\cont{k} \xi(\bx)
\hspace{2cm}
\mbox{for any }\; \bx\in X^k, \bi\in [k]^k.
\end{align}
\end{prop}
\begin{proof}
Suppose that $\xi$ preserves $R_k$, and take $\bx\in X^k=R_k^\X$. It follows that $\bx^\tensor{k}\in R_k^{\Xk}$, so $\xi(\bx^\tensor{k})\in R_k^{\freeM(\Ak)}$. This means that there exists $Q\in\Mminion^{(|R_k^\A|)}=\Mminion^{(n^k)}$ such that $\xi(\bx_\bi)=Q_{/\pi_\bi}
=\Pi_\bi\cont{k} Q$ for each $\bi\in [k]^k$
(where we have used Remark~\ref{rem_Pi_i_equals_P_i_sometimes}). 
Consider now the tuple $\bj=(1,\dots,k)\in [k]^k$, and observe that $\bx_{\bj}=\bx$. Noticing that the contraction by $\Pi_{\bj}$ acts as the identity, we conclude that 
\begin{align*}
\xi(\bx)=\xi(\bx_{\bj})=\Pi_{\bj}\cont{k}{Q}={Q},
\end{align*}
which concludes the proof of~\eqref{eqn_consistency}.

Conversely, suppose~\eqref{eqn_consistency} holds and take $\bx\in R_k^\X=X^k$. We need to show that $\xi(\bx^\tensor{k})\in R_k^{\freeM(\Ak)}$. Take $Q=\xi(\bx)\in \Mminion^{(n^k)}=\Mminion^{(|R_k^\A|)}$. Using again Remark~\ref{rem_Pi_i_equals_P_i_sometimes}, we observe that, for any $\bi\in [k]^k$,
\begin{align*}
\xi(\bx_\bi)
&=
\Pi_\bi\cont{k}\xi(\bx)
=
\xi(\bx)_{/\pi_\bi}
=
Q_{/\pi_\bi},
\end{align*}
whence the result follows.
\end{proof}

Proposition~\ref{lem_a_symmetry}
distils the requirements of the $\BW^k$, $\SA^k$, $\AIP^k$, $\SoS^k$, and $\BA^k$ hierarchies enforcing compatibility between partial assignments\footnote{Cf.~the ``closure under restriction'' property of $\BW^k$ and the requirements $\clubsuit 2$ and $\spadesuit 3$ in Section~\ref{sec_relaxations_hierarchies} applied to $R=R_k$.} from $\X$ to $\A$ into a single list of tensor equations 
\[\xi(\bx_\bi)=\Pi_\bi\cont{k} \xi(\bx)
\]
for $\bx\in X^k$ and $\bi\in [k]^k$.
For $k=1$, the equation is vacuous, since in this case $\Pi_\bi$ is the identity matrix and $\bx_\bi=\bx$. As $k$ increases, it produces a progressively richer system of symmetries that must be satisfied by $\xi$, which corresponds to a progressively stronger relaxation.

\section{Conic minions}
\label{sec_hierarchies_of_conic_minion_tests}
A primary message of this work is that the tensorisation construction establishes a correspondence between the algebraic properties of a minion and the algorithmic properties of the hierarchy of tests built on the minion. For example, we have seen that if the minion is linear some general properties of the solutions of the hierarchy can be deduced by studying the structure of the tensors in $\freeM(\Ak)$. Now, the bounded-width hierarchy has the property that it only seeks assignments that are partial homomorphisms; similarly, the Sherali--Adams, Sum-of-Squares, and $\BA^k$ hierarchies only assign a positive weight to solutions satisfying local constraints.
The next definition identifies the minion property guaranteeing this algorithmic feature.
\begin{defn}
\label{defn_conic_minion}
A linear minion $\Mminion$ of depth $d$ is \emph{conic} if, for any $L\in\N$ and for any $M\in\Mminion^{(L)}$,
$(i)$
$M\neq O_{L,d}$, and
$(ii)$
for any $V\subseteq [L]$, the following implication is true:\footnote{As usual, the sum, product, $0$, and $1$ operations appearing in this definition are to be meant in the semiring $\mathcal{S}$ associated with the linear minion $\Mminion$.}
\begin{align*}
\begin{array}{ll}
\sum_{i\in V}M^T\be_i=\bzero_d \quad\Rightarrow\quad M^T\be_i=\bzero_d\;\;\forall i\in V.
\end{array}
\end{align*}
\end{defn}
Paraphrasing Definition~\ref{defn_conic_minion}, a linear minion $\Mminion$ is conic if any matrix in $\Mminion$ is nonzero and has the property that, whenever some of its rows sum up to the zero vector, each of those rows is the zero vector. It turns out that all minions appearing in Proposition~\ref{prop_examples_linear_minions} are conic, with the notable exception of $\Zaff$.
\begin{prop}
\label{prop_some_minions_are_conic}
The minions $\Hminion$, $\Qconv$, $\Sminion$, and $\Mblpaip$ are conic, while the minion $\Zaff$ is not.
\end{prop}
\begin{proof}
The fact that $\Qconv$ is conic trivially follows by noting that its elements
  are nonnegative vectors whose entries sum up to $1$. 
  
  Similarly, using the
  description of $\Hminion$ as a linear minion on the semiring
  $(\{0,1\},\vee,\wedge)$ (cf.~the proof of
  Proposition~\ref{prop_examples_linear_minions}), the fact that $\Hminion$ is
  conic follows from the fact that $\bigvee_{i\in V}x_i=0$ means that $x_i=0$
  for each $i\in V$. (Observe also that the vectors in $\Hminion$ are nonzero,
  as they are the indicator vectors of nonempty sets.) 
  
  To show that $\Sminion$
  is conic, take $L\in\N$ and $M\in\Sminion^{(L)}$, and notice first that $M\neq
  O_{L,\ale}$ by (C3) in Definition~\ref{defn_minion_SDP}.
Take now $V\subseteq [L]$. If $\sum_{i\in V}M^T\be_i=\bzero_\ale$, using that
  $MM^T$ is a diagonal matrix by (C2), we find
\begin{align*}
0
=
(\sum_{i\in V}M^T\be_i)^T(\sum_{j\in V}M^T\be_j)
=
\sum_{i,j\in V}\be_i^TMM^T\be_j
=
\sum_{i\in V}\be_i^TMM^T\be_i
=
\sum_{i\in V}\|M^T\be_i\|^2,
\end{align*}
which means that $M^T\be_i=\bzero_\ale$ for any $i\in V$, as required.

As for $\Mblpaip$, we shall see in Section~\ref{sec_semi_direct_product} (cf.~Example~\ref{example_minion_BA_semidirect}) that this minion can be obtained as the semi-direct product of $\Qconv$ and $\Zaff$. Then, the fact that $\Mblpaip$ is conic is a direct consequence of the fact that semi-direct products of minions are always conic (cf.~Proposition~\ref{prop_semidirect_product_makes_sense}).

Finally, the element $(1,-1,1)\in\Zaff$ witnesses that $\Zaff$ is not conic.
\end{proof}

\begin{rem}
It is not hard to verify that also the minion  capturing the power of the $\CLAP$ algorithm from~\cite{cz23sicomp:clap} is linear (with $\mathcal{S}=\Q$ and ${d}=\ale$) and, in fact, conic. Since an algorithmic hierarchy built on top of $\CLAP$ has never been studied in the literature, we do not consider the analogue of Theorem~\ref{thm_main_multilinear_tests} for $\CLAP$.
\end{rem}

\begin{rem}
\label{rem_abstract_conic_minion}
We point out that the concept of conic minions can be extended to arbitrary abstract minions (as opposed to just linear) via the notion of essential coordinates. Given an element $M$ of a minion $\Mminion$ of arity, say, $L$, we say that a coordinate $i\in [L]$ is \textit{essential} for $M$ if there exist an integer $L'\in \N$ and two minor maps $\pi,\pi':[L]\to[L']$ such that $\pi|_{[L]\setminus\{i\}}=\pi'|_{[L]\setminus\{i\}}$ but $M_{/\pi}\neq M_{/\pi'}$. (This straightforwardly extends the analogous notion described in~\cite[Definition~5.14]{BBKO21} for function minions.) Now, if $\Mminion$ is linear of width $d$, it is easy to check that a coordinate $i$ is essential for $M\in\Mminion$ precisely when $M^T\be_i\neq\bzero_d$. Hence, the requirements $(i)$ and $(ii)$ in Definition~\ref{defn_conic_minion} can be extended to abstract minions as follows: $(i)$ asks that each $M\in\Mminion$ should have at least one essential coordinate, and $(ii)$ asks that the image of each essential coordinate of $M$ under a minor map $\pi$ should be essential in $M_{/\pi}$. 
\end{rem}

It turns out that this simple property guarantees that the hierarchies of tests built on conic minions only look at assignments yielding partial homomorphisms, as we shall see in Proposition~\ref{lem_partial_homo}. It also follows that conic hierarchies are not fooled by small instances: Proposition~\ref{prop_sherali_adams_exact} establishes that the $k$-th level of such hierarchies is able to correctly classify instances on $k$ (or fewer) elements, as it is well known for the bounded-width, Sherali--Adams, and Sum-of-Squares hierarchies; we may informally express this property by saying that conic hierarchies are ``sound in the limit''.
Moreover, we shall see in the next section that any linear minion can be transformed into a conic minion -- whose hierarchy enjoys the features mentioned above -- via the \emph{semi-direct product} construction.

First of all, we prove that Lemma~\ref{lem_ea_ast_Q} can be slightly strengthened if we are dealing with conic minions, in that the level $k$ for which it holds can be decreased down to $2$. As a consequence, the algebraic description of the Sherali--Adams and Sum-of-Squares hierarchies in terms of the tensorisation construction can be extended to lower levels, cf.~Remark~\ref{rem_for_SA_and_SoS_no_need_high_level}. As in Section~\ref{sec_hierarchies_of_linear_minion_tests}, the letter $d$ shall denote the depth of the relevant minion in all results of this section. 
\begin{lem}
\label{lem_ea_ast_Q_conic}
Let $\Mminion$ be a conic minion, let $2\leq k\in\N$, let $\X,\A$ be two $k$-enhanced $\sigma$-structures, and let $\xi:\X^{\tensor{k}}\to\freeM(\Ak)$ be a homomorphism. For $R\in\sigma$ of arity $r$, let $\bx\in R^\X$ and $\ba\in R^\A$ be such that $\bx\not\prec\ba$. Let $Q\in\Mminion^{(|R^\A|)}$ be such that $Q_{/\pi_\bi}=\xi(\bx_\bi)$ for each $\bi\in [r]^k$. Then $E_\ba\ast Q=\bzero_d$.
\end{lem}
\begin{proof}
Take $\alpha,\beta\in [r]$ such that $x_\alpha=x_\beta$ and $a_\alpha\neq a_\beta$, and consider the tuple $\bj=(\alpha,\dots,\alpha,\beta)\in [r]^k$. Using that $k\geq 2$, we have $\bx_\bj\not\prec\ba_\bj$, since $x_{j_{k-1}}=x_\alpha=x_\beta=x_{j_k}$ and $a_{j_{k-1}}=a_\alpha\neq a_\beta=a_{j_k}$. 
From Lemma~\ref{cor_vanishing_contractions} and Lemma~\ref{lem_multiplication_rule_P} we obtain
\begin{align*}
\bzero_d
=
E_{\ba_\bj}\ast\xi(\bx_\bj)
=
E_{\ba_\bj}\ast Q_{/\pi_\bj}
=
E_{\ba_\bj}\ast P_\bj\ast Q
=
\sum_{\substack{\bb\in R^\A\\\bb_\bj=\ba_\bj}}E_\bb\ast Q.
\end{align*}
Using that $\Mminion$ is a conic minion, we deduce that $E_\bb\ast Q=\bzero_d$ for any $\bb\in R^\A$ such that $\bb_\bj=\ba_\bj$. In particular, $E_\ba\ast Q=\bzero_d$, as required.
\end{proof}

The following result shows that hierarchies of tests built on conic minions only give a nonzero weight to those assignments that yield partial homomorphisms. 

\begin{prop}
\label{lem_partial_homo}
Let $\Mminion$ be a conic minion, let $k\in\N$, let $\X,\A$ be two $k$-enhanced $\sigma$-structures such that $k\geq\min(2,\armax(\sigma))$, and let $\xi:\X^{\tensor{k}}\to\mathbb{F}_{\Mminion}(\A^{\tensor{k}})$ be a homomorphism. Let $R\in\sigma$ have arity $r$, and take $\bx\in X^k$, $\ba\in A^k$, and $\bi\in [k]^r$. If $\bx_\bi\in R^\X$ and $\ba_\bi\not\in R^\A$, then $E_\ba\ast\xi(\bx)=\bzero_d$.
\end{prop}
\begin{proof}
From $\bx_\bi\in R^\X$ we have $\bx_\bi^\tensor{k}\in R^{\X^\tensor{k}}$ and, thus, $\xi(\bx_\bi^\tensor{k})\in R^{\mathbb{F}_{\Mminion}(\A^{\tensor{k}})}$. It follows that there exists ${Q}\in\Mminion^{(|R^\A|)}$ such that $\xi(\bx_{\bi_\bj})={Q}_{/\pi_\bj}$ for each $\bj\in [r]^k$. Proposition~\ref{lem_a_symmetry} then yields
\begin{align}
\label{eqn_1038_2702}
\Pi_{\bi_\bj}\cont{k}\xi(\bx)
&=
\xi(\bx_{\bi_\bj})
=
{Q}_{/\pi_{\bj}}
=
P_\bj\cont{1}{Q}.
\end{align}
Consider, for each $\alpha\in [k]$, the set $S_\alpha=\{\beta\in [r]:i_\beta=\alpha\}$, and fix an element $\hat{\beta}\in [r]$. The tuple $\bj\in [r]^k$ defined by setting $j_\alpha=\min S_\alpha$ if $S_\alpha\neq\emptyset$, $j_\alpha=\hat{\beta}$ otherwise satisfies $\bi_{\bj_\bi}=\bi$. Indeed, for any $\beta\in [r]$, we have $S_{i_\beta}\neq\emptyset$ since $\beta\in S_{i_\beta}$, so $j_{i_\beta}=\min S_{i_\beta}\in S_{i_\beta}$, which means that $i_{j_{i_\beta}}=i_\beta$, as required.
We obtain
\begin{align}
\label{eqn_2903_1640}
E_{\ba_{\bi_\bj}}\ast P_\bj\ast {Q}
&=
\sum_{\substack{\bb\in R^\A\\\bb_\bj=\ba_{\bi_\bj}}}E_\bb\ast {Q}
=
\sum_{\substack{\bb\in R^\A\\ \bx_\bi\prec\bb \\\bb_\bj=\ba_{\bi_\bj}}}E_\bb\ast {Q},
\end{align}
where the first equality comes from Lemma~\ref{lem_multiplication_rule_P} and the second from Lemma~\ref{lem_ea_ast_Q} or Lemma~\ref{lem_ea_ast_Q_conic} (depending on whether $k\geq \armax(\sigma)$ or $k\geq 2$). We claim that the sum on the right-hand side of~\eqref{eqn_2903_1640} equals $\bzero_d$. Indeed, let $\bb\in A^r$ satisfy $\bx_\bi\prec\bb$ and $\bb_\bj=\ba_{\bi_\bj}$. Since $\bi_{\bj_\bi}=\bi$, for any $\alpha\in [r]$ we have $x_{i_\alpha}=x_{i_{j_{i_\alpha}}}$ and, hence, $b_\alpha=b_{j_{i_\alpha}}$. It follows that $\bb=\bb_{\bj_\bi}=\ba_{\bi_{\bj_\bi}}=\ba_\bi\not\in R^\A$, which proves the claim. Combining this with~\eqref{eqn_1038_2702},~\eqref{eqn_2903_1640}, and Lemma~\ref{lem_multiplication_rule}, we find
\begin{align*}
\bzero_d
&=
E_{\ba_{\bi_\bj}}\ast P_\bj\ast {Q}
=
E_{\ba_{\bi_\bj}}\ast(P_\bj\cont{1}Q)
=
E_{\ba_{\bi_\bj}}\ast(\Pi_{\bi_\bj}\cont{k}\xi(\bx))
=
E_{\ba_{\bi_\bj}}\ast
\Pi_{\bi_\bj}\ast\xi(\bx)\\
&=
\sum_{\substack{\bb\in A^k\\ \bb_{\bi_\bj}=\ba_{\bi_\bj}}}E_\bb\ast\xi(\bx)
\end{align*}
so, in particular, $E_\ba\ast\xi(\bx)=\bzero_d$ since $\Mminion$ is a conic minion.
\end{proof}

The next result shows that hierarchies of tests built on conic minions are sound in the limit, in that they correctly classify all instances whose domain size is less than or equal to the hierarchy level. 

\begin{prop}
\label{prop_sherali_adams_exact}
Let $\Mminion$ be a conic minion, let $k\in\N$, 
let $\X,\A$ be two $k$-enhanced $\sigma$-structures such that $k\geq\min(2,\armax(\sigma))$ and
$k\geq|X|$, and suppose that $\Test{\Mminion}{k}(\X,\A)=\YES$. Then $\X\to\A$.
\end{prop}
\begin{proof}
Let $\xi:\Xk\to\freeM(\Ak)$ be a homomorphism witnessing that $\Test{\Mminion}{k}(\X,\A)=\YES$, and assume without loss of generality that $X=[\ell]$ with $\ell\in [k]$. Take the tuple $\bv=(1,\dots,\ell,\ell,\dots,\ell)\in [\ell]^k$, and notice that $\xi(\bv)\neq O_{n^k,d}$ since $\Mminion$ is a conic minion. Therefore, there exists some $\ba\in A^k$ such that $E_\ba\ast\xi(\bv)\neq\bzero_d$. Consider the function $f:X\to A$ defined by $x\to a_x$ for each $x\in X$. We claim that $f$ yields a homomorphism from $\X$ to $\A$.
Let $R\in\sigma$ be a relation symbol of arity $r$ and take a tuple $\bx\in R^\X$.
Notice that $\bx\in [\ell]^r\subseteq [k]^r$ and
$f(\bx)=\ba_\bx$. Furthermore, it holds that $\bv_\bx=\bx$. Applying Proposition~\ref{lem_partial_homo}, we deduce that $f(\bx)\in R^\A$, whence it follows that $f$ is indeed a homomorphism.
\end{proof}

\section{The semi-direct product of minions}
\label{sec_semi_direct_product}

Is it possible for multiple linear minions to ``join forces'', to obtain a new linear minion corresponding to a stronger relaxation? The natural way to do so is to take as the elements of the new linear minion block matrices, whose blocks are the elements of the original minions. Let $\Mminion$ and $\Nminion$ be two linear minions that we wish to ``merge''. A zero row in a matrix of $\Mminion$ corresponds to zero weight assigned to the variable associated with the row by the relaxation given by $\Mminion$. Ideally, we would like to preserve this information when we run the relaxation given by the second minion $\Nminion$. In other words, we require that zero rows in $\Mminion$ should be associated with zero rows in $\Nminion$. For this to make sense (i.e., for the resulting object to be a linear minion), we need to assume that $\Mminion$ is conic. Under this assumption, it turns out that the new linear minion is conic, too. Therefore, this construction yields a method to transform a linear minion into a conic one, by taking its product with a fixed conic minion (for instance, $\Qconv$). Equivalently, the semi-direct product provides a way to turn a hierarchy of linear tests into a more powerful hierarchy of conic tests -- which enjoys the appealing properties described in Section~\ref{sec_hierarchies_of_conic_minion_tests}. 

\begin{prop}
\label{prop_semidirect_product_makes_sense}
Let $\Mminion$ be a conic minion with semiring $\mathcal{S}$ and depth $d$, let $\Nminion$ be a linear minion with semiring $\mathcal{S}$ and depth $d'$,
and consider, for each $L\in\N$, the set $(\Mminion\ltimes\Nminion)^{(L)}=\{[\begin{array}{cc}
    M & N
\end{array}]:M\in \Mminion^{(L)}, N\in\Nminion^{(L)}, \mbox{ and }N^T\be_i=\bzero_{d'} \mbox{ for any }i\in [L] \mbox{ such that }M^T\be_i=\bzero_{d}\}$. Then $\Mminion\ltimes\Nminion=\bigsqcup_{L\in\N}(\Mminion\ltimes\Nminion)^{(L)}$ is a linear minion with semiring $\mathcal{S}$ and depth $d+d'$. Moreover, it is conic.
\end{prop}

\begin{defn}
\label{defn_semidirect_product_minions}
Let $\Mminion$ and $\Nminion$ be a conic minion and a linear minion, respectively, over the same semiring.
The \emph{semi-direct product}
of $\Mminion$ and $\Nminion$ is the conic minion $\Mminion\ltimes\Nminion$ described in Proposition~\ref{prop_semidirect_product_makes_sense}.
\end{defn}

\begin{rem}
\label{rem_semidirect_product_different_semirings}
If the minions $\Mminion$ and $\Nminion$ in Definition~\ref{defn_semidirect_product_minions} have different semirings $\mathcal{S}$ and $\mathcal{S}'$, we cannot in general use the definition to build their semi-direct product. However, it is immediate to check that a linear minion over $\mathcal{S}$ is also a linear minion over any semiring of which $\mathcal{S}$ is a sub-semiring. Hence, if $\mathcal{S}$ is a sub-semiring of $\mathcal{S}'$ (or vice-versa), $\MNminion$ is well defined (see Example~\ref{example_minion_BA_semidirect} below). In general, however, we might not be able to find a common semiring of which $\mathcal{S}$ and $\mathcal{S}'$ are both sub-semirings. In particular, it is not true in general that the direct sum of semirings admits homomorphic injections from the components, see~\cite{dale1977direct}. To be able to define $\MNminion$ also in this case, we would need to redefine linear minions by allowing each of the $d$ columns of the matrices in a linear minion of depth $d$ to contain entries from a possibly different semiring. 
In this way, the requirement in Definition~\ref{defn_semidirect_product_minions} can be circumvented -- which, for example, makes it possible to define the minion $\Hminion\ltimes\Zaff$ (see Remark~\ref{rem:aip}).
\end{rem}

\begin{example}
    \label{example_minion_BA_semidirect}
    Since $\Z$ is a sub-semiring of $\Q$, we can view $\Qconv$ and $\Zaff$ as a conic minion and a linear minion over the same semiring $\Q$, respectively. It is easy to check that their semi-direct product $\Qconv\ltimes\Zaff$ is precisely the minion $\Mblpaip$ from~\cite{bgwz20}.   
\end{example}

\begin{proof}[Proof of Proposition~\ref{prop_semidirect_product_makes_sense}]
We start by showing that $\MNminion$ is a linear minion of depth $d+d'$. Notice that each set $(\MNminion)^{(L)}$ consists of $L\times (d+d')$ matrices having entries in $\mathcal{S}$. Given $\pi:[L]\to[L']$ and $[\begin{array}{cc}M & N\end{array}]\in (\MNminion)^{(L)}$, we claim that $P[\begin{array}{cc}M & N\end{array}]=[\begin{array}{cc}PM & PN\end{array}]$ belongs to $(\MNminion)^{(L')}$, where $P$ is the $L'\times L$ matrix corresponding to $\pi$ as per Definition~\ref{defn_linear_minion}. First, since $\Mminion$ and $\Nminion$ are both linear minions, we have that $PM=M_{/\pi}\in\Mminion^{(L')}$ and $PN=N_{/\pi}\in\Nminion^{(L')}$. Let $j\in [L']$ be such that $(PM)^T\be_j=\bzero_d$.
We find
\begin{align*}
    \bzero_d
    &=
    M^TP^T\be_j
    =
    \sum_{i\in\pi^{-1}(j)}M^T\be_i.
\end{align*}
Using that $\Mminion$ is conic, we obtain $M^T\be_i=\bzero_d$ for each $i\in\pi^{-1}(j)$. By the definition of $(\MNminion)^{(L)}$, this means that $N^T\be_i=\bzero_{d'}$ for each $i\in\pi^{-1}(j)$. Therefore,
\begin{align*}
    \bzero_{d'}
    &=
    \sum_{i\in\pi^{-1}(j)}N^T\be_i
    =
    N^TP^T\be_j
    =
    (PN)^T\be_j,
\end{align*}
which proves the claim. It follows that $\MNminion$ is indeed a linear minion.

To show that $\MNminion$ is conic, we first note that no element of $\MNminion$ is the all-zero matrix, since $\Mminion$ is conic. If $\sum_{i\in V} [\begin{array}{cc}M & N\end{array}]^T\be_i =\bzero_{d+d'}$ for some $L\in \N$, $[\begin{array}{cc}M & N\end{array}]\in(\MNminion)^{(L)}$, and $V\subseteq [L]$, we find in particular that $\sum_{i\in V}M^T\be_i=\bzero_d$, which implies that $M^T\be_i=\bzero_d$ for each $i\in V$ by the fact that $\Mminion$ is conic. By the definition of $(\MNminion)^{(L)}$, this yields $N^T\be_i=\bzero_{d'}$ for each $i\in V$, so 
\begin{align*}
    [
    \begin{array}{cc}M & N\end{array}]^T\be_i
    &=
    \left[\begin{array}{cc}M^T\be_i\\N^T\be_i\end{array}\right]
    =
    \left[\begin{array}{cc}\bzero_d\\\bzero_{d'}\end{array}\right]
    =
    \bzero_{d+d'}
\end{align*}
for each $i\in V$, as required.
\end{proof}

The next result -- crucial for the characterisation of the $\BA^k$ hierarchy in Theorem~\ref{thm_main_multilinear_tests}, cf.~the proof of Proposition~\ref{prop_BAk_acceptance} -- shows that homomorphisms corresponding to the semi-direct product of two minions factor into homomorphisms corresponding to the components.

\begin{prop}
\label{prop_decoupling_BLPAIP_general}
Let $\Mminion$ be a conic minion with semiring $\mathcal{S}$ and depth $d$, let $\Nminion$ be a linear minion with semiring $\mathcal{S}$ and depth $d'$, let $k\in\N$, and let $\X,\A$ be $k$-enhanced $\sigma$-structures such that $k\geq\armax(\sigma)$. Then there exists a homomorphism $\vartheta:\Xk\to\freeMN(\Ak)$ if and only if there exist homomorphisms $\xi:\Xk\to\freeM(\Ak)$ and $\zeta:\Xk\to\freeN(\Ak)$ such that, for any $\bx\in X^k$ and $\ba\in A^k$, $E_\ba\ast\xi(\bx)=\bzero_d$ implies $E_\ba\ast\zeta(\bx)=\bzero_{d'}$. 
\end{prop}
\begin{proof}
To prove the ``if'' part, take two homomorphisms $\xi$ and $\zeta$ as in the statement of the proposition, and consider the map
\begin{align*}
    \vartheta:X^k&\to (\MNminion)^{(n^k)}\\
    \bx&\mapsto \pair{\xi(\bx)}{\zeta(\bx)}.
\end{align*}
Observe that $\vartheta$ is well defined since we are assuming that $E_\ba\ast\xi(\bx)=\bzero_d$ implies $E_\ba\ast\zeta(\bx)=\bzero_{d'}$ for any $\bx\in X^k$ and $\ba\in A^k$. We claim that $\vartheta$ yields a homomorphism from $\Xk$ to $\freeMN(\Ak)$.
To this end, take $R\in\sigma$ of arity $r$ and $\bx\in R^\X$, so $\bx^{\tensor{k}}\in R^{\Xk}$. We need to show that $\vartheta(\bx^\tensor{k})\in R^{\freeMN(\Ak)}$; equivalently, we need to find some $W\in(\MNminion)^{(m)}$ such that $\vartheta(\bx_\bi)=W_{/\pi_\bi}$ for each $\bi\in [r]^k$, where $m=|R^\A|$. Using that $\xi$ is a homomorphism, we have that $\xi(\bx^\tensor{k})\in R^{\freeM(\Ak)}$, so there exists $Q\in \Mminion^{(m)}$ for which $\xi(\bx_\bi)=Q_{/\pi_\bi}$ for each $\bi\in [r]^k$. Similarly, using that $\zeta$ is a homomorphism, we can find $Z\in \Nminion^{(m)}$ for which $\zeta(\bx_\bi)=Z_{/\pi_\bi}$ for each $\bi\in [r]^k$. The crucial part is to show that $\pair{Q}{Z}\in(\MNminion)^{(m)}$. To this end, take $\ba\in R^\A$ such that $E_\ba\ast Q=\bzero_d$; we need to prove that $E_\ba\ast Z=\bzero_{d'}$. Using the assumption that $k\geq r$, let us pick the tuple $\bj=(1,2,\dots,r,1,1,\dots,1)\in [r]^k$. Notice that this choice guarantees that
    $\{\bb\in R^\A:\bb_\bj=\ba_\bj\}=\{\ba\}$.
Hence,
\begin{align*}
    E_{\ba_\bj}\ast \xi(\bx_\bj)
    &=
    E_{\ba_\bj}\ast Q_{/\pi_\bj}
    =
    E_{\ba_\bj}\ast P_{\bj}\ast Q
    =
    \sum_{\substack{\bb\in R^\A\\\bb_\bj=\ba_\bj}}E_\bb\ast Q
    =
    E_\ba\ast Q,
\end{align*}
where the third equality follows from Lemma~\ref{lem_multiplication_rule_P}. Similarly,
$E_{\ba_\bj}\ast \zeta(\bx_\bj)=E_\ba\ast Z$.
Then, from our assumption $E_\ba\ast Q=\bzero_d$ it follows that $E_{\ba_\bj}\ast \xi(\bx_\bj)=\bzero_d$. Using the hypothesis of the proposition, we deduce that $E_{\ba_\bj}\ast \zeta(\bx_\bj)=\bzero_{d'}$, and we thus conclude that $E_\ba\ast Z=\bzero_{d'}$, as wanted. Call $W=\pair{Q}{Z}$. For each $\bi\in [r]^k$, we find
\begin{align*}
    \vartheta(\bx_\bi)
    =
    \pair{\xi(\bx_\bi)}{\zeta(\bx_\bi)}
    =
    \pair{Q_{/\pi_\bi}}{Z_{/\pi_\bi}}
    =
    \pair{P_{\bi}\cont{1} Q}{P_{\bi}\cont{1} Z}
    =
    P_{\bi}\cont{1}\pair{Q}{Z}
    =
    W_{/\pi_\bi},
\end{align*}
as required. This concludes the proof that $\vartheta$ is a homomorphism.

Conversely, let $\vartheta:\Xk\to\freeMN(\Ak)$ be a homomorphism. For each $\bx\in X^k$, write $\vartheta(\bx)\in(\MNminion)^{(n^k)}$ as $\vartheta(\bx)=\pair{M_{(\bx)}}{N_{(\bx)}}$, where $M_{(\bx)}\in\Mminion^{(n^k)}$ and $N_{(\bx)}\in\Nminion^{(n^k)}$. Using the same argument as in the previous part of the proof, we check that the assignment $\bx\mapsto M_{(\bx)}$ (resp. $\bx\mapsto N_{(\bx)}$) yields a homomorphism from $\Xk$ to $\freeM(\Ak)$ (resp. to $\freeN(\Ak)$), and that the implication $E_\ba\ast\xi(\bx)=\bzero_d$ $\;\Rightarrow\;$ $E_\ba\ast\zeta(\bx)=\bzero_{d'}$ is met for each $\bx\in X^k$ and $\ba\in A^k$. 
\end{proof}

\begin{rem}\label{rem:do23}
In recent work~\cite{Dalmau24:lics}, Dalmau and Opr\v{s}al studied
reductions between PCSPs. For a minion $\Mminion$, they construct a new minion $\omega(\Mminion)$ that they use to characterise the applicability of arc-consistency reductions. If $\Mminion$ is linear, $\omega(\Mminion)$ coincides with the semi-direct product between $\Hminion$ and $\Mminion$ (cf.~Remark~\ref{rem_semidirect_product_different_semirings}). It is not hard to show that a linear minion $\Mminion$ satisfies $\Mminion\to\Hminion\ltimes\Mminion$ if and only if it is homomorphically equivalent to a conic minion.
Using this fact, it can be shown\footnote{Personal communication with Jakub
  Opr\v{s}al.} that the $k$-consistency reductions from~\cite{Dalmau24:lics} and the $k$-th level of a hierarchy of minion tests as defined in this paper are equivalent for conic minions.
  We also point out that the definition of semi-direct product of minions can be extended to abstract minions via the notion of essential coordinates (see  Remark~\ref{rem_abstract_conic_minion}). Indeed, we can define $(\Mminion\ltimes\Nminion)^{(L)}$ as the set of pairs $(M,N)$ such that $M\in \Mminion^{(L)}$, $N\in\Nminion^{(L)}$, and any essential coordinate for $N$ is essential for $M$,
    with minor maps defined coordinate-wise. Using essentially the same argument as in the proof of Proposition~\ref{prop_semidirect_product_makes_sense}, we can verify that the set $\Mminion\ltimes\Nminion=\bigsqcup_{L\in\N}(\Mminion\ltimes\Nminion)^{(L)}$ is closed under minor maps provided that $\Mminion$ satisfies the abstract conicity properties of Remark~\ref{rem_abstract_conic_minion}.
\end{rem}

\section{A proof of Theorem~\ref{thm_main_multilinear_tests}}
\label{sec_proof_of_main_thm}
In this section, we prove Theorem~\ref{thm_main_multilinear_tests} using the machinery developed in Sections~\ref{sec_hierarchies_of_minion_tests},~\ref{sec_hierarchies_of_linear_minion_tests},~\ref{sec_hierarchies_of_conic_minion_tests}, and~\ref{sec_semi_direct_product}. 

\begin{thm*}[Theorem~\ref{thm_main_multilinear_tests} restated]
If $k\in\N$ is at least the maximum arity of the template,
\begin{align*}
    \begin{array}{lll}
         \bullet\;\BW^k=\Test{\Hminion}{k}\\[4pt]
         \bullet\;\SA^k=\Test{\Qconv}{k}\\[4pt]
         \bullet\;\AIP^k=\Test{\Zaff}{k}\\[4pt]
         \bullet\;\SoS^k=\Test{\Sminion}{k}\\[4pt]
         \bullet\;\BA^k=\Test{\Mblpaip}{k}.
    \end{array}
\end{align*}
\end{thm*}

The five parts of the theorem will be formally stated and proved separately, in Propositions~\ref{prop_SA_acceptance},~\ref{prop_BW_acceptance},~\ref{prop_AIPk_acceptance},~\ref{prop_SoS_acceptance}, and~\ref{prop_BAk_acceptance}.
The statements in Propositions~\ref{prop_SA_acceptance} and~\ref{prop_SoS_acceptance}, concerning $\SA^k$ and $\SoS^k$, are actually slightly stronger than Theorem~\ref{thm_main_multilinear_tests}, as they do not require that the level $k$ of the hierarchy be at least the maximum arity of the template (cf.~Remark~\ref{rem_for_SA_and_SoS_no_need_high_level}). We start with $\SA^k$, whose proof is slightly simpler than (and provides intuition for) the proof for $\BW^k$.

\begin{prop}
\label{prop_SA_acceptance}
Let $k\in\N$ and let $\X,\A$ be $k$-enhanced $\sigma$-structures such that $k\geq\min(2,\armax(\sigma))$. Then $\SA^k(\X,\A)=\Test{\Qconv}{k}(\X,\A)$.
\end{prop}
\begin{proof}
Suppose $\SA^k(\X,\A)=\YES$ and let the rational numbers $\lambda_{R,\bx,\ba}$ witness it, for $R\in\sigma$, $\bx\in R^\X$, and $\ba\in R^\A$. Consider the map $\xi:X^k\to\cT^{n\cdot\bone_k}(\Q)$ defined by $E_\ba\ast\xi(\bx)=\lambda_{R_k,\bx,\ba}$ for any $\bx\in X^k$, $\ba\in A^k$. We claim that $\xi$ yields a homomorphism from $\Xk$ to $\freeQ(\Ak)$. Notice first that, for any $\bx\in X^k$, $\xi(\bx)$ is an entrywise nonnegative tensor in the space $\cT^{n\cdot\bone_k}(\Q)$ (which can be identified with $\cT^{n\cdot\bone_k,1}(\Q)$). Moreover, using $\clubsuit 1$, we find
\begin{align*}
\sum_{\ba\in A^k}E_\ba\ast \xi(\bx)
=
\sum_{\ba\in R_k^\A}\lambda_{R_k,\bx,\ba}
=
1.
\end{align*}
It follows that $\xi(\bx)\in\Qconv^{(n^k)}$. We now prove that $\xi$ yields a homomorphism from $\Xk$ to $\freeQ(\Ak)$. Take a symbol $R\in\sigma$ of arity $r$ and a tuple $\bx\in R^\X$, so that $\bx^\tensor{k}\in R^{\Xk}$. We need to show that $\xi(\bx^\tensor{k})\in R^{\freeQ(\Ak)}$. Equivalently, we seek some vector $\bq\in\Qconv^{(|R^\A|)}$ such that $\xi(\bx_\bi)=\bq_{/\pi_\bi}$ for any $\bi\in [r]^k$. Consider the vector $\bq\in \cT^{|R^\A|}(\Q)$ defined by $E_\ba\ast \bq=\lambda_{R,\bx,\ba}$ for any $\ba\in R^\A$. Similarly as before, $\clubsuit 1$ implies that $\bq\in\Qconv^{(|R^\A|)}$. For $\bi\in [r]^k$ and $\ba\in A^k$, we have
\begin{align*}
E_\ba\ast\xi(\bx_\bi)
&=
\lambda_{R_k,\bx_\bi,\ba}
=
\sum_{\substack{\bb\in R^\A\\\bb_\bi=\ba}}\lambda_{R,\bx,\bb}
=
\sum_{\substack{\bb\in R^\A\\\bb_\bi=\ba}}E_\bb\ast\bq
=
E_\ba\ast P_\bi\ast \bq
=
E_\ba\ast \bq_{/\pi_\bi},
\end{align*}
where the second equality is $\clubsuit 2$ and the fourth follows from Lemma~\ref{lem_multiplication_rule_P}. We deduce that $\xi$ is a homomorphism, as claimed, which means that $\Test{\Qconv}{k}(\X,\A)=\YES$.

Conversely, suppose that $\xi$ is a homomorphism from $\Xk$ to $\freeQ(\Ak)$ witnessing that $\Test{\Qconv}{k}(\X,\A)=\YES$. We associate with any pair $(R,\bx)$ such that $R\in\sigma$ and $\bx\in R^\X$ a vector $\bq_{R,\bx}\in\Qconv^{(|R^\A|)}$ defined as follows. Using that $\bx^\tensor{k}\in R^{\Xk}$ and $\xi$ is a homomorphism, we deduce that $\xi(\bx^\tensor{k})\in R^{\freeQ(\Ak)}$ -- i.e., there exists $\bq\in\Qconv^{(|R^\A|)}$ such that $\xi(\bx_\bi)=\bq_{/\pi_\bi}=P_\bi\ast \bq$ for each $\bi\in [r]^k$, where $r$ is the arity of $R$. We set $\bq_{R,\bx}=\bq$. We now build a solution to $\SA^k(\X,\A)$ as follows: For any $R\in\sigma$, $\bx\in R^\X$, and $\ba\in R^\A$, we set $\lambda_{R,\bx,\ba}= E_\ba\ast\bq_{R,\bx}$. Notice that each $\lambda_{R,\bx,\ba}$ is a rational number in the interval $[0,1]$. Moreover, for $R\in\sigma$ and $\bx\in R^\X$, we have
\begin{align*}
\sum_{\ba\in R^\A}\lambda_{R,\bx,\ba}
&=
\sum_{\ba\in R^\A}E_\ba\ast\bq_{R,\bx}
=
1
\end{align*}
since $\bq_{R,\bx}\in\Qconv$, thus yielding $\clubsuit 1$. Observe that, for any $\by\in X^k$, we have $\bq_{R_k,\by}=\xi(\by)$. Indeed, letting $\bj=(1,\dots,k)\in [k]^k$, we have
\begin{align}
\label{eqn_1645_16062022_NEW}
\bq_{R_k,\by}
&=
\Pi_{\bj}\ast \bq_{R_k,\by}
=
P_\bj\ast \bq_{R_k,\by}
=
\bq_{{R_k,\by}_{/\pi_{\bj}}}
=
\xi(\by_{\bj})
=
\xi(\by), 
\end{align}
where 
the first equality follows from the fact that the contraction by $\Pi_\bj$ acts as the identity (cf.~the proof of Proposition~\ref{lem_a_symmetry})
and 
the second  
from 
Remark~\ref{rem_Pi_i_equals_P_i_sometimes}.
Then, $\clubsuit 2$ follows by noticing that, for $\bi\in [r]^k$ and $\bb\in A^k$, 
\begin{align*}
\sum_{\substack{\ba\in R^\A\\\ba_\bi=\bb}}\lambda_{R,\bx,\ba}
=
\sum_{\substack{\ba\in R^\A\\\ba_\bi=\bb}}E_\ba\ast\bq_{R,\bx}
=
E_\bb\ast P_\bi\ast\bq_{R,\bx}
=
E_\bb\ast\xi(\bx_\bi)
=
E_\bb\ast\bq_{R_k,\bx_\bi}
=
\lambda_{R_k,\bx_\bi,\bb},
\end{align*}
where the second and fourth equalities follow from Lemma~\ref{lem_multiplication_rule_P} and~\eqref{eqn_1645_16062022_NEW}, respectively. Recall from Proposition~\ref{prop_some_minions_are_conic} that $\Qconv$ is a conic minion. Using either Lemma~\ref{lem_ea_ast_Q} or Lemma~\ref{lem_ea_ast_Q_conic} (depending on whether $k\geq\armax(\sigma)$ or $k\geq 2$), if $\ba\in R^\A$ is such that $\bx\not\prec\ba$, we obtain
\begin{align*}
\lambda_{R,\bx,\ba}
&=
E_\ba\ast\bq_{R,\bx}
=
0,
\end{align*}
thus showing that $\clubsuit 3$ is satisfied, too. It follows that $\SA^k(\X,\A)=\YES$, as required.
\end{proof}
\begin{rem}
\label{rem_redundancy_zero_condition_hierarchies}
    Observe that condition $\clubsuit 3$ is not used in the first implication of the proof of Proposition~\ref{prop_SA_acceptance}, showing that $\SA^k(\X,\A)=\YES$ implies $\Xk\to\freeQ(\Ak)$. Since, however, $\clubsuit 3$ is implied by $\Xk\to\freeQ(\Ak)$ (as showed in the second part of the proof), it follows that $\clubsuit 3$ is redundant in the system~\eqref{eqn_SA_and_AIPk} for $\SA^k$ when $k\geq\min(2,\armax(\sigma))$. 
    This fact can be easily seen from Lemma~\ref{lem_ea_ast_Q} and Lemma~\ref{lem_ea_ast_Q_conic}, which guarantee that the $\ba$-th entry of the witness that a homomorphism $\xi: \Xk\to\freeM(\Ak)$ (i.e., a hierarchy solution) preserves a tuple $\bx\in R^\X$ is zero whenever $\bx\not\prec\ba$---which precisely corresponds to condition $\clubsuit 3$.
    Similarly, the proofs of Propositions~\ref{prop_AIPk_acceptance}  and~\ref{prop_BAk_acceptance} show that $\clubsuit 3$ is redundant for both $\AIP^k$ and $\BA^k$, while the proof of Proposition~\ref{prop_SoS_acceptance} shows that $\spadesuit 4$ is redundant for $\SoS^k$. 
    We do not omit such conditions from the definition of the hierarchies in Section~\ref{sec_relaxations_hierarchies} since this is the way they are commonly described in the literature, with variables corresponding to sets $V\subseteq X$ and partial assignments from $V$ to $A$.
\end{rem}
\begin{prop}
\label{prop_BW_acceptance}
Let $k\in\N$ and let $\X,\A$ be $k$-enhanced $\sigma$-structures such that $k\geq\armax(\sigma)$. Then $\BW^k(\X,\A)=\Test{\Hminion}{k}(\X,\A)$.
\end{prop}
\begin{proof}
Given two sets $S,T$, an integer $p\in\N$, and two tuples $\textbf{s}=(s_1,\dots,s_p)\in S^p, \textbf{t}=(t_1,\dots,t_p)\in T^p$ such that $\textbf{s}\prec\textbf{t}$, we shall consider the function $f_{\textbf{s},\textbf{t}}:\{\textbf{s}\}\to T$ defined by $f_{\textbf{s},\textbf{t}}(s_\alpha)=t_\alpha$ for each $\alpha\in [p]$. Also, we denote by $\epsilon:\emptyset\to\A$ the empty mapping.

Suppose $\BW^k(\X,\A)=\YES$ and let $\mathcal{F}$ be a nonempty collection of partial homomorphisms from $\X$ to $\A$ witnessing it.
Recall from Section~\ref{sec_tensors_tuples} that the space of tensors $\cT^{n\cdot\bone_k,1}(\{0,1\})$ can be identified with $\cT^{n\cdot\bone_k}(\{0,1\})$.
Define the map $\xi:X^k\to\mathcal{T}^{n\cdot\bone_k}(\{0,1\})$ by setting, for $\bx\in X^k$ and $\ba\in A^k$, $E_\ba\ast\xi(\bx)=1$ if $\bx\prec\ba$ and $f_{\bx,\ba}\in \mathcal{F}$, $E_\ba\ast\xi(\bx)=0$ otherwise. We claim that $\xi$ yields a homomorphism from $\X^\tensor{k}$ to $\mathbb{F}_{\mathscr{H}}(\A^\tensor{k})$. Take $R\in\sigma$ of arity $r$ and $\by\in R^\X$, so $\by^\tensor{k}\in R^{\Xk}$. We need to show that $\xi(\by^\tensor{k})\in R^{\mathbb{F}_{\mathscr{H}}(\A^\tensor{k})}$. 
Since $k\geq r$, we can write $\by=\bx_\bi$ for some $\bx\in X^k$, $\bi\in [k]^r$.
Given $\ba\in R^\A$, consider the set $B_\ba=\{\bb\in A^k:\bb_\bi=\ba \mbox{ and }\bx\prec\bb\}$. We define a vector $\bq\in\cT^{|R^\A|}(\{0,1\})$ by letting, for each $\ba\in R^\A$, the $\ba$-th entry of $\bq$ be $1$ if $f_{\bx,\bb}\in \mathcal{F}$ for some $\bb\in B_\ba$, $0$ otherwise. 
We now show that $\bq\in \mathscr{H}^{(|R^\A|)}$; i.e., that $\bq$ is not identically zero. 
Observe first that, since $\mathcal{F}$ is nonempty and closed under restrictions, it contains the empty mapping $\epsilon$. Applying the extension property to $\epsilon$, we find that there exists some $f\in\mathcal{F}$ whose domain is $\{\bx\}$ -- that is, there exists some $\bc\in A^k$ such that $\bx\prec \bc$ and $f_{\bx,\bc}=f\in\mathcal{F}$. Notice that $\by\in R^\X\cap \{\bx\}^r=R^{\X[\{\bx\}]}$ (where we recall that $\X[\{\bx\}]$ is the substructure of $\X$ induced by $\{\bx\}$). Using that $f_{\bx,\bc}$ is a partial homomorphism, we obtain $\bc_\bi=f_{\bx,\bc}(\bx_\bi)=f_{\bx,\bc}(\by)\in R^\A$. We then conclude that $\be_{\bc_\bi}\ast \bq=1$, so $\bq\in \mathscr{H}^{(|R^\A|)}$, as required.
If we manage to show that $\bq_{/\pi_{\bell}}=\xi(\by_{\bell})$ for any $\bell\in [r]^k$, we can conclude that $\xi(\by^\tensor{k})\in R^{\mathbb{F}_{\mathscr{H}}(\A^\tensor{k})}$, thus proving the claim.
Recall from Proposition~\ref{prop_examples_linear_minions} that $\Hminion$ is a linear minion on the semiring $(\{0,1\},\vee,\wedge)$.
For $\ba\in A^k$, using Lemma~\ref{lem_multiplication_rule_P}, we find
\begin{align}
\label{eqn_1405_24062022}
E_\ba\ast \bq_{/\pi_{\bell}}
=
E_\ba\ast P_{\bell}\ast \bq
=
\sum_{\substack{\bc\in R^\A\\\bc_{\bell}=\ba}}E_\bc\ast \bq
=
\bigvee_{\substack{\bc\in R^\A\\\bc_{\bell}=\ba}}E_\bc\ast \bq.
\end{align}
It follows that the expression in~\eqref{eqn_1405_24062022} equals $1$ if 
\begin{align}
\label{eqn_star}
\tag{$\star$}
\exists \bc\in R^\A \,\mbox{ s.t. }\, \bc_{\bell}=\ba \,\mbox{ and }\, f_{\bx,\bb}\in \mathcal{F} \,\mbox{ for some }\, \bb\in B_{\bc},
\end{align}
$0$ otherwise. On the other hand,
$
E_\ba\ast \xi(\by_{\bell})
$
equals $1$ if 
\begin{align}
\label{eqn_bullet}
\tag{$\bullet$}
\by_{\bell}\prec \ba \,\mbox{ and }\, f_{\by_{\bell},\ba}\in \mathcal{F},
\end{align}
$0$ otherwise. We now show that the conditions~\eqref{eqn_star} and~\eqref{eqn_bullet} are equivalent, which concludes the proof of the claim. Suppose that~\eqref{eqn_star} holds. Since $\bb\in B_\bc$, we have $\bx\prec\bb$, which yields $\bx_{\bi_\bell}\prec \bb_{\bi_\bell}$ as ``$\prec$'' is preserved under projections. Using the restriction property applied to $f_{\bx,\bb}$, we find that $f_{\bx_{\bi_{\bell}},\bb_{\bi_{\bell}}}\in \mathcal{F}$. Then,~\eqref{eqn_bullet} follows by observing that $\bx_{\bi_\bell}=\by_\bell$ and, since $\bb\in B_\bc$ and $\bc_\bell=\ba$, $\bb_{\bi_\bell}=\ba$.
Suppose now that~\eqref{eqn_bullet} holds. Using the extension property applied to $f_{\by_{\bell},\ba}=f_{\bx_{\bi_{\bell}},\ba}$ we find that $f_{\bx,\bb}\in \mathcal{F}$ for some $\bb\in A^k$ such that $\bx\prec\bb$ and $\bb_{\bi_{\bell}}=\ba$. Since $f_{\bx,\bb}$ is a partial homomorphism from $\X$ to $\A$, $R^\A\ni f_{\bx,\bb}(\by)=f_{\bx,\bb}(\bx_\bi)=\bb_\bi$. Calling $\bc=\bb_\bi$, we obtain~\eqref{eqn_star}.

Conversely, suppose that $\xi:\X^\tensor{k}\to\mathbb{F}_{\mathscr{H}}(\A^\tensor{k})$ is a homomorphism witnessing that $\Test{\Hminion}{k}(\X,\A)$ accepts, and consider the collection $\mathcal{F}=\{f_{\bx,\ba}:\bx\in X^k,\ba\in\supp(\xi(\bx))\}\cup\{\epsilon\}$. Notice that $\mathcal{F}$ is well defined by virtue of Lemma~\ref{cor_vanishing_contractions}, as $\ba\in\supp(\xi(\bx))$ implies that $\bx\prec \ba$, and it is nonempty. 
We claim that any function $f_{\bx,\ba}\in \mathcal{F}$ is a partial homomorphism from $\X$ to $\A$. (Notice that $\epsilon$ is trivially a partial homomorphism.) Indeed, given $R\in\sigma$ of arity $r$ and $\by\in R^{\X[\{\bx\}]}=R^\X\cap \{\bx\}^r$, we can write $\by=\bx_\bi$ for some $\bi\in [k]^r$. 
Then, $f_{\bx,\ba}(\by)=f_{\bx,\ba}(\bx_\bi)=\ba_\bi\in R^\A$, where we have used Proposition~\ref{lem_partial_homo} (which applies to $\Hminion$ since, by Proposition~\ref{prop_some_minions_are_conic}, $\Hminion$ is a conic minion). To show that $\mathcal{F}$ is closed under restrictions, take $f\in\mathcal{F}$ and $V\subseteq\dom(f)$; we need to show that $f|_{V}\in\mathcal{F}$. The cases $f=\epsilon$ or $V=\emptyset$ are trivial, so we can assume $f=f_{\bx,\ba}$ (which means that $\dom(f)=\{\bx\}$) and write $V=\{\bx_\bell\}$ for some $\bell\in [k]^k$. 
Observe that $f_{\bx,\ba}|_{V}=f_{\bx_\bell,\ba_\bell}$. We claim that $E_{\ba_\bell}\ast\xi(\bx_\bell)=1$. Otherwise, using Proposition~\ref{lem_a_symmetry} and Lemma~\ref{lem_multiplication_rule}, we would have
\begin{align*}
0
&=
E_{\ba_\bell}\ast\xi(\bx_\bell)
=
E_{\ba_\bell}\ast(\Pi_\bell\cont{k}\xi(\bx))
=
E_{\ba_\bell}\ast\Pi_\bell\ast\xi(\bx)
=
\bigvee_{\substack{\bb\in A^k\\ \bb_\bell=\ba_\bell}}E_\bb\ast \xi(\bx).
\end{align*}
This would imply that $E_\bb\ast\xi(\bx)=0$ whenever $\bb\in A^k$ is such that $\bb_\bell=\ba_\bell$; in particular, $E_\ba\ast\xi(\bx)=0$, a contradiction. So $E_{\ba_\bell}\ast\xi(\bx_\bell)=1$, as claimed, and it follows that $f_{\bx_\bell,\ba_\bell}\in \mathcal{F}$. 
We now claim that $\mathcal{F}$ has the extension property up to $k$. Take $f\in\mathcal{F}$ and $V\subseteq X$ such that $|V|\leq k$ and $\dom(f)\subseteq V$; we need to show that there exists $g\in\mathcal{F}$ such that $g$ extends $f$ and $\dom(g)=V$. If $f=\epsilon$ and $V=\emptyset$, the claim is trivial; if $f=\epsilon$ and $V\neq\emptyset$, we can write $V=\{\bx\}$ for some $\bx\in X^k$, and the claim follows by noticing that, by the definition of $\Hminion$, $\supp(\xi(\bx))\neq\emptyset$. Therefore, we can assume that $f\neq\epsilon$, so $f=f_{\bx,\ba}$ for some $\bx\in X^k$, $\ba\in\supp(\xi(\bx))$. Since $V\neq\emptyset$, we can write $V=\{\by\}$ for some $\by\in X^k$. Then, $\dom(f)\subseteq V$ becomes $\{\bx\}\subseteq\{\by\}$, so $\bx=\by_\bell$ for some $\bell\in [k]^k$. Using Proposition~\ref{lem_a_symmetry} and Lemma~\ref{lem_multiplication_rule}, we find
\begin{align*}
1=E_\ba\ast\xi(\bx)=E_\ba\ast\xi(\by_\bell)=E_\ba\ast \Pi_\bell\ast\xi(\by)=\bigvee_{\substack{\bb\in A^k\\\bb_\bell=\ba}}E_\bb\ast\xi(\by),
\end{align*}  
which implies that $E_\bb\ast\xi(\by)=1$ for some $\bb\in A^k$ such that $\bb_\bell=\ba$. It follows that $f_{\by,\bb}\in\mathcal{F}$. Notice that $\dom(f_{\by,\bb})=\{\by\}=V$, and $f_{\by,\bb}|_{\{\bx\}}=f_{\bx,\ba}$, so the claim is true. Hence, $\mathcal{F}$ witnesses that $\BW^k(\X,\A)=\YES$.
\end{proof}
\begin{rem}
Recall from Remark~\ref{rem_enhancement} that, unlike the other hierarchies in Theorem~\ref{thm_main_multilinear_tests}, the $\BW^k$ hierarchy does not require that the structures $\X$ and $\A$ to which it is applied should be $k$-enhanced. In particular, the proof of Proposition~\ref{prop_BW_acceptance} actually establishes the following, slightly stronger statement: \textit{Let $k\in\N$ and let $\X,\A$ be $\sigma$-structures such that $k\geq\armax(\sigma)$. Then $\BW^k(\X,\A)=\Test{\Hminion}{k}(\tilde\X,\tilde\A)$, where $\tilde\X$ (resp. $\tilde\A$) is the $k$-enhanced version of $\X$ (resp. $\A$).} 
\end{rem}

\begin{prop}
\label{prop_AIPk_acceptance}
Let $k\in\N$ and let $\X,\A$ be $k$-enhanced $\sigma$-structures such that $k\geq\armax(\sigma)$. Then $\AIP^k(\X,\A)=\Test{\Zaff}{k}(\X,\A)$.
\end{prop}
\begin{proof}
The proof is analogous to that of Proposition~\ref{prop_SA_acceptance}, the only difference being that Lemma~\ref{lem_ea_ast_Q_conic} cannot be applied in this case since $\Zaff$ is not a conic minion (cf.~Proposition~\ref{prop_some_minions_are_conic}). As a consequence, unlike in Proposition~\ref{prop_SA_acceptance},  we now need to assume that $k\geq\armax(\sigma)$.
\end{proof}

The proof of Theorem~\ref{thm_main_multilinear_tests} for $\SoS^k$, given in Proposition~\ref{prop_SoS_acceptance} below, follows the same scheme as that of Proposition~\ref{prop_SA_acceptance}. There is, however, one additional complication due to the fact that the objects in the minion $\Sminion$ are matrices having infinitely many columns. We deal with this technical issue through the orthonormalisation argument we already used for the proof of Proposition~\ref{prop_acceptance_SDP}: We find an orthonormal basis for the finitely generated vector space defined as the sum of the row-spaces of the matrices in $\Sminion$ appearing as images of a given homomorphism.
\begin{prop}
\label{prop_SoS_acceptance}
Let $k\in\N$ and let $\X,\A$ be $k$-enhanced $\sigma$-structures such that $k\geq\min(2,\armax(\sigma))$. Then $\SoS^k(\X,\A)=\Test{\Sminion}{k}(\X,\A)$.
\end{prop}
\begin{proof}
Suppose that $\SoS^k(\X,\A)=\YES$ and let the family of vectors $\blambda_{R,\bx,\ba}\in \R^\gamma$ witness it, where $\gamma=\sum_{R\in\sigma}|R^\X|\cdot|R^\A|$. Consider the map $\xi:X^k\to\mathcal{T}^{n\cdot\bone_k,\ale}(\R)$ defined by
\begin{align*}
E_\ba\ast\xi(\bx)=
\begin{bmatrix}
\blambda_{R_k,\bx,\ba}\\
\bzero_\ale
\end{bmatrix}
\hspace{1cm}
\bx\in X^k,\ba\in A^k.
\end{align*}
We claim that $\xi$ yields a homomorphism from $\Xk$ to $\freeS(\Ak)$. First of
  all, we need to show that $\xi(\bx)\in \Sminion^{(n^k)}$ for each $\bx\in
  X^k$. The requirement (C1) is trivially satisfied since, by construction, the $j$-th entry of $E_\ba\ast\xi(\bx)$ is zero whenever $j>\gamma$. Given $\ba,\ba'\in A^k$,
\begin{align*}
(E_\ba\ast\xi(\bx))^T(E_{\ba'}\ast\xi(\bx))
&=
\blambda_{R_k,\bx,\ba}\cdot\blambda_{R_k,\bx,\ba'}.
\end{align*}
  If $\ba\neq \ba'$, this quantity is zero by $\spadesuit 2$, so (C2) is satisfied. Finally,
\begin{align*}
\sum_{\ba\in A^k}(E_\ba\ast\xi(\bx))^T(E_{\ba}\ast\xi(\bx))
&=
\sum_{\ba\in A^k}\|\blambda_{R_k,\bx,\ba}\|^2
=1
\end{align*}
by $\spadesuit 1$,
  so (C3) is also satisfied. Therefore, $\xi(\bx)\in\Sminion^{(n^k)}$. To show that $\xi$ is in fact a homomorphism, take $R\in\sigma$ of arity $r$ and $\bx\in R^\X$, so $\bx^\tensor{k}\in R^{\Xk}$. We need to show that $\xi(\bx^\tensor{k})\in R^{\freeS(\Ak)}$. Consider the matrix $Q\in\cT^{|R^\A|,\ale}(\R)$ defined by 
\begin{align*}
Q^T\be_\ba=\begin{bmatrix}
\blambda_{R,\bx,\ba}\\ \bzero_\ale
\end{bmatrix}
\hspace{1cm}
\ba\in R^\A.
\end{align*}
  Using the same arguments as above, we check that $Q$ satisfies (C1) and that
  $\be_\ba^TQQ^T\be_{\ba'}=\blambda_{R,\bx\,\ba}\cdot\blambda_{R,\bx,\ba'}$, so
  (C2) follows from $\spadesuit 2$ and (C3) from $\spadesuit 1$. Therefore, $Q\in\Sminion^{(|R^\A|)}$. We now claim that $\xi(\bx_\bi)=Q_{/\pi_\bi}$ for each $\bi\in [r]^k$. Indeed, observe that, for each $\ba\in A^k$,
\begin{align*}
E_\ba\ast Q_{/\pi_\bi}
&=
E_\ba\ast P_\bi\ast Q
=
\big(\sum_{\substack{\bb\in R^\A\\\bb_\bi=\ba}}E_\bb\big)\ast Q
=
Q^T\big(\sum_{\substack{\bb\in R^\A\\\bb_\bi=\ba}}\be_\bb\big)
=
\sum_{\substack{\bb\in R^\A\\\bb_\bi=\ba}}Q^T\be_\bb\\
&=
\begin{bmatrix}
\sum_{\substack{\bb\in R^\A\\\bb_\bi=\ba}}\blambda_{R,\bx,\bb}\\\bzero_\ale
\end{bmatrix}
=
\begin{bmatrix}
\blambda_{R_k,\bx_\bi,\ba}
\\ \bzero_\ale
\end{bmatrix}
=
E_\ba\ast\xi(\bx_\bi),
\end{align*}
where the second and sixth equalities are obtained using Lemma~\ref{lem_multiplication_rule_P}
 and $\spadesuit 3$, respectively. It follows that the claim is true, so $\xi(\bx^\tensor{k})\in R^{\freeS(\Ak)}$, which concludes the proof that $\xi$ is a homomorphism and that $\Test{\Sminion}{k}(\X,\A)=\YES$.

Conversely, let $\xi:\Xk\to\freeS(\Ak)$ be a homomorphism witnessing that $\Test{\Sminion}{k}(\X,\A)=\YES$. Take $R\in\sigma$ of arity $r$ and $\bx\in R^\X$. We have that $\bx^\tensor{k}\in R^{\Xk}$, so $\xi(\bx^\tensor{k})\in R^{\freeS(\Ak)}$ since $\xi$ is a homomorphism. As a consequence, we can fix a matrix $Q_{R,\bx}\in\Sminion^{(|R^\A|)}$ satisfying $\xi(\bx_\bi)={Q_{R,\bx}}_{/\pi_\bi}$ for each $\bi\in[r]^k$. Consider the set $S=\{Q_{R,\bx}^T\be_\ba:R\in\sigma,\bx\in R^\X,\ba\in R^\A\}$ and the vector space $\mathcal{U}=\Span(S)\subseteq\R^\ale$, and observe that $\dim(\mathcal{U})\leq |S|\leq \sum_{R\in\sigma}|R^\X|\cdot|R^\A|=\gamma$. Choose a vector space $\mathcal{V}$ of dimension $\gamma$ such that $\mathcal{U}\subseteq\mathcal{V}\subseteq \R^\ale$. Using the Gram–Schmidt process, we find a projection matrix ${Z}\in\R^{\ale,\gamma}$ such that ${Z}^T{Z}=I_\gamma$ and ${Z}{Z}^T\bv=\bv$ for any $\bv\in\mathcal{V}$. Consider the family of vectors
\begin{align}
\label{eqn_1111_1904}
\blambda_{R,\bx,\ba}={Z}^TQ_{R,\bx}^T\be_\ba\in\R^\gamma \hspace{1cm} R\in\sigma,\bx\in R^\X,\ba\in R^\A.
\end{align}
We claim that~\eqref{eqn_1111_1904} witnesses that $\SoS^k(\X,\A)=\YES$. Take $R\in\sigma$ of arity $r$ and $\bx\in R^\X$. Recall from Proposition~\ref{prop_some_minions_are_conic} that $\Sminion$ is a conic minion. Using either Lemma~\ref{lem_ea_ast_Q} or Lemma~\ref{lem_ea_ast_Q_conic} (depending on whether $k\geq\armax(\sigma)$ or $k\geq 2$), given $\ba\in R^\A$ such that $\bx\not\prec\ba$,
we find
\begin{align*}
\blambda_{R,\bx,\ba}
&=
{Z}^TQ_{R,\bx}^T\be_\ba
=
{Z}^T\bzero_\ale
=
\bzero_\gamma,
\end{align*}
so $\spadesuit 4$ holds. $\spadesuit 1$ follows from
\begin{align*}
\sum_{\ba\in R^\A}\|\blambda_{R,\bx,\ba}\|^2
&=
\sum_{\ba\in R^\A}
\be_\ba^TQ_{R,\bx}{Z}{Z}^TQ_{R,\bx}^T\be_\ba
=
\sum_{\ba\in R^\A}
\be_\ba^TQ_{R,\bx}Q_{R,\bx}^T\be_\ba
=
\tr(Q_{R,\bx}Q_{R,\bx}^T)
=
1,
\end{align*}
where the second equality is true since $Q_{R,\bx}^T\be_\ba\in
S\subseteq\mathcal{U}\subseteq\mathcal{V}$ and the fourth follows from (C3).
Similarly, using (C2), we find that, if $\ba\neq\ba'\in R^\A$,
\begin{align*}
\blambda_{R,\bx,\ba}\cdot\blambda_{R,\bx,\ba'}
&=
\be_\ba^TQ_{R,\bx}{Z}{Z}^TQ_{R,\bx}^T\be_{\ba'}
=
\be_\ba^TQ_{R,\bx}Q_{R,\bx}^T\be_{\ba'}
=
0,
\end{align*}
so $\spadesuit 2$ holds. We now show that $Q_{R_k,\by}=\xi(\by)$ for any $\by\in X^k$. Indeed, using the same argument as in~\eqref{eqn_1645_16062022_NEW}, letting $\bj=(1,\dots,k)\in [k]^k$, we have
\begin{align}
\label{eqn_1645_16062022}
Q_{R_k,\by}
&=
\Pi_{\bj}\cont{k} Q_{R_k,\by}
=
Q_{{R_k,\by}_{/\pi_{\bj}}}
=
\xi(\by_{\bj})
=
\xi(\by). 
\end{align}
Given $\bi\in [r]^k$ and $\bb\in A^k=R_k^\A$, we have
\begin{align*}
\sum_{\substack{\ba\in R^\A\\\ba_\bi=\bb}}\blambda_{R,\bx,\ba}
&=
\sum_{\substack{\ba\in R^\A\\\ba_\bi=\bb}}{Z}^TQ_{R,\bx}^T\be_\ba
=
\sum_{\substack{\ba\in R^\A\\\ba_\bi=\bb}}E_\ba\ast Q_{R,\bx}\ast {Z}
=
E_\bb\ast P_\bi\ast Q_{R,\bx}\ast {Z}
=
E_\bb\ast {Q_{R,\bx}}_{/\pi_\bi}\ast {Z}\\
&=
E_\bb\ast\xi(\bx_\bi)\ast {Z}
=
E_\bb\ast Q_{R_k,\bx_\bi}\ast {Z}
=
{Z}^TQ_{R_k,\bx_\bi}^T\be_\bb
=
\blambda_{R_k,\bx_\bi,\bb},
\end{align*}
where the third and sixth equalities follow from Lemma~\ref{lem_multiplication_rule_P} and~\eqref{eqn_1645_16062022}, respectively. This shows that $\spadesuit 3$ holds, too, so that~\eqref{eqn_1111_1904} yields a solution for $\SoS^k(\X,\A)$, as claimed.
\end{proof}
\begin{prop}
\label{prop_BAk_acceptance}
Let $k\in\N$ and let $\X,\A$ be $k$-enhanced $\sigma$-structures such that $k\geq\armax(\sigma)$. Then $\BA^k(\X,\A)=\Test{\Mblpaip}{k}(\X,\A)$.
\end{prop}
\begin{proof}
    Recall from Section~\ref{sec_relaxations_hierarchies} that $\BA^k(\X,\A)=\YES$ is equivalent to the existence of a rational nonnegative solution (denoted by the superscript $(\Bmat)$) and an integer solution (denoted by the superscript $(\Amat)$) to the system~\eqref{eqn_SA_and_AIPk}, such that 
    \begin{align}
    \label{eqn_refinement_step_01032023}
    \lambda^{(\Bmat)}_{R,\bx,\ba}=0
    \hspace{.35cm}
    \Rightarrow
    \hspace{.5cm}
    \lambda^{(\Amat)}_{R,\bx,\ba}=0
    \end{align}
    for each $R\in\sigma$, $\bx\in R^\X$, and $\ba\in R^\A$. Note that requiring~\eqref{eqn_refinement_step_01032023} for each $R\in\sigma$ is equivalent to only requiring it for $R=R_k$. Indeed, take some $R\in\sigma$ of arity $r$, $\bx\in R^\X$, and $\ba\in R^\A$, and consider the tuple $\bi=(1,2,\dots,r,1,1,\dots,1)\in [r]^k$, which is well defined as $k\geq r$. Noting that $\{\bb\in R^\A:\bb_\bi=\ba_\bi\}=\{\ba\}$, we find from $\clubsuit 2$ that
    \begin{align*}
        \lambda^{(\Bmat)}_{R,\bx,\ba}
        =
        \sum_{\substack{\bb\in R^\A\\\bb_\bi=\ba_\bi}}\lambda^{(\Bmat)}_{R,\bx,\bb}
        =
        \lambda^{(\Bmat)}_{R_k,\bx_\bi,\ba_\bi}
    \end{align*}
    and, similarly, $\lambda^{(\Amat)}_{R,\bx,\ba}=        \lambda^{(\Amat)}_{R_k,\bx_\bi,\ba_\bi}$. As a consequence, if~\eqref{eqn_refinement_step_01032023} holds for $R_k$, it also holds for $R$.
    Therefore, it follows from Propositions~\ref{prop_SA_acceptance} and~\ref{prop_AIPk_acceptance} that $\BA^k(\X,\A)=\YES$ is equivalent to the existence of homomorphisms $\xi:\Xk\to\freeQ(\Ak)$ and $\zeta:\Xk\to\freeZ(\Ak)$ such that $\supp(\zeta(\bx))\subseteq\supp(\xi(\bx))$ for each $\bx\in X^k$. By virtue of Proposition~\ref{prop_decoupling_BLPAIP_general}, this happens precisely when $\Xk\to\freeBA(\Ak)$. Since $\Mblpaip=\BAminion$ (cf.~Example~\ref{example_minion_BA_semidirect}), this is equivalent to $\Xk\to\mathbb{F}_{\Mblpaip}(\Ak)$; i.e., to $\Test{\Mblpaip}{k}(\X,\A)=\YES$.
\end{proof}

\begin{rem}
\label{rem_for_SA_and_SoS_no_need_high_level}
The characterisations of $\SA^k$ and $\SoS^k$ in Propositions~\ref{prop_SA_acceptance} and~\ref{prop_SoS_acceptance} hold for any higher level than the first, unlike the characterisation of $\AIP^k$ in Proposition~\ref{prop_AIPk_acceptance}. This is due to the fact that $\Qconv$ and $\Sminion$ are conic minions, so Lemma~\ref{lem_ea_ast_Q_conic} applies, while $\Zaff$ is not. As for $\BW^k$, Proposition~\ref{prop_BW_acceptance} requires $k\geq\armax(\sigma)$ even if $\Hminion$ is a conic minion. The reason for this lies in the definition of the bounded-width hierarchy. Essentially, any constraint whose scope has more than $k$ distinct variables does not appear among the constraints of the partial homomorphisms witnessing acceptance of $\BW^k$, while it does appear in the requirements of $\SA^k$ and $\SoS^k$. Finally, assuming $k\geq\armax(\sigma)$ is also required in the characterisation of $\BA^k$ in Proposition~\ref{prop_BAk_acceptance}, in order to make use of Proposition~\ref{prop_decoupling_BLPAIP_general} and of the characterisation of $\AIP^k$.
\end{rem}

\begin{rem}
\label{rem:aip}
As it was shown in Section~\ref{sec_hierarchies_of_conic_minion_tests}, hierarchies of relaxations built on conic minions (such as $\BW^k$, $\SA^k$, $\SoS^k$, and $\BA^k$) are ``sound in the limit'', in that their $k$-th level correctly classifies instances $\X$ with $|X|\leq k$ (cf.~Proposition~\ref{prop_sherali_adams_exact}). This is not the case for the non-conic hierarchy $\AIP^k$, as it was established in the follow-up work~\cite{CZ25sicomp:approximate}.
In~\cite{Berkholz17:soda}, a stronger affine hierarchy was
  proposed, which -- contrary to $\AIP^k$ -- requires that the variables in the relaxation should be partial homomorphisms and is thus sound in the limit. By virtue of Proposition~\ref{lem_partial_homo}, this requirement can be captured by taking the semi-direct product of any conic minion and $\Zaff$. In particular, it follows that the hierarchy in~\cite{Berkholz17:soda} is not stronger than the hierarchy built on the minion  $\Hminion\ltimes\Zaff$ (cf.~Remark~\ref{rem_semidirect_product_different_semirings}).
In recent work~\cite{OConghaileC22:mfcs}, a different algorithm for
  $\parPCSPs$ has been proposed. The relationship of~\cite{OConghaileC22:mfcs} with
  our work is an interesting direction for future research.
\end{rem}

\section*{Acknowledgements}
The authors are grateful to Marcin Kozik, Jakub Opr\v{s}al, and Caterina Viola for useful discussions on $\SDP$ relaxations. We also thank the anonymous reviewers of this paper, and its extended abstract~\cite{cz23soda:minions}.

\appendix

\section{Notes on relaxations and hierarchies}
\label{appendix_notes_relaxations_hierarchies}
In this appendix, we discuss some basic properties and alternative formulations of the relaxations, and hierarchies thereof, presented in Section~\ref{sec_relaxations_hierarchies}.

\subsection{$\mbox{SA}^{\mbox{k}}$}
\label{subsec_appendix_SA_AIPk}

The hierarchy defining $\SA^k$ given by the system~\eqref{eqn_SA_and_AIPk} slightly differs from the one described in~\cite{Butti21:mfcs}. For completeness, we report below the hierarchy in~\cite{Butti21:mfcs} and show that it is equivalent to the one adopted in this work.

Given two $\sigma$-structures $\X,\A$, introduce a variable $\mu_V(f)$ for every subset $V\subseteq X$ with $1\leq |V|\leq k$ and every function $f:V\to A$, and a variable $\mu_{R,\bx}(f)$ for every $R\in\sigma$, every $\bx\in R^\X$, and every $f:\{\bx\}\to A$. The $k$-th level of the hierarchy defined in~\cite{Butti21:mfcs} is given by the following constraints:
\begin{align}
\label{eqn_defn_Sherali_Adams_BD}
\left.
\tag{$\varheart$}
\begin{array}{lll}
\mbox{($\varheart 1$)} & \displaystyle\sum_{f:V\to A}\mu_V(f)=1 & V\subseteq X \mbox{ s.t. }1\leq|V|\leq k\\
\mbox{($\varheart 2$)} & \displaystyle \mu_U(f)=\sum_{\substack{g:V\to A,\\ g|_{U}=f}}\mu_V(g) & 
U\subseteq V\subseteq X \mbox{ s.t. } 1\leq|V|\leq k, U\neq\emptyset, f:U\to A 
\\
\mbox{($\varheart 3$)} & \displaystyle \mu_U(f)=\sum_{\substack{g:\{\bx\}\to A,\\ g|_{U}=f}}\mu_{R,\bx}(g)
&
R\in\sigma, \bx\in R^\X, U\subseteq \{\bx\}
\mbox{ s.t. }1\leq|U|\leq k, f:U\to A 
\\
\mbox{($\varheart 4$)} & \displaystyle \mu_{R,\bx}(f)=0 & 
R\in\sigma, \bx\in R^\X, f:\{\bx\}\to A
\mbox{ s.t. } f(\bx)\not\in R^\A.
\end{array}
\right\}
\end{align}

\begin{lem}
\label{lem_our_SA_equals_BD_SA}
Let $k\in\N$, let $\X,\A$ be two $\sigma$-structures, and let $\tilde\X$
  (resp.\,$\tilde{\A}$) be the structure obtained from $\X$ (resp.\,$\A$) by
  adding the relation $R_k^{\tilde{\X}}=X^k$ (resp.\,$R_k^{\tilde{\A}}=A^k$). Then the system~\eqref{eqn_defn_Sherali_Adams_BD} applied to $\X$ and $\A$ is equivalent to the system~\eqref{eqn_SA_and_AIPk} applied to $\tilde\X$ and $\tilde\A$.
\end{lem}
\begin{proof}
Let $\lambda$ be a solution to~\eqref{eqn_SA_and_AIPk} applied to $\tilde\X$ and $\tilde\A$. Given $V\subseteq X$ with $1\leq |V|\leq k$ and $f:V\to A$, let $\bx\in X^k$ be such that $V=\{\bx\}$ and set $\mu_V(f)=\lambda_{R_k,\bx,f(\bx)}$. We claim that this assignment does not depend on the choice of $\bx$; i.e., we claim that $\lambda_{R_k,\bx,f(\bx)}=\lambda_{R_k,\by,f(\by)}$  whenever $\bx,\by\in X^k$ are such that $\{\bx\}=\{\by\}$. The latter condition implies that $\bx=\by_\bi$ and $\by=\bx_\bj$ for some $\bi,\bj\in [k]^k$. Using $\clubsuit 2$ and $\clubsuit 3$, we find
\begin{align*}
\lambda_{R_k,\by,f(\by)}
&=
\lambda_{R_k,\bx_\bj,f(\bx_\bj)}
=
\sum_{\substack{\ba\in A^k\\\ba_\bj=f(\bx_\bj)}}\lambda_{R_k,\bx,\ba}
=
\sum_{\substack{\ba\in A^k\\\ba_\bj=f(\bx_\bj)\\\bx\prec\ba}}\lambda_{R_k,\bx,\ba}
=
\lambda_{R_k,\bx,f(\bx)}+\sum_{\substack{\ba\in A^k\\\ba_\bj=f(\bx_\bj)\\\bx\prec\ba\\\ba\neq f(\bx)}}\lambda_{R_k,\bx,\ba}.
\end{align*}
The claim then follows if we show that there is no $\ba\in A^k$ such that $\ba_\bj=f(\bx_\bj)$, $\bx\prec\ba$, and $\ba\neq f(\bx)$. If such $\ba$ exists, using that $\bx=\bx_{\bj_\bi}$, we find that for some $p\in [k]$
\begin{align*}
a_p
&\neq
f(x_p)
=
f(x_{j_{i_p}})
=
a_{j_{i_p}}.
\end{align*}
Since $\bx\prec\ba$, this implies that $x_p\neq x_{j_{i_p}}$, a contradiction. Therefore, the claim is true. 
Additionally, given $R\in\sigma$, $\bx\in R^\X$, and $f:\{\bx\}\to A$, we set $\mu_{R,\bx}(f)=\lambda_{R,\bx,f(\bx)}$ if $f(\bx)\in R^\A$, $\mu_{R,\bx}(f)=0$ otherwise. 
It is straightforward to check that $\mu$ satisfies all constraints in the system~\eqref{eqn_defn_Sherali_Adams_BD} applied to $\X$ and $\A$.

Conversely, let $\mu$ be a solution to~\eqref{eqn_defn_Sherali_Adams_BD} applied to $\X$ and $\A$. As in the proof of Proposition~\ref{prop_BW_acceptance}, given two sets $S,T$, an integer $p\in\N$, and two tuples $\textbf{s}\in S^p, \textbf{t}\in T^p$ such that $\textbf{s}\prec\textbf{t}$, we define the map $f_{\textbf{s},\textbf{t}}:\{\textbf{s}\}\to T$ by $f_{\textbf{s},\textbf{t}}(s_\alpha)=t_\alpha$ for each $\alpha\in [p]$.
For every $R\in\sigma$, $\bx\in R^\X$, and $\ba\in R^\A$, we set $\lambda_{R,\bx,\ba}=\mu_{R,\bx}(f_{\bx,\ba})$ if $\bx\prec\ba$, $\lambda_{R,\bx,\ba}=0$ otherwise. Additionally, for every $\bx\in X^k=R_k^{\tilde{\X}}$ and $\ba\in A^k=R_k^{\tilde{\A}}$, we set $\lambda_{R_k,\bx,\ba}=\mu_{\{\bx\}}(f_{\bx,\ba})$ if $\bx\prec\ba$, $\lambda_{R_k,\bx,\ba}=0$ otherwise. It is easily verified that $\lambda$ yields a solution to~\eqref{eqn_SA_and_AIPk} applied to $\tilde{\X}$ and $\tilde{\A}$.
\end{proof}

We also note that~\cite{Atserias22:soda} has yet another definition of the Sherali--Adams hierarchy. However, it was shown in~\cite[Appendix~A]{Butti21:mfcs} that the hierarchy given in~\cite{Atserias22:soda} interleaves with the one in~\cite{Butti21:mfcs} and, by virtue of Lemma~\ref{lem_our_SA_equals_BD_SA}, with the hierarchy used in this work. In particular, the class of $\PCSP$s solved by constant levels of the hierarchy is the same for all definitions.

\subsection{SDP}
\label{subsec_appendix_SDP}
The relaxation defined by~\eqref{eqn_SDP_def} is not in semidefinite programming form, because of the constraint $\vardiamond4$. However, it can be easily translated into a semidefinite program by introducing $\gamma$ additional variables $\bmu_1,\dots,\bmu_\gamma$ taking values in $\R^\gamma$, and requiring that the following constraints are met:
\begin{align*}
\begin{array}{llllll}
\mbox{($\vardiamond4'$)}&\displaystyle\bmu_p\cdot\bmu_q=\delta_{p,q}& p,q\in [\gamma]\\
\mbox{($\vardiamond4''$)}&\displaystyle\sum_{\substack{\ba\in R^\A\\a_i=a}}\blambda_{R,\bx,\ba}\cdot\bmu_p=\blambda_{x_i,a}\cdot\bmu_p&R\in\sigma,\bx\in R^\X,a\in A, i\in [\ar(R)],p\in[\gamma]
\end{array}
\end{align*}
where $\delta_{p,q}$ is the Kronecker delta. One easily checks that the requirements $\vardiamond4'$ and $\vardiamond4''$ are together equivalent to the requirement $\vardiamond4$, and they are expressed in semidefinite programming form.
Next, we give a proof for the following basic result.
\begin{prop*}[Proposition~\ref{prop_trivial_stuff_SDP} restated]
Let $\X,\A$ be two $\sigma$-structures.
The system~\eqref{eqn_SDP_def} implies the following facts:
\begin{align*}
\begin{array}{lllllll}
(i)&\displaystyle\|\sum_{a\in A}\blambda_{x,a}\|^2=1 & x\in X;\\
(ii)&\displaystyle\sum_{\ba\in R^\A}\|\blambda_{R,\bx,\ba}\|^2=\|\sum_{\ba\in R^\A}\blambda_{R,\bx,\ba}\|^2=1 & R\in\sigma,\bx\in R^\X;\\
(iii)&\displaystyle\sum_{\substack{\ba\in R^\A\\a_i=a,\;a_j=a'}}\|\blambda_{R,\bx,\ba}\|^2=\blambda_{x_i,a}\cdot\blambda_{x_j,a'} & R\in\sigma,\bx\in R^\X,a,a'\in A,i,j\in[\ar(R)].
\intertext{If, in addition, $\X$ and $\A$ are $2$-enhanced,}
(iv)&\displaystyle\sum_{a\in A}\blambda_{x,a}=\sum_{a\in A}\blambda_{x',a} & x,x'\in X.
\end{array}
\end{align*}
\end{prop*}
\begin{proof}

\begin{itemize}
\item[$(i)$]
We have
\begin{align*}
\|\sum_{a\in A}\blambda_{x,a}\|^2
&=
\left(\sum_{a\in A}\blambda_{x,a}\right)\cdot\left(\sum_{a'\in A}\blambda_{x,a'}\right)
=
\sum_{a,a'\in A}\blambda_{x,a}\cdot\blambda_{x,a'}
=
\sum_{a\in A}\|\blambda_{x,a}\|^2
=
1,
\end{align*}
where the third equality comes from $\vardiamond2$ and the fourth from $\vardiamond1$.
\item[$(ii)$]
We have
\begin{align*}
\sum_{\ba\in R^\A}\|\lambda_{R,\bx,\ba}\|^2
&=
\sum_{\ba,\ba'\in R^\A}\lambda_{R,\bx,\ba}\cdot\lambda_{R,\bx,\ba'}
=
\left(\sum_{\ba\in R^\A}\blambda_{R,\bx,\ba}\right)\cdot\left(\sum_{\ba'\in R^\A}\blambda_{R,\bx,\ba'}\right)\\
&=
\|\sum_{\ba\in R^\A}\blambda_{R,\bx,\ba}\|^2
=
\|\sum_{a\in A}\sum_{\substack{\ba\in R^\A\\a_1=a}}\blambda_{R,\bx,\ba}\|^2
=
\|\sum_{a\in A}\blambda_{x_1,a}\|^2
=
1,
\end{align*}
where the first equality comes from $\vardiamond3$, the fifth from $\vardiamond4$, and the sixth from part $(i)$ of this proposition. 
\item[$(iii)$]
We have
\begin{multline*}
\blambda_{x_i,a}\cdot\blambda_{x_j,a'} =
\left(\sum_{\substack{\ba\in R^\A\\
a_i=a}}\blambda_{R,\bx,\ba}\right)\cdot\left(\sum_{\substack{\ba'\in R^\A\\
a'_j=a'}}\blambda_{R,\bx,\ba'}\right)
=
\sum_{\substack{\ba,\ba'\in R^\A\\
a_i=a,\;a'_j=a'}}\blambda_{R,\bx,\ba}\cdot\blambda_{R,\bx,\ba'}
\\=
\sum_{\substack{\ba\in R^\A\\
a_i=a,\;a_j=a'}}\|\blambda_{R,\bx,\ba}\|^2,
\end{multline*}
where the first equality comes from $\vardiamond4$ and the third from $\vardiamond3$.
\item[$(iv)$]
If $\X$ and $\A$ are $2$-enhanced, we have
\begin{align*}
\sum_{a\in A}\blambda_{x,a}
&=
\sum_{a\in A}\sum_{\substack{\ba\in R_2^\A\\a_1=a}}\blambda_{R_2,(x,x'),\ba}
=
\sum_{\ba\in R_2^\A}\blambda_{R_2,(x,x'),\ba}
=
\sum_{a\in A}\sum_{\substack{\ba\in R_2^\A\\a_2=a}}\blambda_{R_2,(x,x'),\ba}
=
\sum_{a\in A}\blambda_{x',a},
\end{align*}
where the first and fourth equalities come from $\vardiamond4$.\qedhere
\end{itemize}
\end{proof}

We point out that slightly different versions of the ``standard $\SDP$
relaxation'' appeared in the literature on $\CSP$s, some of which use parts
$(i)$ through $(iii)$ of Proposition~\ref{prop_trivial_stuff_SDP} as constraints
defining the relaxation. In particular, certain versions require that the scalar
products $\blambda_{x,a}\cdot\blambda_{y,b}$ should be nonnegative for all
choices of $x,y\in X$ and $a,b\in A$. For example, this is the case of the
$\SDP$ relaxation used in~\cite{Barto16:sicomp}. It follows from
Proposition~\ref{prop_Lasserre_Tulsiani}, proved in
Appendix~\ref{subsec_appendix_SoS}, that one can enforce nonnegativity of the
scalar products by taking the second level of the $\SoS$ hierarchy of the $\SDP$
relaxation as defined in this work.

\subsection{$\mbox{SoS}^{\mbox{k}}$}
\label{subsec_appendix_SoS}

First, we note that
the relaxation defined by~\eqref{eqn_Lasserre_def} can be easily translated into a semidefinite program through the procedure described at the beginning of Appendix~\ref{subsec_appendix_SDP}. Next, we show that the $2k$-th level of the SoS hierarchy enforces additional constraints -- in particular, nonnegativity of the scalar products of the SoS vectors -- on the vectors corresponding to the $k$-th level.

\begin{prop}
\label{prop_Lasserre_Tulsiani}
Let $k\in\N$, let $\X,\A$ be $2k$-enhanced $\sigma$-structures, suppose that $\SoS^{2k}(\X,\A)=\YES$, and let $\blambda$ denote a solution. Then $\blambda$ satisfies the following additional constraints:
\[
\begin{array}{lllll}
\mbox{(i)} & \displaystyle\blambda_{R_{2k},(\bx,\bx),(\ba,\ba)}\cdot\blambda_{R_{2k},(\by,\by),(\bb,\bb)}\geq 0 &\quad& 
\begin{array}{lll}
\bx,\by\in X^k,\ba,\bb\in A^k
\end{array}
\\[15pt]
\mbox{(ii)} & \displaystyle
\blambda_{R_{2k},(\bx,\bx),(\ba,\ba)}\cdot\blambda_{R_{2k},(\by,\by),(\bb,\bb)}=0 &\quad&  
\begin{array}{ll}
\bx,\by\in X^k,\ba,\bb\in A^k,\\ \ba_\bi\neq\bb_\bj\mbox{ for some }\bi,\bj\in [k]^k
\mbox{ such that }\bx_\bi=\by_\bj
\end{array}\\[15pt]
\mbox{(iii)}  &\displaystyle 
\begin{array}{llll}
\blambda_{R_{2k},(\bx,\bx),(\ba,\ba)}\cdot\blambda_{R_{2k},(\by,\by),(\bb,\bb)}=\\\blambda_{R_{2k},(\hat\bx,\hat\bx),(\hat\ba,\hat\ba)}\cdot\blambda_{R_{2k},(\hat\by,\hat\by),(\hat\bb,\hat\bb)}
\end{array}\quad& 
&\begin{array}{lll}
\bx,\hat\bx,\by,\hat\by\in X^k,\ba,\hat\ba,\bb,\hat\bb\in A^k,\\
(\hat\bx,\hat\by)_{\bell}=(\bx,\by),(\hat\ba,\hat\bb)_{\bell}=(\ba,\bb)\\ \mbox{for some } \bell\in[2k]^{2k}\mbox{ such that }|\{\bell\}|=2k.
\end{array}
\end{array}
\]
\end{prop}
\begin{proof}
Observe that, for $\bx,\by\in X^k$ and $\ba,\bb\in A^k$,
\begin{align}
\label{eqn_1749_1804_A}
\notag
\blambda_{R_{2k},(\bx,\bx),(\ba,\ba)}\cdot\blambda_{R_{2k},(\by,\by),(\bb,\bb)}
&=
\left(\sum_{\substack{\bc\in A^k}}\blambda_{R_{2k},(\bx,\by),(\ba,\bc)}\right)\cdot\left(\sum_{\substack{\bc'\in A^k}}\blambda_{R_{2k},(\bx,\by),(\bc',\bb)}\right)\\
\notag
&=
\sum_{\bc,\bc'\in A^k}\blambda_{R_{2k},(\bx,\by),(\ba,\bc)}\cdot\blambda_{R_{2k},(\bx,\by),(\bc',\bb)}\\
&=
\|\blambda_{R_{2k},(\bx,\by),(\ba,\bb)}\|^2,
\end{align} 
where the first and third equalities come from $\spadesuit 3$ and $\spadesuit 2$, respectively. Hence, $(i)$ holds. If, in addition, $\ba_\bi\neq\bb_\bj$ for some $\bi,\bj\in [k]^k$ such that $\bx_\bi=\by_\bj$, we deduce that $(\bx_\bi,\by_\bj)\not\prec(\ba_\bi,\bb_\bj)$ and, therefore,
\begin{align}
\label{eqn_1749_1804_B}
\notag
0
&=
\|\blambda_{R_{2k},(\bx_\bi,\by_\bj),(\ba_\bi,\bb_\bj)}\|^2
=
\|\sum_{\substack{(\bc,\bd)\in A^{2k}\\ \bc_\bi=\ba_\bi,\bd_\bj=\bb_\bj}}\blambda_{R_{2k},(\bx,\by),(\bc,\bd)}\|^2\\
&=
\sum_{\substack{(\bc,\bd)\in A^{2k}\\ \bc_\bi=\ba_\bi,\bd_\bj=\bb_\bj}}\|\blambda_{R_{2k},(\bx,\by),(\bc,\bd)}\|^2
\geq
\|\blambda_{R_{2k},(\bx,\by),(\ba,\bb)}\|^2,
\end{align}
where the first, second, and third equalities come from $\spadesuit 4$, $\spadesuit 3$, and $\spadesuit 2$, respectively. Combining~\eqref{eqn_1749_1804_A} and~\eqref{eqn_1749_1804_B}, we obtain $(ii)$. Suppose now that $\bx,\hat\bx,\by,\hat\by\in X^k$ and $\ba,\hat\ba,\bb,\hat\bb\in A^k$ are such that
$(\hat\bx,\hat\by)_{\bell}=(\bx,\by)$ and $(\hat\ba,\hat\bb)_{\bell}=(\ba,\bb)$ for some $\bell\in[2k]^{2k}$ such that $|\{\bell\}|=2k$. 
Using $\spadesuit 3$, we find
\begin{align*}
\blambda_{R_{2k},(\bx,\by),(\ba,\bb)}
&=
\blambda_{R_{2k},(\hat\bx,\hat\by)_{\bell},(\hat\ba,\hat\bb)_{\bell}}
=
\sum_{\substack{(\bc,\bd)\in A^{2k}\\ (\bc,\bd)_{\bell}=(\hat\ba,\hat\bb)_{\bell}}}\blambda_{R_{2k},(\hat\bx,\hat\by),(\bc,\bd)}
=
\blambda_{R_{2k},(\hat\bx,\hat\by),(\hat\ba,\hat\bb)},
\end{align*}
where the last equality is due to the fact that $|\{\bell\}|=2k$. Hence, $(iii)$ follows from~\eqref{eqn_1749_1804_A}.
\end{proof}

We observe that the relaxation in~\eqref{eqn_Lasserre_def} is formally different from the one described in~\cite{Tulsiani09:stoc}. However, it can be shown that the $2k$-th level of the hierarchy as defined here is at least as tight as the $k$-th level of the hierarchy as defined in~\cite{Tulsiani09:stoc}. Let us denote the two relaxations by $\SoS$ and $\SoS'$, respectively.
First of all, each variable in $\SoS'$ corresponds to a subset $S$ of $X$ and
an assignment $f:S\to A$, while in $\SoS$ the variables correspond to pairs of
tuples $\bx\in R^\X, \ba\in R^\A$. This is an inessential difference, as one can
check through the same argument used to prove
Lemma~\ref{lem_our_SA_equals_BD_SA} -- in particular, $\spadesuit 4$ ensures
that the only variables having nonzero weight are those corresponding to
well-defined assignments (cf.~Footnote~\ref{foot:rep}).
The $k$-th level of $\SoS'$ contains constraints that, in our language, are expressed as $\spadesuit 4$, $\spadesuit 1$, and parts $(i),(ii),(iii)$ of Proposition~\ref{prop_Lasserre_Tulsiani}. By virtue of Proposition~\ref{prop_Lasserre_Tulsiani}, therefore, any solution $\blambda$ to the $2k$-th level of $\SoS$ yields a solution to the $k$-th level of $\SoS'$ -- which means that the $2k$-th level of $\SoS$ is at least as tight as the $k$-th level of $\SoS'$.

{\small
\bibliographystyle{plainurl}
\bibliography{cz_local}
}

\end{document}